\documentclass[a4paper,twocolumn,10pt,accepted=2023-04-12]{quantumarticle}
\pdfoutput=1

\usepackage[utf8]{inputenc}
\usepackage[english]{babel}
\usepackage[T1]{fontenc}
\usepackage[numbers,sort&compress]{natbib}

\usepackage{float}
\usepackage{verbatim}
\usepackage{dsfont}
\usepackage[normalem]{ulem}

\usepackage{amsmath}
\usepackage{amssymb}
\usepackage{mathrsfs}
\usepackage{amsthm}
\usepackage{mathtools}
\usepackage{wasysym}
\usepackage{breqn}
\usepackage{comment}
\usepackage{bbm}
\usepackage{nicefrac}
\usepackage{soul}
\usepackage{stackrel}
\usepackage{braket}

\usepackage{multirow}
\usepackage[table]{xcolor}
\usepackage{array}
\newcolumntype{P}[1]{>{\centering\arraybackslash}p{#1}}

\newcommand{\de}[1]{\left( #1 \right)}

\newtheorem{thm}{Theorem}
\newtheorem{conjec}{Conjecture}

\usepackage{graphicx}
\usepackage{tikz}

\usepackage{hyperref}
\usepackage{framed}

\def\diracspacing{0.7pt}

\newcommand{\ketbra}[2]{| \hspace{\diracspacing} #1 \rangle \langle #2 \hspace{\diracspacing} |} 

\newcommand{\norm}[1]{\left\|#1\right\|}
\newcommand{\abs}[1]{\left|#1\right|}
\DeclareMathOperator{\Tr}{Tr}
\newcommand{\id}{\mathbbm{1}}
\newcommand{\BH}{\beta_{\mathrm{H}}}
\newcommand{\BM}{\beta_{\mathrm{M}}}
\newcommand{\im}{\mathbbm{i}}

\newtheorem{Lmm}{Lemma}

\usepackage{cases}

\begin{document}

\title{Boosting device-independent cryptography with tripartite nonlocality}

\author{Federico Grasselli}
\email{federico.grasselli@hhu.de}
\orcid{0000-0003-2966-7813}
\author{Gl\'{a}ucia Murta}
\author{Hermann Kampermann}
\author{Dagmar Bru\ss}
\affiliation{Institut f\"ur Theoretische Physik III, Heinrich-Heine-Universit\"at D\"usseldorf, Universit\"atsstraße 1, D-40225 D\"usseldorf, Germany}

\maketitle

\begin{abstract}
Device-independent (DI) protocols, such as DI conference key agreement (DICKA) and DI randomness expansion (DIRE), certify private randomness by observing nonlocal correlations when two or more parties test a Bell inequality. While most DI protocols are restricted to bipartite Bell tests, harnessing multipartite nonlocal correlations may lead to better performance. Here, we consider tripartite DICKA and DIRE protocols based on testing multipartite Bell inequalities, specifically: the Mermin-Ardehali-Belinskii-Klyshko (MABK) inequality, and the Holz and the Parity-CHSH inequalities introduced in the context of DICKA protocols. We evaluate the asymptotic performance of the DICKA (DIRE) protocols in terms of their conference key rate (net randomness generation rate), by deriving lower bounds on the conditional von Neumann entropy of one party's outcome and two parties' outcomes. For the Holz inequality, we prove a tight analytical lower bound on the one-outcome entropy and conjecture a tight lower bound on the two-outcome entropy. We additionally re-derive the analytical one-outcome entropy bound for the MABK inequality with a much simpler method and obtain a numerical lower bound on the two-outcome entropy for the Parity-CHSH inequality. Our simulations show that DICKA and DIRE protocols employing tripartite Bell inequalities can significantly outperform their bipartite counterparts. Moreover, we establish that genuine multipartite entanglement is not a precondition for multipartite DIRE while its necessity for DICKA remains an open question.
\end{abstract}

\section{Background}

The security of practical quantum cryptographic protocols \cite{pir-advances-qcrypto,review-pan} holds as far as the theoretical model of the quantum devices used in the protocol accurately describes their experimental implementation. Indeed, any small deviation from the ideal functionality of a quantum device can be exploited by an eavesdropper to breach the security of the protocol, as demonstrated by several quantum hacking attacks \cite{makarov-nature,makarov-nature2,review-pan}.

This leaves the user(s) with only two possibilities to ensure that their quantum cryptographic protocol is actually secure. They can either thoroughly characterize the devices being used, verifying that every assumption on the device is met in practice. However, this procedure might be challenging and beyond the capabilities of an end user. The other possibility is represented by device-independent (DI) cryptography, whose security is guaranteed independently of the inner workings of the employed devices \cite{YaoMayers,Acin2006,BarrentKent}.

The typical setting of a DI protocol is the Bell scenario \cite{Bellineq}. Here, two (or more) parties hold uncharacterized devices, treated as ``black boxes''. Each party can interact with their device by selecting an input, which prompts the device to return an output. In a quantum realization of the DI protocol, the device corresponds to a quantum system and the party interacts with it by choosing a measurement setting (input) and collecting the measurement outcome (output). By repeating this procedure a sufficient number of times and by revealing a fraction of their input-output pairs, the parties can characterize the probability distribution of the outputs, given the inputs.

To each Bell scenario can be associated a correlation inequality, called Bell inequality \cite{reviewBell}. If the distribution of the outputs observed by the parties violates a Bell inequality, the outputs are said to be nonlocally correlated. This occurs, e.g., if the parties perform appropriate measurements on their share of an entangled state in a loophole-free\footnote{Note that, since we assume quantum mechanics to hold, the Bell experiment does not need to close the locality loophole, as far as the parties' devices are isolated.} Bell experiment \cite{loopholefree1,loopholefree2}.

The intuition behind the security principle of DI protocols is that an eavesdropper cannot have full information on the parties' outputs if they are nonlocally correlated, regardless of the physical implementation of the devices. In fact, if the eavesdropper held a classical variable fully predicting the parties' outputs, that would represent a local explanation of the observed correlations \cite{AcinBrunner2007}. Therefore, the nonlocality of the outcomes, certified in a device-independent manner by the violation of a Bell inequality, can be used to infer randomness and secrecy with respect to an eavesdropper.
In particular, by observing nonlocality in the outcomes, DI quantum key distribution (DIQKD) \cite{AcinBrunner2007,PironioAcin2009,Masanes2011,VidickDIQKD,Arnon-Friedman2018} and its multiparty generalization, DI conference key agreement (DICKA) \cite{SG01,SG_pra_01,Holz2019DICKA,JeremyParityCHSH,CKA-review}, enable a set of parties to share a common secret key, while DI randomness expansion (DIRE) \cite{ColbeckThesis2006,Pironio2010,Colbeck2011,securityDIrandomness1,securityDIrandomness2,securityDIrandomness3,Woodhead2018} expands the initial share of private randomness of one or more parties.

In order to benchmark DI protocols based on different Bell scenarios and Bell inequalities, one needs to quantify the minimum amount of secret randomness in the parties' outcomes, for a given Bell violation. The figure of merit is the conditional von Neumann entropy of the parties’ outcomes, given the eavesdropper’s quantum side information, up to corrections due to finite-size effects. Indeed, the conditional von Neumann entropy of a set of outcomes determines the rate of secret bits generated by DIRE, DIQKD and DICKA protocols.

Although any lower bound on this quantity is also a valid measure of secret randomness, tighter bounds imply higher secret bit generation rates and hence more efficient and robust DI protocols. This is particularly important since today's quantum technology is mature enough to enable the experimental implementation of DI protocols, as testified by recent DIRE \cite{DIREexp1,DIREexp2} and DIQKD experiments \cite{DIQKDexp1,DIQKDexp2,DIQKDexp3}.

The derivation of tight bounds on the conditional von Neumann entropies relevant for the security of DI protocols has been a major theoretical challenge in the field of DI cryptography. In the bipartite DI scenario, tight analytical bounds on one-outcome entropies \cite{PironioAcin2009} were derived for the CHSH inequality \cite{CHSH} and its variants \cite{AsymCHSH-Woodhead,masini-entropy-bounds,Sekatski-entropy-bounds}. Recently, two-outcome entropy bounds were investigated in~\cite{Colbeck2021} for the CHSH inequality. In parallel, reliable numerical lower bounds on the conditional von Neumann entropy can be obtained with the techniques developed in \cite{Brown-numerical-vNentropy-bounds,Tan-numerical-vNentropy-bounds,Renner-SDP}.

\section{Summary of results} \label{sec:summary-results}

In this work we consider a tripartite DI scenario where three unknown quantum systems are individually measured by Alice, Bob and Charlie, respectively. Every party can perform one of two measurements, labelled by inputs $0$ and $1$, each of which yields a binary outcome, either $0$ or $1$. The two measurements are: $A_0$ and $A_1$ for Alice, $B_0$ and $B_1$ for Bob and $C_0$ and $C_1$ for Charlie. In this scenario, the parties can either test a tripartite Bell inequality or a bipartite Bell inequality, in which case one of the parties remains idle.

For brevity of notation, we define: $B_{\pm}:=(B_0 \pm B_1)/2$ and $C_{\pm}:=(C_0 \pm C_1)/2$, while $\overset{L}{\leq}$ ($\overset{Q}{\leq}$) indicates the local (quantum) bound. Moreover, $\left<A_xB_yC_z\right>$ represents the correlation function:
\begin{align}
   \sum_{a,b,c} (-1)^{a+ b+ c} \Pr[A_x=a,B_y=b,C_z=c],
\end{align}
and similarly for the two-party correlators. The Bell inequalities considered in this work are the following:

\begin{itemize}
    \item The tripartite \emph{Holz inequality} \cite{Holz2019DICKA},
    \begin{equation}
    \begin{split}
     \beta_{\mathrm{H}} =& \left<A_1B_{+}C_{+}\right>-\left<A_0B_{-}\right>\\& -\left<A_0C_{-}\right> -\left<B_{-}C_{-}\right> \overset{L}{\leq} 1 \overset{Q}{\leq} 3/2.
     \end{split}\label{timo-ineq}
\end{equation}

    \item The tripartite \emph{Parity-CHSH inequality} \cite{JeremyParityCHSH},
    \begin{equation}
        \beta_{\mathrm{pC}} = \left<A_1B_{-}C_0\right>+\left<A_0B_{+}\right> \overset{L}{\leq} 1 \overset{Q}{\leq} \sqrt{2}. \label{parity-chsh-ineq}
    \end{equation}

    \item The tripartite \emph{Mermin-Ardehali-Belinskii-Klyshko (MABK) inequality} \cite{Mermin,Ardehali,BK93},
    \begin{align}
    \begin{split}
    \BM=&\braket{A_0 B_0 C_1} + \braket{A_0 B_1 C_0} \\
    & +\braket{A_1 B_0 C_0} 
    - \braket{A_1 B_1 C_1} \overset{L}{\leq} 2 \overset{Q}{\leq} 4.\label{mabk-ineq}
    \end{split}
    \end{align}
    \item The bipartite family of \emph{asymmetric Clauser-Horne-Shimony-Holt (CHSH) inequalities} \cite{AsymCHSH,AsymCHSH-Woodhead}, parametrized by $\alpha\in\mathbbm{R}$,
        \begin{align}
        \begin{split}
            \beta_{\alpha\mathrm{C}}&= 2\alpha\braket{A_0 B_+} + 2\braket{A_1 B_-} \\ 
            &\overset{L}{\leq} \left\lbrace \begin{array}{ll}
            2 \abs{\alpha}     &\mbox{if } \abs{\alpha}>1 \\
            2   & \mbox{if } \abs{\alpha}\leq 1
        \end{array}\right. 
        \\ &\overset{Q}{\leq} 2 \sqrt{1+\alpha^2}.
        \end{split} \label{alphachsh-ineq}
        \end{align}
\end{itemize}
The goal of our work is to benchmark the performance of DICKA and DIRE protocols based on the above Bell inequalities and determine which Bell inequality is optimal for each cryptographic task.

The crucial ingredient for our comparison is the derivation of tight analytical and numerical lower bounds on one-outcome conditional entropies, $H(A_0|E)$, and two-outcome conditional entropies, $H(A_0 B_0|E)$, as a function of the violation of the considered Bell inequality. Indeed, the entropy $H(A_0|E)$ determines the conference key rate of DICKA protocols, while the two-outcome entropy $H(A_0 B_0|E)$ determines the net randomness generation rate of our DIRE protocols.

In order to have a fair comparison, we provide the parties with an equivalent entanglement resource in each Bell scenario, which is chosen to be a noisy version of the entangled state which maximally violates each of the Bell inequalities. The parties then perform the measurements that would lead, in the absence of noise, to maximal Bell violation. In particular, when the parties test the Holz, Parity-CHSH and MABK inequality, they share a locally-depolarized GHZ state:
\begin{equation}
    \rho^{(3)}= \mathcal{D}^{\otimes 3}(\ketbra{\mathrm{GHZ}}{\mathrm{GHZ}}), \label{depolGHZ}
\end{equation}
where the map $\mathcal{D}$ acts on every qubit as follows:
\begin{align}
    \mathcal{D}(\sigma) = p \sigma +\frac{1-p}{2} \id, \label{local-depol}
\end{align}
and where $\ket{\mathrm{GHZ}}=(\ket{000} + \ket{111})/\sqrt{2}$ is the GHZ state. Conversely, when the parties test the (bipartite) asymmetric CHSH inequalities, they share the bipartite version of the GHZ state, namely the Bell state $\ket{\Phi^+}=(\ket{00}+\ket{11})/\sqrt{2}$, also subjected to local depolarization:
\begin{equation}
    \rho^{(2)} = \mathcal{D}^{\otimes 2} (\ketbra{\Phi^+}{\Phi^+}). \label{depolBell+}
\end{equation}
The noise parameter, $p$, is linked to the probability that each qubit is depolarized\footnote{The effect of photon loss would be modelled similarly to local depolarization \eqref{local-depol}, as: $\mathcal{L}(\sigma) = p \sigma +(1-p) \ketbra{vac}{vac}$, where $\ket{vac}$ is the vacuum. Since the detection loophole forbids discarding no-detection events, assigning a random measurement outcome when a photon is lost would have the same effect of local depolarization \eqref{local-depol}. Hence, in our simulations $1-p$ can also be seen as the probability that a photon is lost.}, given by $1-p$. We also study the case in which the ideal GHZ and Bell states are globally depolarized:
\begin{align}
    \rho^{(3)}= p \ketbra{\mathrm{GHZ}}{\mathrm{GHZ}} +(1-p) \frac{\id}{8}, \label{global-depol-GHZ}
\end{align}
for three parties testing a tripartite Bell inequality and
\begin{align}
    \rho^{(2)}= p \ketbra{\mathrm{\Phi^+}}{\mathrm{\Phi^+}} +(1-p) \frac{\id}{4}, \label{global-depol-Bell}
\end{align}
for two parties testing a bipartite Bell inequality. In this case, $1-p$ is the probability that the three-qubit (two-qubit) state is depolarized. Further details on the optimal measurement settings of each inequality are given in Appendix~\ref{app:bounds-and-optimal-measurements}.

We compare the performance of DICKA protocols based on the inequalities \eqref{timo-ineq}, \eqref{parity-chsh-ineq} and \eqref{alphachsh-ineq}, by computing their asymptotic conference key rates for the two noise models outlined above. Similarly, we compare DIRE protocols based on each of the four Bell inequalities (for the bipartite Bell inequality we set $\alpha=1$, which recovers the CHSH inequality) in terms of their asymptotic net randomness generation rate, when the randomness is extracted from the outcomes of two parties. 

For both DICKA and DIRE protocols, we observe that tripartite Bell inequalities can provide a performance advantage over the family of bipartite Bell inequalities in \eqref{alphachsh-ineq}, which are currently regarded as being optimal for DI tasks such as DIQKD \cite{AsymCHSH-Woodhead,Sekatski-entropy-bounds}.

The performance comparisons are enabled by bounds on the conditional von Neumann entropy. The derivation of conditional entropy bounds in the multipartite scenario was first addressed in \cite{JeremyMABK} and then more thoroughly in \cite{Grasselli-PRXQuantum}, where one-outcome and two-outcome entropy bounds were derived for the MABK inequality, a full-correlator Bell inequality \cite{Mermin,Ardehali,BK93}.

In this work, we take a significant step further and provide tight analytical entropy bounds as a function of the violation of the Holz inequality. More precisely, we derive a tight analytical bound for the one-outcome entropy $H(A_0|E)$ (see Theorem~\ref{thm:Hbound}) and provide an analytical conjecture of the tight bound for the two-outcome entropy $H(A_0B_0|E)$  (Conjecture~\ref{conj:H2bound}), which is robustly confirmed by numerical data. To the best of our knowledge, our bound on $H(A_0|E)$ is the first tight analytical bound derived for a non-full-correlator Bell inequality, like the Holz inequality. And our conjectured bound on $H(A_0B_0|E)$ is the first multi-outcome analytical bound for a non-full-correlator Bell inequality.

The derivation of the analytical bound on $H(A_0|E)$ for the tripartite Holz inequality builds on an entropic uncertainty relation, similarly to the approach used in \cite{AsymCHSH-Woodhead,masini-entropy-bounds} for the CHSH inequality. However, the increased number of parties and the asymmetry with respect to permutations of parties makes our derivation highly non-trivial. We report the full proof of the bound and of its tightness in Appendix~\ref{app:Holz-one-outcome-proof}. By following the same approach, in Appendix~\ref{app:MABK-one-outcome-proof} we rederive the analytical bound on $H(A_0|E)$ for the MABK inequality with a proof that is considerably simpler than the derivation in \cite{Grasselli-PRXQuantum}. Besides, in Appendix~\ref{app:tightness-parityCHSH-bound} we prove that the analytical lower bound on $H(A_0|E)$ for the Parity-CHSH inequality, originally derived in \cite{JeremyParityCHSH}, is actually tight.

We additionally compute numerical bounds on the two-outcome entropy $H(A_0B_0|E)$ for the Parity-CHSH and CHSH inequalities, which are used to compute the corresponding DIRE rates. A detailed calculation of the bounds is provided in Appendix~\ref{app:two-outcome-numerical}. The numerical bound for the Parity-CHSH inequality is a new result, while the one for the CHSH inequality has been independently derived in \cite{Colbeck2021}.

The remainder of the paper is structured as follows. In Sec.~\ref{sec:entropy-bounds} we present our analytical and numerical entropy bounds. In Sec.~\ref{sec:DICKA} we apply our bounds to DICKA protocols and compare their performance to deduce which Bell inequality is optimal. In Sec.~\ref{sec:DIRE} we perform an analogous comparison for DIRE protocols. We discuss our results and conclude in Sec.~\ref{sec:discussion}, where Table~\ref{tab:summary} provides an overview of all the considered entropy bounds. The analytical and numerical calculations of the entropy bounds are presented in Appendices~\ref{app:Holz-one-outcome-proof} to \ref{app:tightness-parityCHSH-bound}, while in Appendix~\ref{app:bounds-and-optimal-measurements} we summarize the Bell inequalities and their entropy bounds.

\section{One-outcome and two-outcome entropy bounds} \label{sec:entropy-bounds}

In this section, we present our analytical bounds on the conditional von Neumann entropy when the parties test the Holz inequality. Additionally, we compare the analytical and numerical bounds derived in this work with other bounds, when the parties test the inequalities \eqref{timo-ineq}-\eqref{alphachsh-ineq}. The bounds are then used to compute DICKA and DIRE rates, respectively, in Sec.~\ref{sec:DICKA} and \ref{sec:DIRE}.

\subsection{Single party's outcome} \label{sec:one-outcome-entropy}

We obtain a tight analytical lower bound on the conditional entropy of Alice's outcome $A_0$, when Alice, Bob and Charlie test the Holz inequality \eqref{timo-ineq}.

\begin{thm}\label{thm:Hbound}
Let Alice, Bob and Charlie test the Holz inequality \cite{Holz2019DICKA} and let $\beta_{\mathrm{H}}$ be the expected Bell value. Then, the von Neumann entropy of Alice's outcome $A_0$ conditioned on Eve's information $E$ satisfies
\begin{equation}
    H(A_0|E) \geq  1-h\left[\frac{1}{4}\left( \beta_{\mathrm{H}} + 1+ \sqrt{\beta_{\mathrm{H}}^2 + 2\beta_{\mathrm{H}} -3}\right)\right] \label{Hbound},
\end{equation}
where $h(x)=-x \log_2 x + (1-x) \log_2 (1-x)$ is the binary entropy. Moreover, the bound is tight. That is, for every Bell value $\beta_{\mathrm{H}}$ there exists a quantum strategy (state and measurements) which attains that Bell value and whose conditional entropy is given by the r.h.s. of \eqref{Hbound}.
\end{thm}

Here we provide a sketch of the proof of Theorem~\ref{thm:Hbound}, a detailed proof is presented in Appendix~\ref{app:Holz-one-outcome-proof}.

\begin{proof}[Proof sketch]
First, we employ Jordan's Lemma \cite{Masanes06,Grasselli-PRXQuantum} to simplify the problem at hand, without loss of generality. In particular, we show that we can focus on deriving a convex lower bound on $H(A_0|E)$ when Alice, Bob and Charlie share a three-qubit state and perform rank-one binary projective measurements on their respective qubits.

We identify the plane induced by the qubit observables of each party to be the $(x,z)$ plane of the Bloch sphere. Then, we choose the local reference frames such that the Bell value \eqref{timo-ineq} is simplified and Alice's measurement $A_0$ corresponds to the Pauli measurement $\sigma_z$: $A_0=Z$. With these choices, we show that the three-qubit state $\rho_{ABC}$ shared by the parties can be assumed to be block-diagonal in the GHZ basis, without loss of generality. We are thus left to derive a lower bound on $H(Z|E)$.

The next step is to employ the uncertainty relation for von Neumann entropies \cite{uncertrel2010} in combination with other properties of the conditional von Neumann entropy to obtain the lower bound:
\begin{align}
    H(Z|E) \geq 1- h\left(\frac{1+\abs{\braket{XXX}}}{2}\right), \label{sketch1}
\end{align}
where $\braket{XXX}$ is the expectation value of a $\sigma_x$ measurement performed by all parties. Note that a similar step to the one above is employed in \cite{AsymCHSH-Woodhead} to derive an entropy bound when two parties test the asymmetric CHSH inequality.

The last decisive step of our proof lies in the ability to link the expectation value $\braket{XXX}$ to the Bell value $\beta_H$. We show that they can be related by the following non-linear inequality:
\begin{equation}
    \abs{\braket{XXX}} \geq \frac{\beta_{\mathrm{H}}}{2} - \frac{1}{2} + \frac{1}{2}\sqrt{\beta_{\mathrm{H}}^2 + 2\beta_{\mathrm{H}} -3} \label{sketch2}.
\end{equation}
By combining \eqref{sketch2} with \eqref{sketch1} and with the fact that $h(1/2 + x) $ is monotonically decreasing for $x>0$, we obtain the result in \eqref{Hbound}.
\end{proof}

The Holz inequality was introduced in \cite{Holz2019DICKA} as a multipartite generalization of the CHSH inequality (i.e., all the parties have two inputs and two outputs) and its construction was tailored for DICKA protocols. In Sec.~\ref{sec:DICKA}, we employ the tight entropy bound we derived in Theorem~\ref{thm:Hbound} to show that the Holz inequality indeed leads to DICKA protocols with the currently best-known performance.

Besides, the technique used to derive the entropy bound of Theorem~\ref{thm:Hbound} can  constitute a simpler alternative to the approach used in \cite{Grasselli-PRXQuantum}, as demonstrated by our re-derivation of the single-outcome entropy bound for the MABK inequality from \cite{Grasselli-PRXQuantum}. We provide the details of the derivation in Appendix~\ref{sec:MABK-bound}.

\begin{figure}[ht] 
	\centering
		\textbf{Local depolarization}\par\medskip
		\includegraphics[width=0.9\linewidth,keepaspectratio]{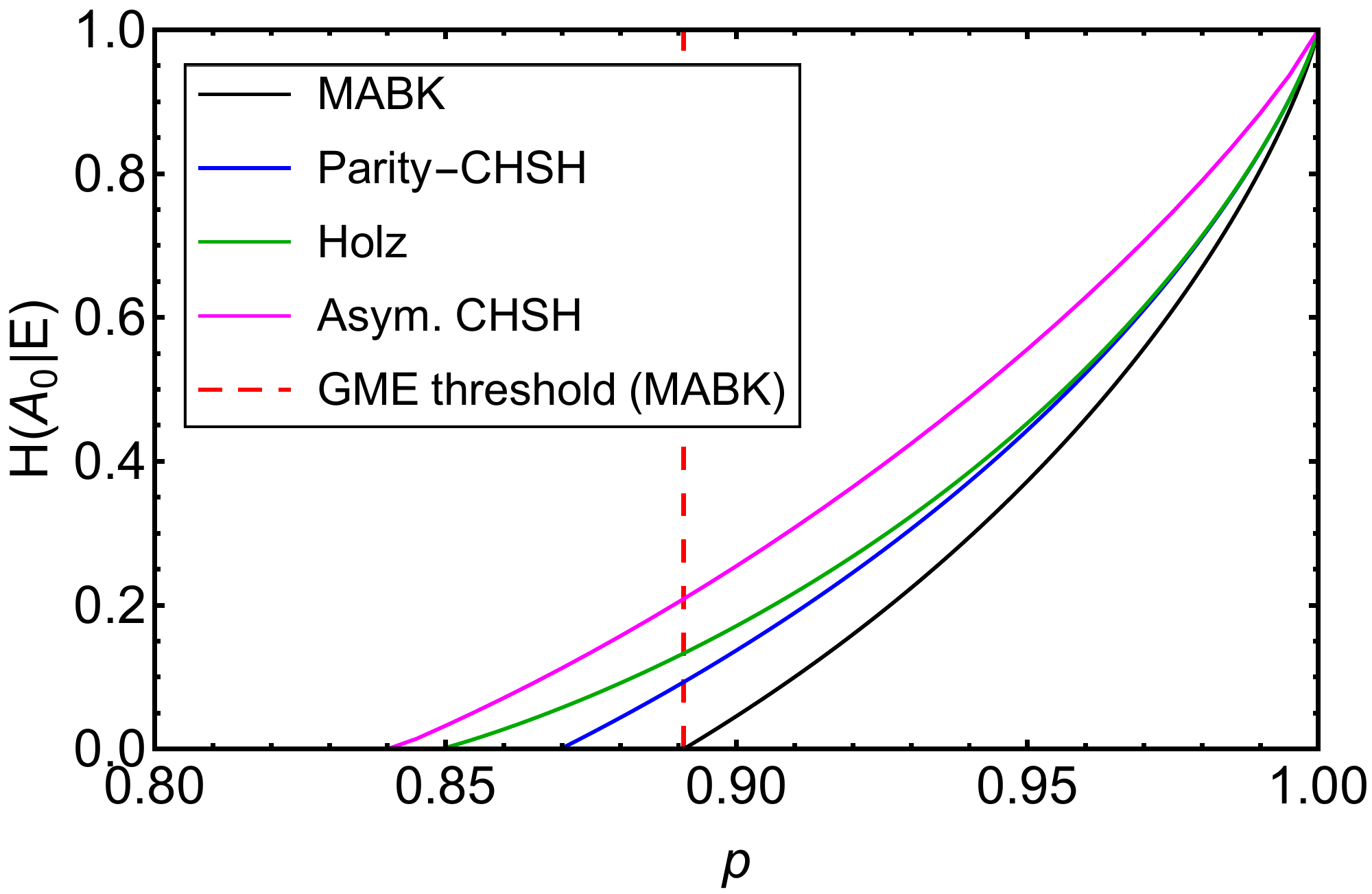}
	    \vspace{1ex}
		\textbf{Global depolarization}\par\medskip
		\includegraphics[width=0.9\linewidth,keepaspectratio]{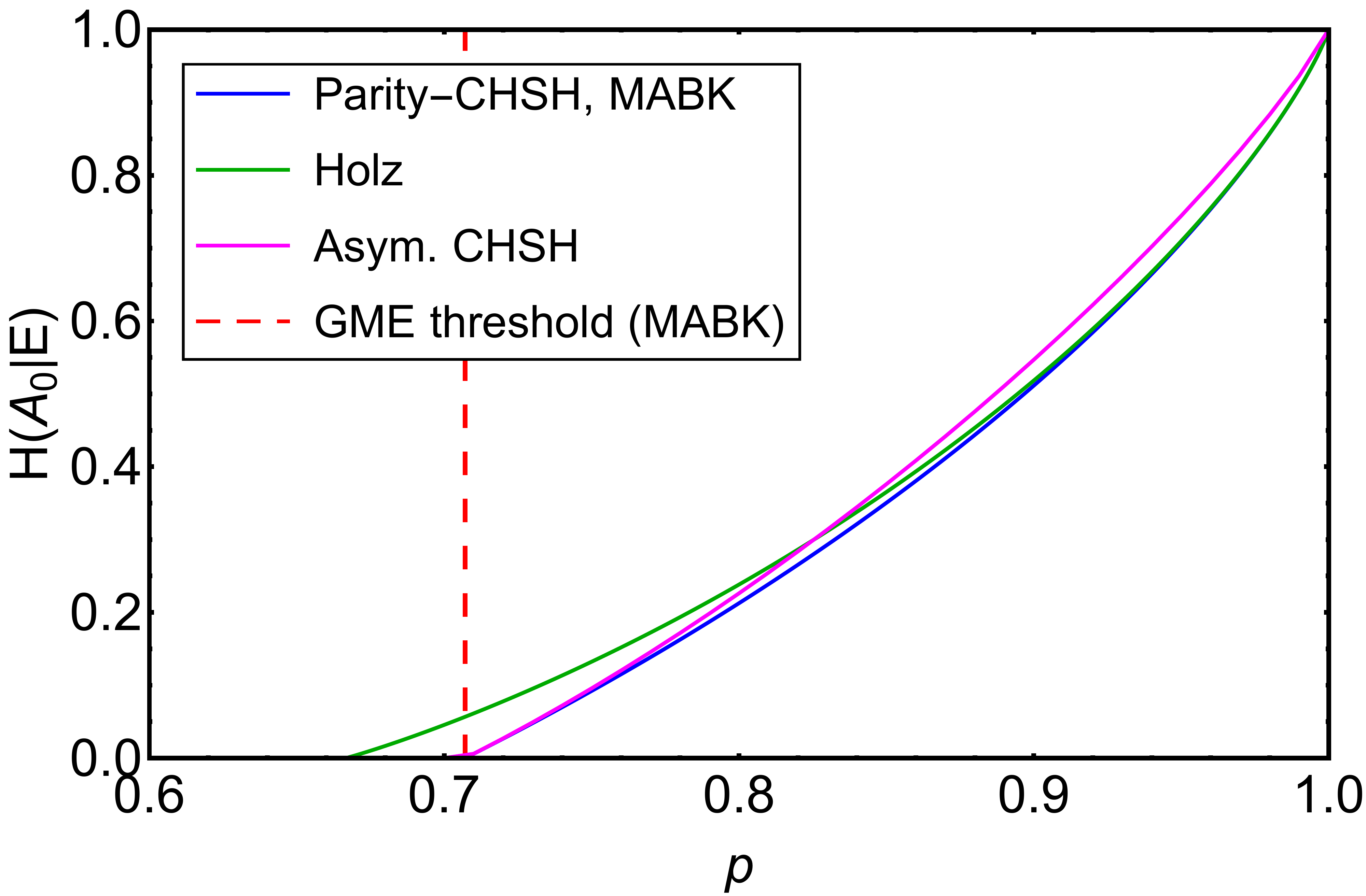}
	\caption{Analytical lower bounds on the conditional entropy $H(A_0|E)$ of Alice's outcome $A_0$ for various Bell inequalities, when three (two) parties are given a GHZ (Bell) state that has been locally and globally depolarized with probability $1-p$. The bound for the Holz inequality is derived in Theorem~\ref{thm:Hbound}, while the bounds for the Parity-CHSH, the asymmetric CHSH, and the MABK inequality are taken from \cite{JeremyParityCHSH}, \cite{AsymCHSH-Woodhead}, and \cite{Grasselli-PRXQuantum} (and re-derived in Appendix~\ref{app:MABK-one-outcome-proof}), respectively. All bounds are reported in Appendix~\ref{app:bounds-and-optimal-measurements}.}
	\label{fig:Aentropy-comparison}
\end{figure}

In Fig.~\ref{fig:Aentropy-comparison} we compare the lower bound on $H(A_0|E)$ derived in Theorem~\ref{thm:Hbound} for the Holz inequality with analogous bounds for the Parity-CHSH and asymmetric CHSH inequality from Refs.~\cite{JeremyParityCHSH,AsymCHSH-Woodhead}, and the bound for the MABK inequality from~\cite{Grasselli-PRXQuantum} re-derived in Appendix~\ref{app:MABK-one-outcome-proof}.
Note that all the bounds are tight (we prove the tightness of the Parity-CHSH bound in Appendix~\ref{app:tightness-parityCHSH-bound}), except for the case of the MABK inequality, and that we maximize the bound for the asymmetric CHSH inequality \eqref{alphachsh-ineq} over the parameter $\alpha$. We plot the entropy bounds as a function of the depolarization parameter $p$, where $1-p$ is the probability of local or global depolarization (see Sec.~\ref{sec:summary-results}). From the plot with local depolarization we observe that,  for a fixed value of $p$, the largest entropy is certified by the bipartite asymmetric CHSH inequality, while the Holz inequality provides the largest bound among the tripartite inequalities. This is expected, since for local noise the Bell violation is proportional to $\sim p^N$, where $N$ is the number of parties testing the inequality. Hence, the violation decreases for increasing number of parties and fixed noise and so does the entropy bound. This fact does not necessarily hold with other noise models, e.g. global depolarization, where the Holz inequality leads to the largest entropy at high noise levels (low $p$).

Interestingly, the entropy bounds for the Holz and the Parity-CHSH inequality in Fig.~\ref{fig:Aentropy-comparison} are non-zero below the genuine multipartite entanglement (GME) threshold of the MABK inequality. This suggests that GME might not be necessary to certify the privacy of a single party's outcome when testing multipartite Bell inequalities. Indeed, this is the case for asymmetric Bell inequalities like the Holz and the Parity-CHSH inequality studied here, while a previous study \cite{Grasselli-PRXQuantum} on the permutationally-invariant MABK inequality showed that GME is necessary to extract private randomness from a party's outcome.

It is straightforward to find a non-GME state that leads to non-zero entropy in the case of the Parity-CHSH inequality, where Charlie's role is trivial. Indeed, depending on the outcome of Charlie's only measurement $C_0$, Alice and Bob are effectively testing two distinct CHSH inequalities, one for outcome $C_0=0$ and another for outcome $C_0=1$. Therefore, one could easily obtain the maximal violation of such inequality by distributing the biseparable state: $\ket{\Phi^+}\otimes \ket{0}$ and selecting the optimal CHSH measurements for Alice and Bob and setting Charlie's measurement to be $Z$. Because the Parity-CHSH inequality can be seen as a special case of the Holz inequality where Charlie's measurements coincide ($C_0=C_1$), we can obtain non-zero entropy with non-GME states also for the Holz inequality by choosing the same setup as above, up to some change of sign.

\subsection{Two parties' outcomes}

In this section we compare lower bounds on the joint conditional entropy $H(A_0 B_0|E)$ of Alice's outcome $A_0$ and Bob's outcome $B_0$ when the parties test the inequalities \eqref{timo-ineq}-\eqref{alphachsh-ineq} (we set $\alpha=1$ in the asymmetric CHSH inequality, which reduces it to the standard CHSH inequality), in view of their application for DIRE.

\begin{figure}[ht] 
	\centering
		\textbf{Local depolarization}\par\medskip
		\includegraphics[width=0.9\linewidth,keepaspectratio]{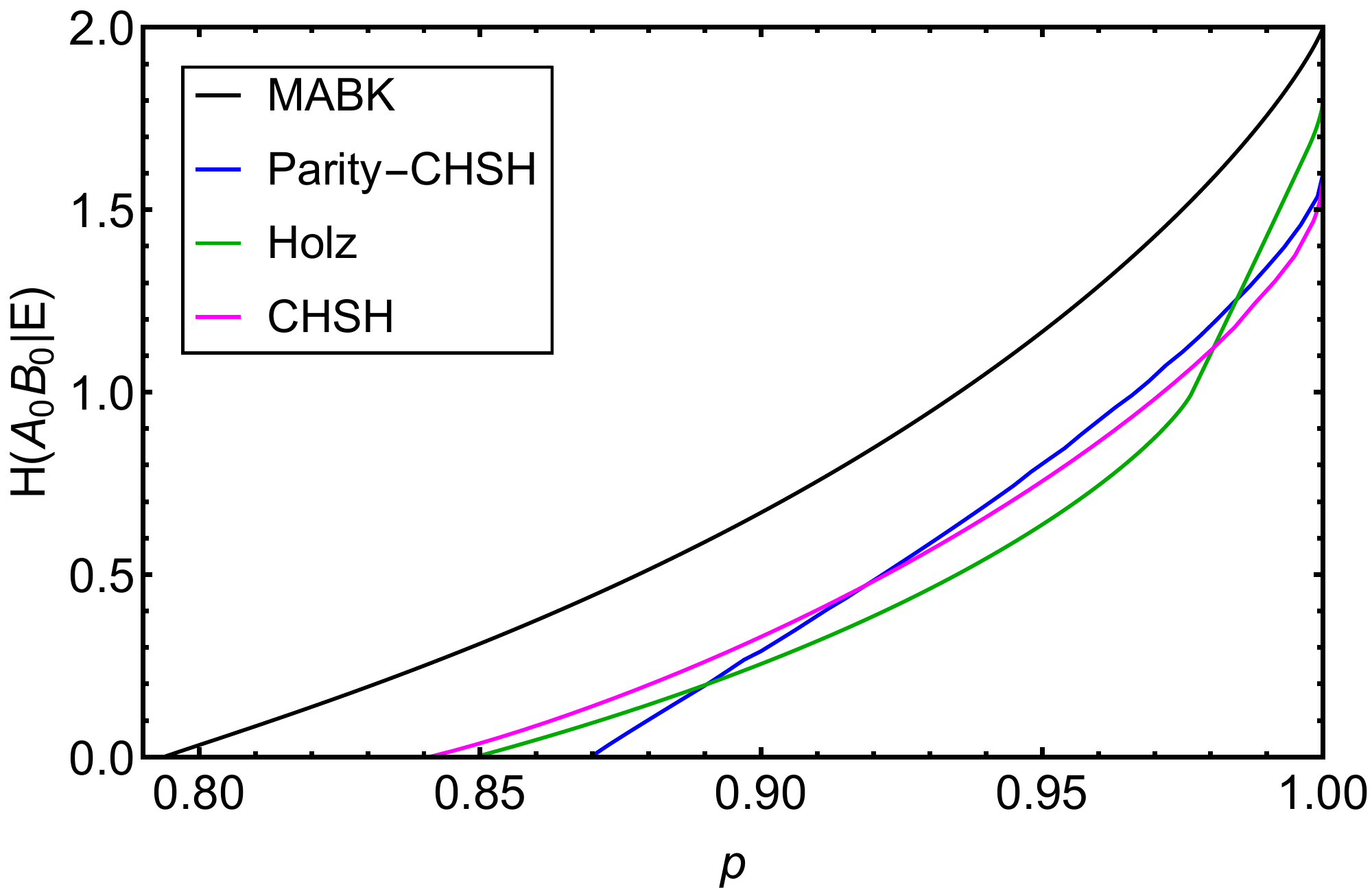}
	    \vspace{1ex}
		\textbf{Global depolarization}\par\medskip
		\includegraphics[width=0.9\linewidth,keepaspectratio]{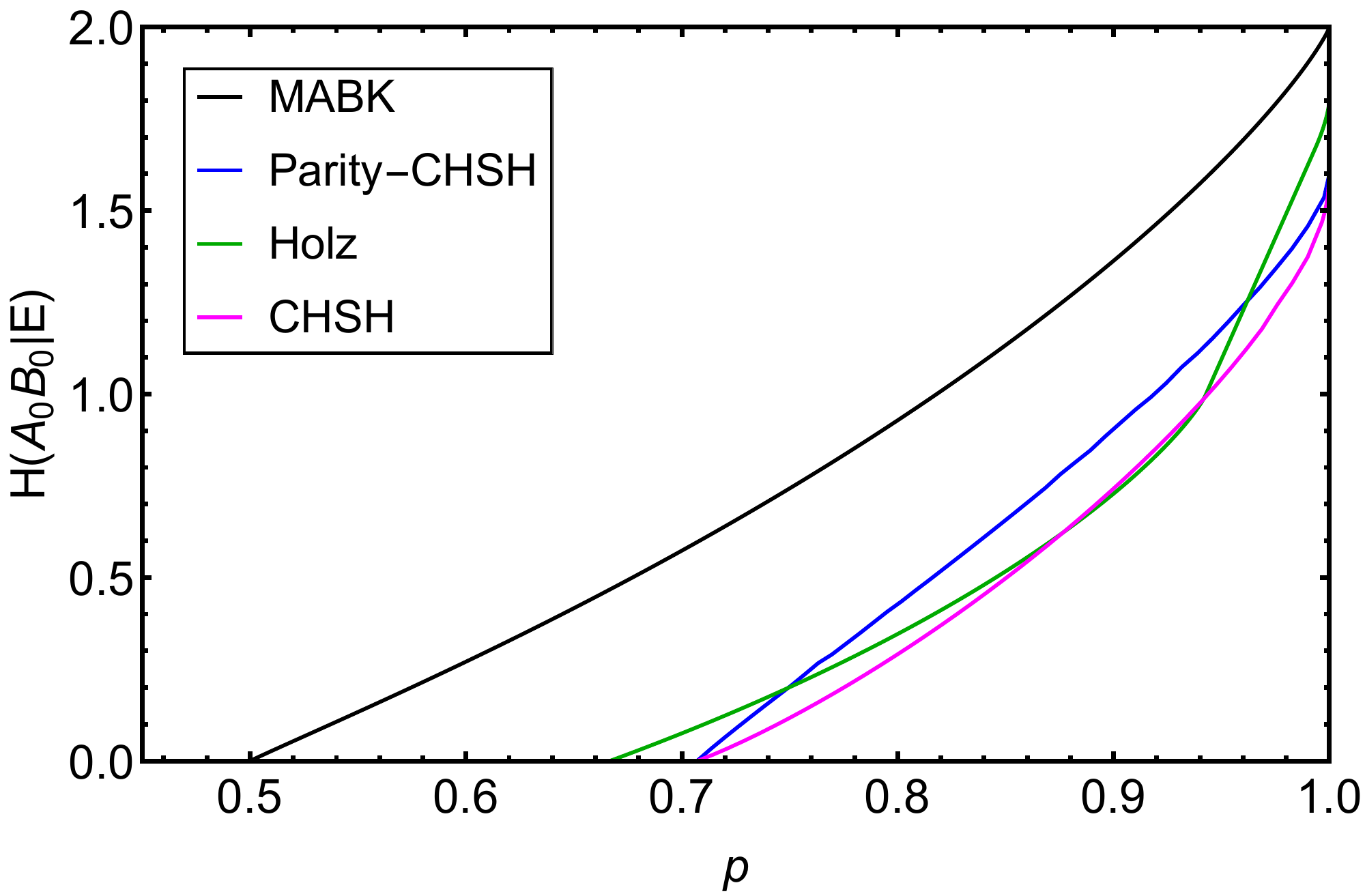}
	\caption{Lower bounds on the conditional entropy $H(A_0 B_0|E)$ of Alice's and Bob's outcomes for various Bell inequalities, when three (two) parties are given a GHZ (Bell) state that has been locally and globally depolarized with probability $1-p$. The bound for the MABK inequality is analytical and was obtained in \cite{Grasselli-PRXQuantum}, the bound for the Holz inequality is a conjectured tight analytical expression given by \eqref{ABentropybound-holz}, while the bounds for the Parity-CHSH and CHSH inequality are numerical (see Appendix~\ref{app:two-outcome-numerical}).}
	\label{fig:ABentropy-comparison}
\end{figure}

The entropy bound when three parties test the MABK inequality \eqref{mabk-ineq} is analytical and was derived in \cite{Grasselli-PRXQuantum}; we report it in Appendix~\ref{app:bounds-and-optimal-measurements}. Conversely, the entropy bounds for the Parity-CHSH, CHSH and Holz inequality are novel numerical bounds obtained by directly minimizing the entropy over all states and measurements yielding a given Bell violation. In order to achieve this, we significantly simplify the optimization problem for each inequality (details in Appendix~\ref{app:two-outcome-numerical}) before carrying out the numerical computation.

In particular, the numerical bound on $H(A_0 B_0 |E)$ for the Holz inequality relies on the intermediate results used to derive Theorem~\ref{thm:Hbound}, which allow us to simplify the optimization problem by reducing the number of variables. As a result, we optimize over just one measurement direction, i.e. one angle, and over block-diagonal states. This simplification also allowed us to conjecture the form of the corresponding analytical bound \eqref{ABentropybound-holz}.

\begin{conjec}\label{conj:H2bound}
Let Alice, Bob and Charlie test the Holz inequality \cite{Holz2019DICKA} and let $\beta_{\mathrm{H}}$ be the expected Bell value. Then, the joint von Neumann entropy of Alice's outcome $A_0$ and Bob's outcome $B_0$, conditioned on Eve's information $E$, satisfies
\begin{align}
&H(A_0 B_0 |E) \geq \nonumber\\
&\left\lbrace\begin{array}{ll}
    \eta(\beta_{\mathrm{H}})     &\beta_{\mathrm{H}}\in[1,\sqrt{2}]  \\[2ex]
     \displaystyle\frac{\theta(\beta^*_{\mathrm{H}},x(\beta^*_{\mathrm{H}}))-1}{\beta^*_{\mathrm{H}}-\sqrt{2}}(\beta_{\mathrm{H}}-\sqrt{2}) +1      &\beta_{\mathrm{H}}\in(\sqrt{2},\beta^*_{\mathrm{H}}] \\[3ex]
    \theta(\beta_{\mathrm{H}},x(\beta_{\mathrm{H}})) &\beta_{\mathrm{H}}\in(\beta^*_{\mathrm{H}},3/2],
    \end{array}\right. 
    \label{ABentropybound-holz}
\end{align}
where the functions $\eta$, $\theta$ and $x$, and the parameter $\beta^*_{\mathrm{H}}$, are reported in Appendix~\ref{app:bounds-and-optimal-measurements}. Moreover, the bound is tight.
\end{conjec}

The bound on $H(A_0 B_0|E)$ when three parties test the Parity-CHSH inequality is also obtained by direct numerical optimization, similarly to the bound for the Holz inequality. As a matter of fact, note that the Parity-CHSH inequality \eqref{parity-chsh-ineq} is a particular case (upon relabeling the observables) of the Holz inequality \eqref{timo-ineq} when Charlie's two measurements coincide, i.e. $C_0 =C_1$, or equivalently $C_-=0$.

For the numerical computation of the entropy bound when two parties test the CHSH inequality, we apply the results of \cite{PironioAcin2009,Grasselli-PRXQuantum} to the CHSH scenario and parametrize the state shared by Alice and Bob as a Bell-diagonal state. We remark that the same bound has been independently computed in \cite{Colbeck2021} with numerical techniques. The details on the optimization problem solved for each numerical bound are given in Appendix~\ref{app:two-outcome-numerical}.

It is important to remark that the numerical curves obtained by directly minimizing the entropy cannot be treated as reliable lower bounds, as the numerical optimization is non-convex and may return local minima. Nevertheless, we believe that our optimizations are very close to the corresponding tight lower bounds.

In Fig.~\ref{fig:ABentropy-comparison} we plot the bounds on $H(A_0 B_0|E)$ as a function of the depolarization parameter $p$, in the cases of local and global depolarization. We observe that three parties testing the MABK inequality can certify a considerably higher amount of randomness for Alice's and Bob's outcomes, compared to testing the other inequalities, both for locally and globally depolarized states.

In Fig.~\ref{fig:ABentropy-Holz} we plot the analytical conjecture \eqref{ABentropybound-holz} and the corresponding numerical bound on $H(A_0 B_0|E)$ for the Holz inequality. We observe that the bound presents a distinct behavior for $\beta_{\mathrm{H}}\leq \sqrt{2}$ and $\beta_{\mathrm{H}}> \sqrt{2}$, represented by two distinct functions $\eta(\beta_{\mathrm{H}})$ and $\theta(\beta_{\mathrm{H}},x(\beta_{\mathrm{H}}))$ in \eqref{ABentropybound-holz}. Interestingly, the boundary of the two regions ($\beta_{\mathrm{H}}= \sqrt{2}$) coincides with the GME threshold \cite{Holz2019DICKA} above which the Holz inequality certifies genuine multipartite entanglement shared by the three parties. For smaller violations, the inequality cannot certify genuine multipartite entanglement and the entropy is bounded by: $H(A_0 B_0|E)\leq \eta(\sqrt{2})=1$, while for larger violations the entropy increases rapidly and eventually surpasses the other entropy bounds, except for MABK (see Fig.~\ref{fig:ABentropy-comparison}).

\begin{figure}[htb]
	\centering
    \includegraphics[width=0.9\linewidth,keepaspectratio]{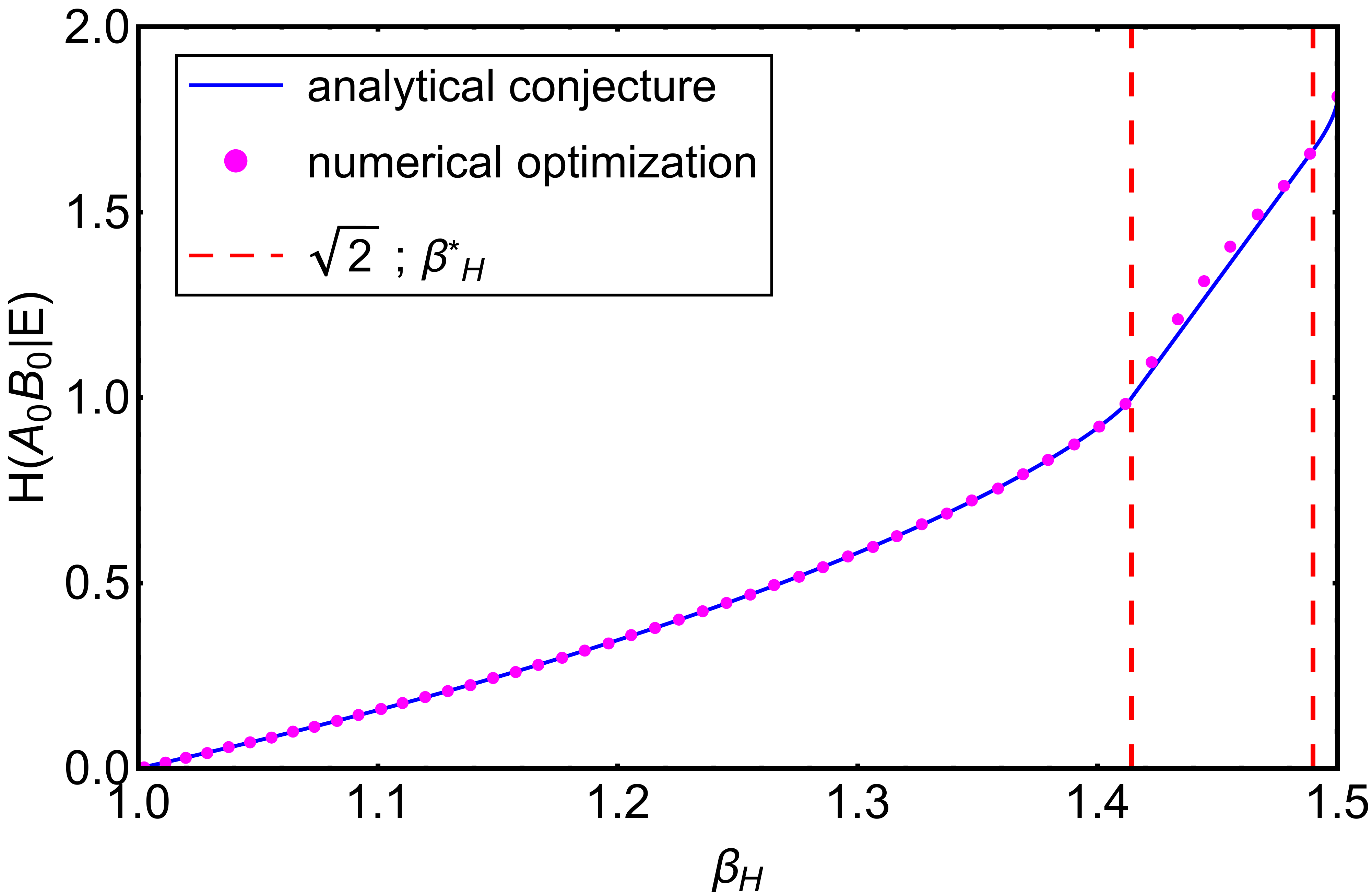}
	\caption{Lower bound on $H(A_0 B_0|E)$ as a function of the violation of the tripartite Holz inequality. The plot points are obtained by numerically minimizing the entropy for a fixed Bell value $\beta_{\rm H}$ (see Appendix~\ref{app:two-outcome-numerical} for details), while the blue solid line is our conjectured analytical bound \eqref{ABentropybound-holz}. Note that the numerical curve is concave in the interval enclosed by the red dashed lines while our analytical bound is convex in the whole domain, as required for a DI entropy bound.}
	\label{fig:ABentropy-Holz}
\end{figure}

We remark that the discrepancy between our analytical conjecture and the numerical curve in Fig.~\ref{fig:ABentropy-Holz} is because the latter is not always a convex function of the violation. Indeed, within the interval $(\sqrt{2},\beta^*_{\mathrm{H}}]$, the numerical curve and its conjectured analytical expression, $\theta(\beta_{\mathrm{H}},x(\beta_{\mathrm{H}}))$, become concave.
However, a DI lower bound on a conditional entropy must be a convex function of the violation. If this is not the case, Eve could distribute a convex combination of states yielding an entropy lower than that certified by the bound, thus spoiling its validity. For this, our conjectured bound \eqref{ABentropybound-holz} is constructed as the convex hull of the numerical curve, which guarantees that the bound is convex in the whole interval $(\sqrt{2},3/2]$. In particular, we replace the concave part of the curve by taking the tangent to the function $\theta(\beta_{\mathrm{H}},x(\beta_{\mathrm{H}}))$ at $\beta_{\mathrm{H}}=\beta^*_{\mathrm{H}}$, such that the point of coordinates $(\sqrt{2},1)$ belongs to the tangent line. This explains the definition of $\beta^*_{\mathrm{H}}$ given in Appendix~\ref{app:bounds-and-optimal-measurements}.

\section{Device-independent conference key agreement}\label{sec:DICKA}

The goal of device-independent conference key agreement (DICKA) is to establish a secret conference key among $N>2$ parties in a DI fashion. For this, it is necessary to certify the secrecy of Alice's outcome used to generate the key, which we choose to be $A_0$. This is done by testing a Bell inequality and computing a lower bound on $H(A_0|E)$, which indicates what fraction of Alice's outcome bit is secret with respect to the eavesdropper Eve. At the same time, Alice and the other parties want to obtain correlated outcomes to form the shared conference key. While such outcomes can be obtained from an additional measurement setting for the other parties, Alice's key-generating setting must be the same that is proved to be secret, i.e., $A_0$ \cite{HolzComment,Arnon-Friedman2018}. Due to potential noise affecting the parties' key-generating outcomes, Alice publicly broadcasts some error correction information for the other parties to correct their key bits and match Alice's. Asymptotically, the error-correction information needed by party $i$ to correct their key --affected by a bit error rate $Q_i$-- is given by a fraction $h(Q_i)$ of the whole key. Since the error-correction information is public, it is not secure, and must be subtracted from the fraction of secret key bits. Thus, the asymptotic conference key rate of a DICKA protocol, that is, the asymptotic rate of secret conference key bits produced per distributed  state, is given by \cite{JeremyMABK,CKAbook}:
\begin{align}
    r_{\rm DICKA}=H(A_0|E) - \max_{2 \leq i \leq N} h(Q_i), \label{DICKA-rate}
\end{align}
where we maximize the error-correction information over the error rates $Q_i$ so that even the party with the noisiest raw key can recover Alice's key.

Using the entropy bound on $H(A_0 |E)$ derived in Theorem~\ref{thm:Hbound} for the Holz inequality, together with the other bounds considered in Subsec.~\ref{sec:one-outcome-entropy} (the bound for the asymmetric CHSH inequality is numerically optimized over $\alpha$), we can compute the asymptotic secret key rate of DICKA protocols based on the Holz, the Parity-CHSH and the asymmetric CHSH inequality, where the latter is implemented as a concatenation of bipartite DIQKD protocols. In contrast, it is conjectured \cite{HolzComment,Grasselli-PRXQuantum} that the MABK inequality cannot be used in a DICKA protocol since Alice's optimal measurements yielding large violations are different from the key-generating measurement she uses to establish a shared conference key with the other parties.

In Fig.~\ref{fig:rate-comparison}, we plot the asymptotic conference key rate \eqref{DICKA-rate} of the DICKA protocols as a function of the parameter $p$, in logarithmic scale. The DICKA rates are obtained by using the optimal strategies reported in Appendix~\ref{app:bounds-and-optimal-measurements}, which require Alice's outcome $A_0$ to be the result of a Pauli $Z$ measurement. This setting allows the parties to obtain perfectly correlated key bits --in the ideal scenario of no depolarization-- if the others also choose $Z$ as their additional key-generating measurement.
\begin{figure}[htb] 
	\centering
	\textbf{Local depolarization}\par\medskip
		\includegraphics[width=0.9\linewidth,keepaspectratio]{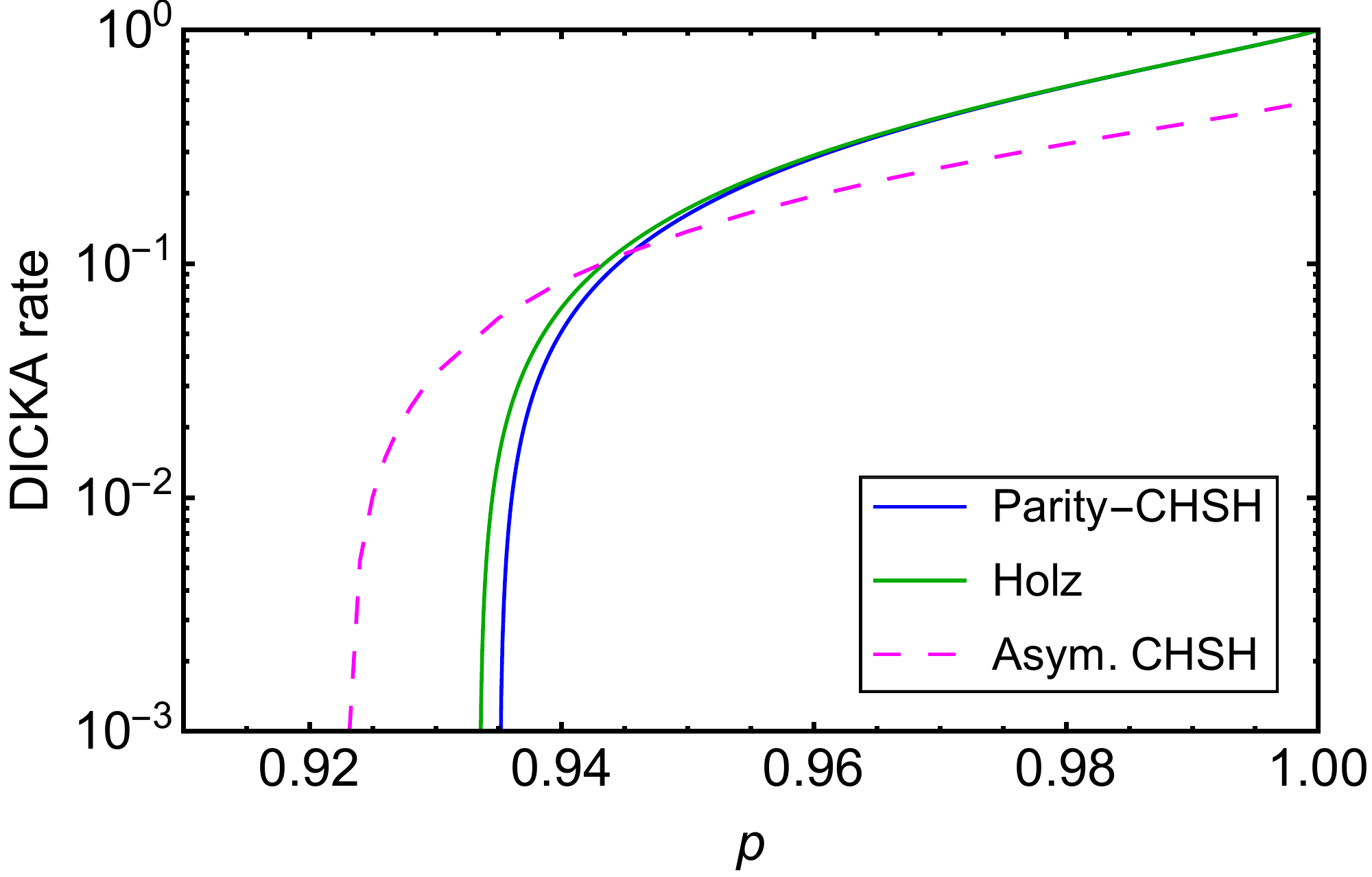}\\
		\vspace{1ex}
	\textbf{Global depolarization}\par\medskip
		\includegraphics[width=0.9\linewidth,keepaspectratio]{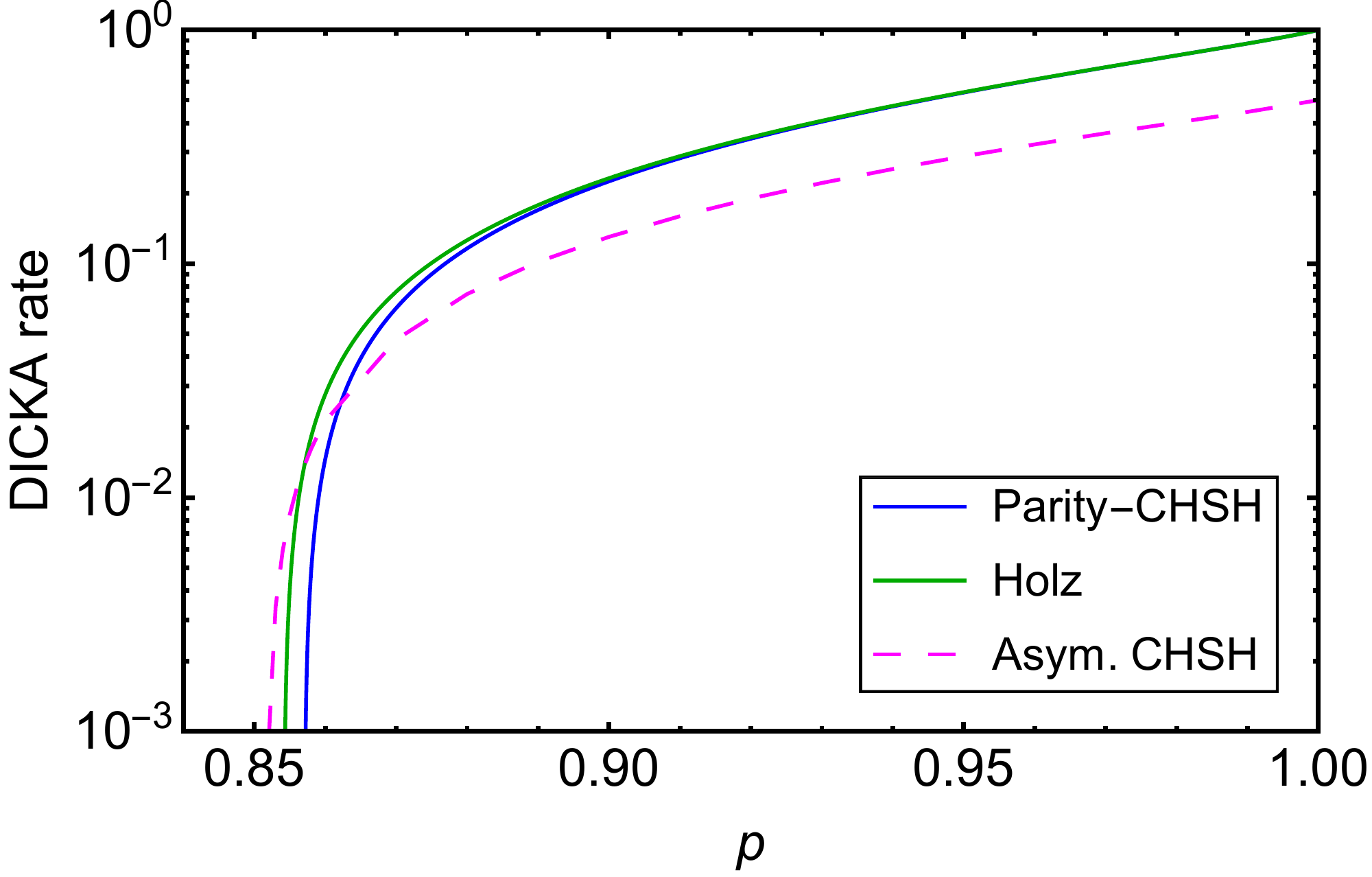}
	\caption{Asymptotic conference key rates, Eq.~\eqref{DICKA-rate}, of tripartite DICKA protocols based on the Holz inequality, Parity-CHSH inequality and a concatenation of bipartite DIQKD based on the asymmetric CHSH inequality, as a function of the parameter $p$ (the probability of local or global depolarization is $1-p$). The DICKA protocol based on the asymmetric CHSH inequality is composed of two consecutive DIQKD protocols that Alice performs with Bob and then with Charlie. Hence its key rate earns a factor of $1/2$. The bit error rate between every pair of parties, for both the GHZ and Bell state, is $Q=(1-p^2)/2$ ($Q=(1-p)/2$) when the state is locally (globally) depolarized.}
	\label{fig:rate-comparison}
\end{figure}

From Fig.~\ref{fig:rate-comparison} we observe that the optimal tripartite inequality for DICKA is the Holz inequality, which was expressly designed for this scope in \cite{Holz2019DICKA}. The plot also shows a clear advantage in establishing a conference key with a DICKA protocol based on multipartite entanglement rather than Bell pairs. This advantage nearly covers the whole range of $p$ in the case of global depolarization, while it vanishes for low values of $p$ and local depolarization. This is due to the starker effect of local depolarization on multipartite entangled states, which reduces their ability to violate a Bell inequality for a given value of local noise. The threshold values for $p$ above which a non-zero conference key can be extracted, in the case of local and global depolarization, are reported in Table~\ref{tab:p-val-DICKA}.
\begin{table}[ht] 
\centering
\begin{tabular}{|ccc|}
    \hline
    Bell ineq. & local noise & global noise \\
    \hline
    Holz &  $0.934$  & $0.855$ \\
    Parity-CHSH & $0.936$  & $0.858$ \\
    asym. CHSH & $0.923$  & $0.852$ \\
    \hline
\end{tabular}
\caption{Threshold values for $p$ (with $1-p$ being the probability of local or global depolarization, see Sec~\ref{sec:summary-results}) such that a non-zero conference key can be extracted, in the asymptotic limit, by DICKA protocols based on different Bell inequalities.}
\label{tab:p-val-DICKA}
\end{table}

We remark that, in our comparison, a DICKA protocol based on a bipartite Bell inequality, such as the asymmetric CHSH inequality \eqref{alphachsh-ineq}, is obtained as a concatenation of bipartite DIQKD protocols where Alice performs a DIQKD protocol first with Bob and then with Charlie\footnote{Alice then uses the two keys established with Bob and Charlie to distribute the conference key with one-time-pad.}. For this, the total number of states distributed per conference key bit doubles, causing a factor $1/2$ in the conference key rate \eqref{DICKA-rate}. In the case of $N$ parties establishing a conference key with concatenated DIQKD protocols, the conference key rate is reduced by a factor of $1/(N-1)$. 

Another important drawback of implementing DICKA by a concatenation of DIQKD protocols is that the security of the established conference key is spoiled unless Alice uses a new device for every iteration of the DIQKD protocol~\cite{Murta_2019,Memory_attack}, making it more resource demanding. We argue on this issue in Sec.~\ref{sec:discussion}, where we also discuss about alternative definitions of conference key rates where the advantage provided by multipartite entanglement can still be retained.

\section{Device-independent randomness expansion}\label{sec:DIRE}

A DIRE protocol aims to expand the initial share of private randomness of one or more parties in a DI way, by testing a Bell inequality. The amount of randomness in the outcomes of one or more parties is certified by computing a lower bound on a suitable conditional entropy, in terms of the observed Bell violation. Here we can envision a setup where the parties are located in the same lab\footnote{For device-independent randomness certification, it is essential, however, to ensure that the potentially malicious devices do not communicate.} and wish to explore the randomness of their joint outcomes. The goal is achieved when the amount of randomness produced by the protocol is greater than the input randomness used for testing the Bell inequality.

In this section we investigate the applicability to DIRE of the bounds on $H(A_0 B_0|E)$ presented in Sec.~\ref{sec:entropy-bounds}. Such bounds, as we will see, are particularly suited for spot-checking DIRE protocols \cite{Colbeck2021}, where Alice and Bob generate randomness with inputs $A_0$ and $B_0$ in most of the rounds and only sporadically test the Bell inequality with random inputs.

In Fig.~\ref{fig:ABentropy-comparison} of Sec.~\ref{sec:entropy-bounds} we observed that the two-outcome entropy bound for the MABK inequality is significantly larger than the other bounds. However, this does not necessarily imply that a DIRE protocol based on testing the MABK inequality can generate more \textit{net} randomness than DIRE protocols based on the other inequalities. This is due to the fact that the MABK inequality requires a larger amount of input randomness (two random inputs for Alice, Bob and Charlie compared to the Parity-CHSH and the CHSH inequality where Charlie has a fixed input or remains idle). A definitive answer could come from a thorough finite-key analysis of the DIRE protocols via the entropy accumulation theorem \cite{EAT}. Here instead, we aim at gaining intuition on the input/output randomness tradeoff by computing lower bounds on the asymptotic net randomness generation rate of DIRE protocols based on the four inequalities \eqref{timo-ineq}-\eqref{alphachsh-ineq}, which accounts for the input randomness required by each Bell test.

The net randomness generation rate of a DIRE protocol is the fraction of fresh random bits produced per distributed state, i.e., per round. In a spot-checking DIRE protocol a public source of randomness, shared by all parties, declares whether each round is a testing round ($T=1$ with probability $\gamma$) or a randomness-generation round ($T=0$ with probability $1-\gamma$). In a randomness-generation round, Alice and Bob select input $0$ and collect the outcomes $A_0$ and $B_0$, which generate $H(A_0 B_0 |E)$ bits of secret randomness (if the protocol involves additional parties, they also select a predefined input). In a testing round, the parties locally choose random inputs for their devices (represented by a joint random variable $I$) and test the selected Bell inequality. Since this step does not require public communication if the parties operate in the same lab, they consider their outcomes as part of the generated secret randomness. The conditional entropy that quantifies the amount of secret randomness generated in a test round is $H(AB|E)$. This entropy is larger than $H(A_0 B_0 |E)$, since in this case the inputs that generated the outputs ($AB$) are random and unknown to Eve. Thus, the output randomness generation rate of the spot-checking DIRE protocol is given by:
\begin{align}
    H(AB|TE) &=  (1-\gamma) H(A_0 B_0 |E) + \gamma H(AB|E) \nonumber\\
    &\geq H(A_0 B_0 |E), \label{output-randomness-spot}
\end{align}
where we used the strong sub-additivity of the von Neumann entropy in the inequality.

The input randomness consumed in a generic round of a spot-checking DIRE protocol is given by:
\begin{align}
    H(T,I) &= H(I|T) + H(T) \nonumber\\
    &= \gamma H(I|T=1) + (1-\gamma) H(I|T=0) + h(\gamma) \nonumber\\
    &= r \gamma + h(\gamma) , \label{input-randomness-spot}
\end{align}
where $r$ is the total number of random bits required as inputs by the selected Bell inequality. Then, the asymptotic net randomness generation rate of a spot-checking DIRE protocol is obtained by subtracting the input randomness \eqref{input-randomness-spot} from the output randomness \eqref{output-randomness-spot}, as follows:
\begin{align}
    r_{\rm spDIRE} &= H(AB|TE) - H(T,I) \nonumber\\
    &\geq H(A_0 B_0 |E) - r \gamma - h(\gamma). \label{net-randomness-spot}
\end{align}
Thus, the net randomness generation rate of a DIRE protocol based on testing the CHSH or Parity-CHSH inequality satisfies:
\begin{align}
    r_{\rm spDIRE} \geq H(A_0 B_0 |E) -2\gamma -h(\gamma) \label{net-randomness-spot-chsh},
\end{align}
while that of a DIRE protocol based on the MABK or the Holz inequality satisfies:
\begin{align}
    r_{\rm spDIRE} \geq H(A_0 B_0 |E) -3\gamma -h(\gamma) \label{net-randomness-spot-mabk},
\end{align}
since three random bits are required in each testing round. We can now use the two-outcome entropy bounds presented in Sec.~\ref{sec:entropy-bounds} in combination with \eqref{net-randomness-spot-chsh} and \eqref{net-randomness-spot-mabk} to compare the performance of the corresponding DIRE protocols.

We remark that, asymptotically, the optimal value for $\gamma$ tends to zero, i.e. it is sufficient to test the Bell inequality on a negligible fraction of rounds in order to learn the exact Bell violation. In this case, the net randomness generation rate in \eqref{net-randomness-spot} is given by $H(A_0 B_0 |E)$ with an equality sign. However, in order to investigate the effect of test rounds on spot-checking DIRE protocols where more input randomness is required (MABK and Holz inequalities), we set $\gamma$ to $\gamma=0.033\%$, which is the value used in the DIRE experiment of \cite{DIREexp1}.

\begin{figure}[htb]
	\centering
    \textbf{Local depolarization}\par\medskip
		\includegraphics[width=1\linewidth,keepaspectratio]{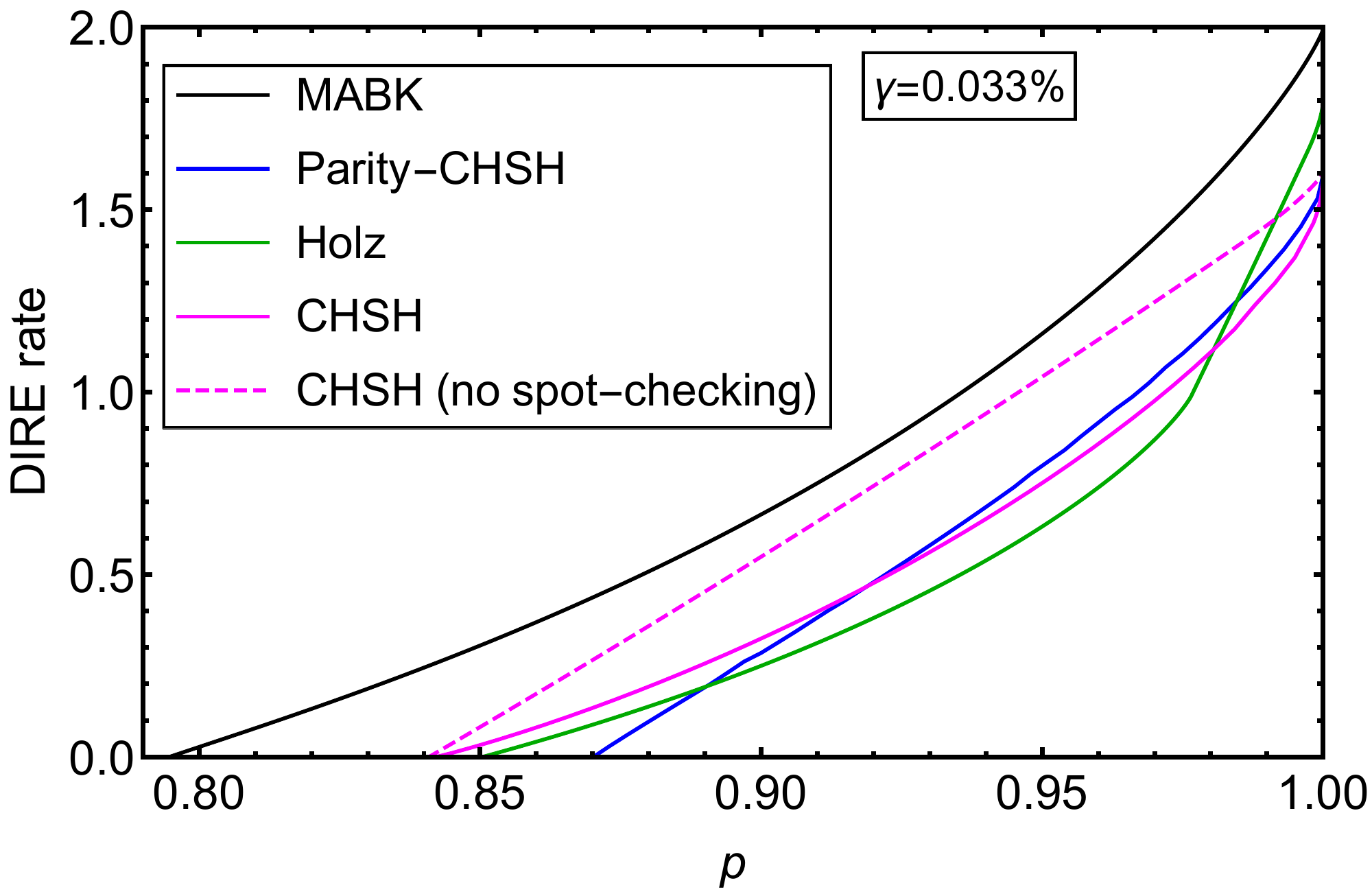}
	\vspace{1ex}
	\textbf{Global depolarization}\par\medskip
		\includegraphics[width=1\linewidth,keepaspectratio]{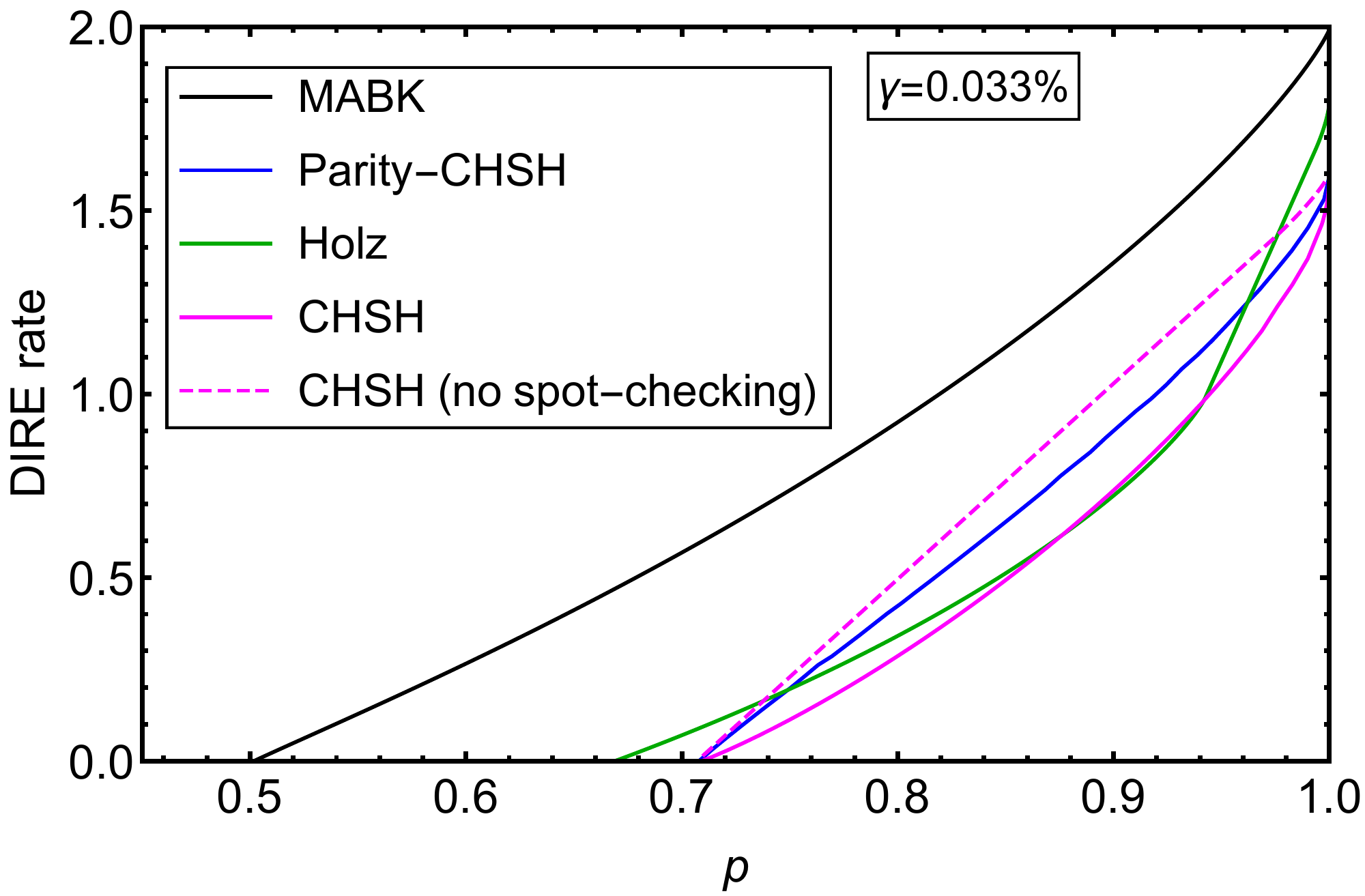}
	\caption{Asymptotic net randomness generation rates for DIRE protocols that extract secret randomness from the outcomes of two parties, as a function of $p$, where $1-p$ is the probability of local or global depolarization. The solid lines correspond to spot-checking DIRE protocols where a test round is performed on a fraction $\gamma$ of the total set of rounds, while the dashed line represents a DIRE protocol where the CHSH inequality is tested in every round and the random inputs are recycled \cite{Colbeck2021}. A spot-checking DIRE protocol based on the MABK inequality still generates the largest amount of net randomness. The value of $\gamma$ is set to the experimental value of \cite{DIREexp1}.}
	\label{fig:DIRE-rates}
\end{figure}

In order to benchmark the DIRE rates \eqref{net-randomness-spot-chsh} and \eqref{net-randomness-spot-mabk}, we consider a new type of DIRE protocol without spot-checking and based on the CHSH inequality \cite{Colbeck2021}. In this DIRE protocol there is no distinction between rounds and in each round Alice (Bob) randomly selects her (his) input $X$ ($Y$) and uses the outputs $A$ and $B$ to test the CHSH inequality. If the outputs are not publicly revealed (e.g. Alice and Bob are in the same lab), they can form part of the randomness generated by the protocol. Moreover, the protocol recycles the input randomness $X$ and $Y$ --which is also secret and unknown to Eve-- and appends it to the output randomness before extracting the secret random string. In this case, the asymptotic net randomness generation rate of the protocol is given by:
\begin{align}
    r_{\rm DIRE} &= H(A B XY |E) -H(XY) \nonumber\\
    &= H(A B |XY E)\label{net-randomness-chsh},
\end{align}
where, physically, the conditional entropy $H(A B |XY E)$ expresses the uncertainty that Eve has about Alice's and Bob's outcomes, when Alice and Bob use random inputs $X$ and $Y$ and the inputs become known to Eve after the measurement. The authors in \cite{Colbeck2021} conjecture a tight analytical lower bound on $H(A B |XY E) $ as a function of the CHSH value $\beta_C$ (reported in Appendix~\ref{app:bounds-and-optimal-measurements}), which we employ to plot the DIRE rate \eqref{net-randomness-chsh}.

In Fig.~\ref{fig:DIRE-rates} we plot the asymptotic net randomness generation rates of spot-checking DIRE protocols based on the CHSH and Parity-CHSH inequalities \eqref{net-randomness-spot-chsh} and on the MABK and Holz inequalities \eqref{net-randomness-spot-mabk}, as well as the net randomness generation rate of the DIRE protocol without spot-checking based on the CHSH inequality \eqref{net-randomness-chsh}. The threshold values for $p$ above which we have a positive randomness generation rate, assuming $\gamma=0$, can be calculated analytically since each of the analyzed two-outcome entropy bounds yields non-zero randomness as soon as the corresponding classical bound is violated. The numerical values we obtain are reported in Table~\ref{tab:p-val-DIRE}.
\begin{table}[ht] 
\centering
\begin{tabular}{|ccc|}
    \hline
    Bell ineq. & local noise & global noise \\
    \hline
   MABK &  $0.794$  & $0.500$  \\ 
Parity-CHSH &  $0.870$  & $0.707$  \\  
Holz &  $0.84$9  & $0.667$   \\   
 CHSH & $0.841$  & $0.707$  \\ 
CHSH (no spot-ch.) &  $0.841$  & $0.707$ \\
    \hline
\end{tabular}
\caption{Threshold values for $p$ (with $1-p$ being the probability of local or global depolarization, see Sec~\ref{sec:summary-results}) such that, asymptotically, DIRE protocols based on the analyzed Bell inequalities yield a positive net-randomness generation rate.}
\label{tab:p-val-DIRE}
\end{table}

From Fig.~\ref{fig:DIRE-rates} we observe that the optimal Bell inequality for DIRE is the MABK inequality, both in terms of randomness generation rate and noise tolerance. However, DIRE protocols based on the Holz inequality can also outperform protocols based on the CHSH inequality in the low-noise regime (high values of $p$), even when compared to CHSH-based DIRE protocols which recycle the input randomness. Moreover, for the chosen realistic value of $\gamma$, the effect of test rounds on the DIRE rates is negligible as they approximately coincide with the entropy curves in Fig.~\ref{fig:ABentropy-comparison}.

Nevertheless, for a fixed Bell inequality, a DIRE protocol without spot-checking and with recycled input randomness may yield more net randomness than the corresponding protocol with spot-checking. This holds for any DIRE protocol based on a Bell inequality which is symmetric with respect to permutations of the parties' observables, like the CHSH inequality. Indeed, the asymptotic net randomness generation rate of, say, a bipartite spot-checking protocol is given by a lower bound $F(\beta)$ on $H(A_0 B_0 |E)$ (recall that $\gamma\to 0$ in the asymptotic regime), while the rate of the protocol without spot-checking and with recycled input randomness is given by a lower bound on $H(AB|XYE)$. Due to the permutation symmetry of the inequality, the lower bound on $H(A_0 B_0 |E)$ is actually valid for any combination of inputs of Alice and Bob: $F(\beta) \leq H(A_k B_l |E) \equiv H(AB|X=k,Y=l,E)$ for all $k,l\in \{0,1\}$, hence it is also a  lower bound on $H(AB|XYE)=\sum_{k,l} p_{kl} H(AB|X=k,Y=l,E)$. However, it is likely that a direct calculation of $H(AB|XYE)$ would lead to a tighter lower bound and hence to a higher rate for the DIRE protocol without spot-checking. This is confirmed in the CHSH case by Fig.~\ref{fig:DIRE-rates}, where the dashed magenta line (DIRE protocol without spot-checking) lies significantly above the solid magenta line (DIRE protocol with spot-checking).

\section{Discussion and conclusion} \label{sec:discussion}

In this work we consider a two-input/two-output tripartite device-independent (DI) scenario with different tripartite Bell inequalities, namely: the Holz inequality \cite{Holz2019DICKA}, the MABK inequality \cite{Mermin,Ardehali,BK93}, and the Parity-CHSH inequality \cite{JeremyParityCHSH}, as well as the family of (bipartite) asymmetric CHSH inequalities \cite{AsymCHSH,AsymCHSH-Woodhead}. We investigate the asymptotic performance of DI conference key agreement (DICKA) and DI randomness expansion (DIRE) protocols when the different Bell inequalities are tested to certify private randomness in the parties' outcomes. To this aim, we present analytical and numerical lower bounds on the conditional von Neumann entropy of a single party's outcome and of two parties' outcomes, as a function of the Bell inequality violation. We provide a concise overview of the bounds in Table~\ref{tab:summary}.

Specifically, for the Holz inequality \cite{Holz2019DICKA} we derive a tight analytical bound on the one-outcome entropy, $H(A_0|E)$, and conjecture a tight analytical bound on the two-outcome entropy, $H(A_0B_0|E)$, that is strongly supported by numerical results. These are the first tight analytical bounds on the conditional von Neumann entropy for a non-full-correlator Bell inequality --apart from the one-outcome entropy bound derived in \cite{JeremyParityCHSH} for the Parity-CHSH inequality, whose tightness is only proved in this work in Appendix~\ref{app:tightness-parityCHSH-bound}. For the Parity-CHSH and CHSH inequality we instead compute numerical bounds on the two-outcome entropy $H(A_0B_0|E)$.

By using the derived bounds together with bounds obtained in previous literature  \cite{AsymCHSH-Woodhead,JeremyParityCHSH,Grasselli-PRXQuantum}, we compute the asymptotic conference key rate (net randomness generation rate) of DICKA (DIRE) protocols based on the above-mentioned Bell inequalities. We remark, however, that the analytical bounds presented in this work could be applied to finite-key analyses of DICKA and DIRE protocols through the entropy accumulation theorem \cite{EAT,Arnon-Friedman2018}, although we leave this as a matter for future work.

\begin{table*}[t]
\centering
\renewcommand{\arraystretch}{1.2}
\begin{tabular}{|P{3cm}|c|c|c|c|} 
 \hline
 \multirow{2}{*}{\textbf{Bell inequality}} & \multicolumn{2}{c|}{$\mathbf{H(A_0 |E)}$} & \multicolumn{2}{c|}{$\mathbf{H(A_0 B_0|E)}$}\\
 \cline{2-5}
 & lower bound & tight & lower bound & tight \\
 \hline
    \multirow{2}{*}{Holz} & Theorem~\ref{thm:Hbound}   & YES & Conjecture~\ref{conj:H2bound} & YES \\
    &  [This work] &  [This work] & [This work] & [This work]\\
 \hline
    \multirow{2}{*}{Parity-CHSH} & \eqref{entropybound-parity-app}  & YES & numerical & YES$^*$ \\
    &  \cite{JeremyParityCHSH} &  [This work] & [This work] & [This work]\\
 \hline
    \multirow{2}{*}{MABK} & Appendix~\ref{app:MABK-one-outcome-proof}  & NO & \eqref{ABentropybound-mabk-app} & NO \\
    &  [This work] \& \cite{JeremyMABK,Grasselli-PRXQuantum} &  \cite{Grasselli-PRXQuantum} & \cite{Grasselli-PRXQuantum} & \cite{Grasselli-PRXQuantum}\\
 \hline
    \multirow{2}{*}{\parbox{3cm}{asymmetric CHSH (CHSH for $\alpha=1$)}} & \eqref{entropybound-alphachsh-app}  & YES & numerical ($\alpha=1$)  & YES$^*$  \\
    &  \cite{AsymCHSH-Woodhead} &  \cite{AsymCHSH-Woodhead} & [This work] \& \cite{Colbeck2021} & [This work] \& \cite{Colbeck2021} \\
 \hline
\end{tabular}
\caption{Summary of the one-outcome and two-outcome entropy bounds used to investigate DIRE and DICKA protocols based on different Bell inequalities. The expressions of all the analytical bounds are reported in Appendix~\ref{app:bounds-and-optimal-measurements}. We additionally employ the tight analytical bound \eqref{entropybound-colbeck-app} on $H(AB|XYE)$, conjectured in \cite{Colbeck2021}, to study a CHSH-based DIRE protocol without spot-checking and with recycled input randomness. $^*$Note that the numerical lower bounds obtained in this work are not reliable as they are the result of a non-convex minimization of the entropy over the set of states compatible with the Bell violation. Hence, their tightness is understood as the existence of an implementation which attains the plotted curve.}
\label{tab:summary}
\end{table*}

Importantly, our results show that DI protocols based on multipartite Bell inequalities can outperform implementations based on bipartite Bell inequalities under different noise models (local and global depolarizing noise) and for a broad range of noise levels.

For the task of DICKA, the Holz inequality turns out to be the one that yields the largest key rate and, in general, DICKA is better performed with multipartite Bell inequalities. In Sec.~\ref{sec:DICKA} we remark that a DICKA protocol based on a bipartite Bell inequality must be obtained as a concatenation of DIQKD protocols subsequently run, e.g., between Alice and each of the other parties. Hence, its conference key rate, defined as the fraction of secret conference key bits per distributed state, is overly penalized compared to the key rate of a DICKA protocol based on a multipartite Bell inequality. In this regard, one could argue that a more practical definition of conference key rate is given by the fraction of secret conference key bits generated per unit of time. In this case, the relationship between the key rates in Fig.~\ref{fig:rate-comparison} might change significantly due to, e.g., a faster distribution of Bell pairs compared to multipartite entangled states. However, future quantum networks might generate highly non-trivial resource states in order to suit the different needs of its nodes \cite{experimentalCKA-graphstates}, or could present peculiar topologies (e.g. bottle-necks \cite{Epping}) such that the distillation of a multipartite entangled state and of a Bell pair require the same amount of resources and time. In these cases, the advantage of performing DICKA with multipartite entangled states rather than a concatenation of DIQKD protocols would be retained even with more practical definitions of conference key rate.

Another important drawback of implementing DICKA by a concatenation of DIQKD protocols is that the security of the established conference key is spoiled unless Alice uses a new device for every iteration of the DIQKD protocol. Indeed, reusing the same quantum device in independent runs of a DIQKD protocol could lead to security loopholes \cite{Memory_attack,Murta_2019}. Thus, in such a case, Alice would need to possess $N-1$ quantum measurement devices in order to establish a conference key with $N-1$ other parties. For the tripartite DICKA protocol based on the asymmetric CHSH inequality, she would need two distinct quantum devices. Conversely, Alice can establish a conference key with an arbitrary number of parties with only one device, by a single run of a DICKA protocol based on multipartite Bell inequalities, such as the Holz inequality.

Concerning DIRE, we observe that a spot-checking DIRE protocol based on the MABK inequality is the one that yields the largest amount of net randomness in the outcomes of two parties, even when compared with a recent CHSH-based DIRE protocol without spot-checking, where the input randomness is recycled \cite{Colbeck2021}. The advantage of the MABK inequality over the other inequalities could lie in its permutational symmetry, which does not privilege one party at the expense of the other parties (as in the Holz and the Parity-CHSH inequality). However, the CHSH inequality is also permutationally-invariant but its bound on $H(A_0 B_0 |E)$ lies well below the bound for the MABK inequality. This could be explained by the fact that, for any full-correlator Bell inequality such as the MABK and the CHSH inequality, one can assume all the marginal distributions of the outcomes to be symmetric, without loss of generality \cite{Grasselli-PRXQuantum}. Therefore, for the tripartite MABK inequality we have that the unconditional entropy of Alice's and Bob's outcomes is maximal: $H(A_0 B_0)=2$. While in general this is not true for the bipartite CHSH inequality: $H(A_0 B_0)<2$, where the conditions on the marginals ($\braket{A_0}=\braket{B_0}=0$) cannot fix the joint distribution of $A_0 B_0$.

Our work suggests many possible lines of future research. To start with, we emphasize that the techniques employed in the derivation of the one-outcome entropy bounds for the Holz and the MABK inequality (Theorem~\ref{thm:Hbound} and Appendix~\ref{app:MABK-one-outcome-proof}) are applicable to general Bell inequalities with two inputs and two outputs per party. For instance, one could generalize Theorem~\ref{thm:Hbound} to the multipartite scenario where the $N$-party Holz inequality is tested \cite{Holz2019DICKA}, although the way to achieve this might be highly non-trivial. This result, nevertheless, could lead to the best DICKA rate achievable by $N$ parties, since the Holz inequality was introduced in~\cite{Holz2019DICKA} exactly for the purpose of DICKA. On a similar note, one could derive a one-outcome entropy bound that accounts for noisy pre-processing and bias of Alice's raw output, similarly to what has been done in the bipartite case for the CHSH inequality \cite{AsymCHSH-Woodhead,masini-entropy-bounds}.

In the case of DIRE, it is important to remark that the net randomness generated by testing tripartite Bell inequalities can be significantly increased compared to the results presented in Sec.~\ref{sec:DIRE}, making multipartite nonlocality even more beneficial for DI cryptography. One obvious way to increase the DI randomness is to combine the outputs of all three parties in the randomness-generation rounds, instead of only using Alice's and Bob's outputs. However, this requires the derivation of bounds on three-outcome entropies of the form  $H(A_i B_j C_k |E)$, for which no analytical nor numerical result is yet available. Another way to generate more randomness from multipartite nonlocality is to extend the idea of DIRE protocols without spot-checking and with recycled input randomness to the multiparty scenario. Indeed, such protocols can outperform the corresponding spot-checking protocol that tests the same Bell inequality, especially when the latter is permutationally invariant. Therefore, an important avenue to improve the randomness generated by the MABK-based DIRE protocol, and any multipartite DIRE protocol based on permutationally-invariant Bell inequalities, is the derivation of entropy bounds on quantities like $H(AB|XYZE)$ and $H(ABC|XYZE)$.

Besides deriving new entropy bounds and improving the performance of DI protocols, our work leaves an interesting question open. In Sec.~\ref{sec:entropy-bounds} we show that genuine multipartite entanglement (GME) is not always necessary to certify non-zero entropy in a single party's outcome when testing a multipartite Bell inequality, especially when the latter presents asymmetries. This means that GME is not a precondition for DIRE protocols with multipartite Bell inequalities. Conversely, DICKA schemes, apart from certifying private randomness, also require the parties to obtain correlated outcomes that can form a shared conference key. It remains an open question whether non-GME states can be used to perform DICKA, while their usefulness has been established in the case of device-dependent CKA \cite{Carrara2020}.

\begin{acknowledgements}
We acknowledge support by the Deutsche Forschungsgemeinschaft (DFG, German Research Foundation) under Germany’s Excellence Strategy - Cluster of Excellence Matter and Light for Quantum Computing (ML4Q) EXC 2004/1 -390534769. D.B. and H.K. acknowledge support by the QuantERA project QuICHE, via the German Ministry for Education and Research (BMBF Grant No. 16KIS1119K). F.G. and D.B. acknowledge support by the DFG Individual Research Grant BR2159/6-1.
\end{acknowledgements}

\bibliographystyle{quantum}

\bibliography{biblio-quantum}

\begin{thebibliography}{10}

\bibitem{pir-advances-qcrypto}
S.~Pirandola, U.~L. Andersen, L.~Banchi, M.~Berta, D.~Bunandar, R.~Colbeck,
  D.~Englund, T.~Gehring, C.~Lupo, C.~Ottaviani, J.~L. Pereira, M.~Razavi,
  J.~Shamsul Shaari, M.~Tomamichel, V.~C. Usenko, G.~Vallone, P.~Villoresi, and
  P.~Wallden.
\newblock ``Advances in quantum cryptography''.
\newblock \href{https://dx.doi.org/10.1364/AOP.361502}{Adv. Opt. Photon. {\bf
  12}, 1012--1236}~(2020).

\bibitem{review-pan}
Feihu Xu, Xiongfeng Ma, Qiang Zhang, Hoi-Kwong Lo, and Jian-Wei Pan.
\newblock ``Secure quantum key distribution with realistic devices''.
\newblock \href{https://dx.doi.org/10.1103/RevModPhys.92.025002}{Rev. Mod.
  Phys. {\bf 92}, 025002}~(2020).

\bibitem{makarov-nature}
Lars Lydersen, Carlos Wiechers, Christoffer Wittmann, Dominique Elser, Johannes
  Skaar, and Vadim Makarov.
\newblock ``Hacking commercial quantum cryptography systems by tailored bright
  illumination''.
\newblock \href{https://dx.doi.org/10.1038/nphoton.2010.214}{Nature Photonics
  {\bf 4}, 686--689}~(2010).

\bibitem{makarov-nature2}
Ilja Gerhardt, Qin Liu, Ant{\'i}a Lamas-Linares, Johannes Skaar, Christian
  Kurtsiefer, and Vadim Makarov.
\newblock ``Full-field implementation of a perfect eavesdropper on a quantum
  cryptography system''.
\newblock \href{https://dx.doi.org/10.1038/ncomms1348}{Nature Communications
  {\bf 2}, 349}~(2011).

\bibitem{YaoMayers}
A.~Yao and D.~Mayers.
\newblock ``Quantum cryptography with imperfect apparatus''.
\newblock In IEEE 54th Annual Symposium on Foundations of Computer Science.
\newblock \href{https://dx.doi.org/10.1109/SFCS.1998.743501}{Page 503}.
\newblock Los Alamitos, CA, USA~(1998). IEEE Computer Society.

\bibitem{Acin2006}
Antonio Ac\'in, Nicolas Gisin, and Lluis Masanes.
\newblock ``From bell's theorem to secure quantum key distribution''.
\newblock \href{https://dx.doi.org/10.1103/PhysRevLett.97.120405}{Phys. Rev.
  Lett. {\bf 97}, 120405}~(2006).

\bibitem{BarrentKent}
Jonathan Barrett, Adrian Kent, and Stefano Pironio.
\newblock ``Maximally nonlocal and monogamous quantum correlations''.
\newblock \href{https://dx.doi.org/10.1103/PhysRevLett.97.170409}{Phys. Rev.
  Lett. {\bf 97}, 170409}~(2006).

\bibitem{Bellineq}
J.~S. Bell and Alain Aspect.
\newblock ``Speakable and unspeakable in quantum mechanics: Collected papers on
  quantum philosophy''.
\newblock \href{https://dx.doi.org/10.1017/CBO9780511815676}{Cambridge
  University Press}. ~(2004).
\newblock 2 edition.

\bibitem{reviewBell}
Nicolas Brunner, Daniel Cavalcanti, Stefano Pironio, Valerio Scarani, and
  Stephanie Wehner.
\newblock ``Bell nonlocality''.
\newblock \href{https://dx.doi.org/10.1103/RevModPhys.86.419}{Rev. Mod. Phys.
  {\bf 86}, 419--478}~(2014).

\bibitem{loopholefree1}
B.~Hensen, H.~Bernien, A.~E. Dr{\'e}au, A.~Reiserer, N.~Kalb, M.~S. Blok,
  J.~Ruitenberg, R.~F.~L. Vermeulen, R.~N. Schouten, C.~Abell{\'a}n, W.~Amaya,
  V.~Pruneri, M.~W. Mitchell, M.~Markham, D.~J. Twitchen, D.~Elkouss,
  S.~Wehner, T.~H. Taminiau, and R.~Hanson.
\newblock ``Loophole-free bell inequality violation using electron spins
  separated by 1.3 kilometres''.
\newblock \href{https://dx.doi.org/10.1038/nature15759}{Nature {\bf 526},
  682--686}~(2015).

\bibitem{loopholefree2}
Marissa Giustina, Marijn A.~M. Versteegh, S\"oren Wengerowsky, Johannes
  Handsteiner, Armin Hochrainer, Kevin Phelan, Fabian Steinlechner, Johannes
  Kofler, Jan-\AA{}ke Larsson, Carlos Abell\'an, Waldimar Amaya, Valerio
  Pruneri, Morgan~W. Mitchell, J\"orn Beyer, Thomas Gerrits, Adriana~E. Lita,
  Lynden~K. Shalm, Sae~Woo Nam, Thomas Scheidl, Rupert Ursin, Bernhard
  Wittmann, and Anton Zeilinger.
\newblock ``Significant-loophole-free test of bell's theorem with entangled
  photons''.
\newblock \href{https://dx.doi.org/10.1103/PhysRevLett.115.250401}{Phys. Rev.
  Lett. {\bf 115}, 250401}~(2015).

\bibitem{AcinBrunner2007}
Antonio Ac\'in, Nicolas Brunner, Nicolas Gisin, Serge Massar, Stefano Pironio,
  and Valerio Scarani.
\newblock ``Device-independent security of quantum cryptography against
  collective attacks''.
\newblock \href{https://dx.doi.org/10.1103/PhysRevLett.98.230501}{Phys. Rev.
  Lett. {\bf 98}, 230501}~(2007).

\bibitem{PironioAcin2009}
Stefano Pironio, Antonio Ac\'in, Nicolas Brunner, Nicolas Gisin, Serge Massar,
  and Valerio Scarani.
\newblock ``Device-independent quantum key distribution secure against
  collective attacks''.
\newblock \href{https://dx.doi.org/10.1088/1367-2630/11/4/045021}{New Journal
  of Physics {\bf 11}, 045021}~(2009).

\bibitem{Masanes2011}
Llu{\'i}s Masanes, Stefano Pironio, and Antonio Ac{\'i}n.
\newblock ``Secure device-independent quantum key distribution with causally
  independent measurement devices''.
\newblock \href{https://dx.doi.org/10.1038/ncomms1244}{Nature Communications
  {\bf 2}, 238}~(2011).

\bibitem{VidickDIQKD}
Umesh Vazirani and Thomas Vidick.
\newblock ``Fully device-independent quantum key distribution''.
\newblock \href{https://dx.doi.org/10.1103/PhysRevLett.113.140501}{Phys. Rev.
  Lett. {\bf 113}, 140501}~(2014).

\bibitem{Arnon-Friedman2018}
Rotem Arnon-Friedman, Fr{\'e}d{\'e}ric Dupuis, Omar Fawzi, Renato Renner, and
  Thomas Vidick.
\newblock ``Practical device-independent quantum cryptography via entropy
  accumulation''.
\newblock \href{https://dx.doi.org/10.1038/s41467-017-02307-4}{Nature
  Communications {\bf 9}, 459}~(2018).

\bibitem{SG01}
Valerio Scarani and Nicolas Gisin.
\newblock ``Quantum communication between n partners and bell's inequalities''.
\newblock \href{https://dx.doi.org/10.1103/PhysRevLett.87.117901}{Phys. Rev.
  Lett. {\bf 87}, 117901}~(2001).

\bibitem{SG_pra_01}
Valerio Scarani and Nicolas Gisin.
\newblock ``Quantum key distribution between n partners: Optimal eavesdropping
  and bell's inequalities''.
\newblock \href{https://dx.doi.org/10.1103/PhysRevA.65.012311}{Phys. Rev. A
  {\bf 65}, 012311}~(2001).

\bibitem{Holz2019DICKA}
Timo Holz, Hermann Kampermann, and Dagmar Bru\ss{}.
\newblock ``Genuine multipartite bell inequality for device-independent
  conference key agreement''.
\newblock \href{https://dx.doi.org/10.1103/PhysRevResearch.2.023251}{Phys. Rev.
  Research {\bf 2}, 023251}~(2020).

\bibitem{JeremyParityCHSH}
J\'er\'emy Ribeiro, Gl\'aucia Murta, and Stephanie Wehner.
\newblock ``Reply to ``comment on fully device-independent conference key
  agreement''''.
\newblock \href{https://dx.doi.org/10.1103/PhysRevA.100.026302}{Phys. Rev. A
  {\bf 100}, 026302}~(2019).

\bibitem{CKA-review}
Gláucia Murta, Federico Grasselli, Hermann Kampermann, and Dagmar Bruß.
\newblock ``Quantum conference key agreement: A review''.
\newblock
  \href{https://dx.doi.org/https://doi.org/10.1002/qute.202000025}{Advanced
  Quantum Technologies {\bf 3}, 2000025}~(2020).

\bibitem{ColbeckThesis2006}
Roger Colbeck.
\newblock ``Quantum and relativistic protocols for secure multi-party
  computation''~(2011).
\newblock  \href{http://arxiv.org/abs/0911.3814}{arXiv:0911.3814}.

\bibitem{Pironio2010}
S.~Pironio, A.~Ac\'in, S.~Massar, A.~Boyer de~la Giroday, D.~N. Matsukevich,
  P.~Maunz, S.~Olmschenk, D.~Hayes, L.~Luo, T.~A. Manning, et~al.
\newblock ``Random numbers certified by bell's theorem''.
\newblock \href{https://dx.doi.org/10.1038/nature09008}{Nature {\bf 464},
  1021--1024}~(2010).

\bibitem{Colbeck2011}
Roger Colbeck and Adrian Kent.
\newblock ``Private randomness expansion with untrusted devices''.
\newblock \href{https://dx.doi.org/10.1088/1751-8113/44/9/095305}{Journal of
  Physics A: Mathematical and Theoretical {\bf 44}, 095305}~(2011).

\bibitem{securityDIrandomness1}
Carl~A. Miller and Yaoyun Shi.
\newblock ``Robust protocols for securely expanding randomness and distributing
  keys using untrusted quantum devices''.
\newblock \href{https://dx.doi.org/10.1145/2885493}{J. ACM{\bf 63}}~(2016).

\bibitem{securityDIrandomness2}
Stefano Pironio and Serge Massar.
\newblock ``Security of practical private randomness generation''.
\newblock \href{https://dx.doi.org/10.1103/PhysRevA.87.012336}{Phys. Rev. A
  {\bf 87}, 012336}~(2013).

\bibitem{securityDIrandomness3}
Serge Fehr, Ran Gelles, and Christian Schaffner.
\newblock ``Security and composability of randomness expansion from bell
  inequalities''.
\newblock \href{https://dx.doi.org/10.1103/PhysRevA.87.012335}{Phys. Rev. A
  {\bf 87}, 012335}~(2013).

\bibitem{Woodhead2018}
Erik Woodhead, Boris Bourdoncle, and Antonio Ac\'in.
\newblock ``Randomness versus nonlocality in the {M}ermin-{B}ell experiment
  with three parties''.
\newblock \href{https://dx.doi.org/10.22331/q-2018-08-17-82}{{Quantum} {\bf 2},
  82}~(2018).

\bibitem{DIREexp1}
Wen-Zhao Liu, Ming-Han Li, Sammy Ragy, Si-Ran Zhao, Bing Bai, Yang Liu,
  Peter~J. Brown, Jun Zhang, Roger Colbeck, Jingyun Fan, Qiang Zhang, and
  Jian-Wei Pan.
\newblock ``Device-independent randomness expansion against quantum side
  information''.
\newblock \href{https://dx.doi.org/10.1038/s41567-020-01147-2}{Nature Physics
  {\bf 17}, 448--451}~(2021).

\bibitem{DIREexp2}
Lynden~K. Shalm, Yanbao Zhang, Joshua~C. Bienfang, Collin Schlager, Martin~J.
  Stevens, Michael~D. Mazurek, Carlos Abell{\'a}n, Waldimar Amaya, Morgan~W.
  Mitchell, Mohammad~A. Alhejji, Honghao Fu, Joel Ornstein, Richard~P. Mirin,
  Sae~Woo Nam, and Emanuel Knill.
\newblock ``Device-independent randomness expansion with entangled photons''.
\newblock \href{https://dx.doi.org/10.1038/s41567-020-01153-4}{Nature Physics
  {\bf 17}, 452--456}~(2021).

\bibitem{DIQKDexp1}
Wei Zhang, Tim van Leent, Kai Redeker, Robert Garthoff, Ren{\'{e}} Schwonnek,
  Florian Fertig, Sebastian Eppelt, Wenjamin Rosenfeld, Valerio Scarani,
  Charles C.-W. Lim, and Harald Weinfurter.
\newblock ``A device-independent quantum key distribution system for distant
  users''.
\newblock \href{https://dx.doi.org/10.1038/s41586-022-04891-y}{Nature {\bf
  607}, 687--691}~(2022).

\bibitem{DIQKDexp2}
D.~P. Nadlinger, P.~Drmota, B.~C. Nichol, G.~Araneda, D.~Main, R.~Srinivas,
  D.~M. Lucas, C.~J. Ballance, K.~Ivanov, E.~Y.-Z. Tan, P.~Sekatski, R.~L.
  Urbanke, R.~Renner, N.~Sangouard, and J.-D. Bancal.
\newblock ``Experimental quantum key distribution certified by
  bell{\textquotesingle}s theorem''.
\newblock \href{https://dx.doi.org/10.1038/s41586-022-04941-5}{Nature {\bf
  607}, 682--686}~(2022).

\bibitem{DIQKDexp3}
Wen-Zhao Liu, Yu-Zhe Zhang, Yi-Zheng Zhen, Ming-Han Li, Yang Liu, Jingyun Fan,
  Feihu Xu, Qiang Zhang, and Jian-Wei Pan.
\newblock ``Toward a photonic demonstration of device-independent quantum key
  distribution''.
\newblock \href{https://dx.doi.org/10.1103/PhysRevLett.129.050502}{Phys. Rev.
  Lett. {\bf 129}, 050502}~(2022).

\bibitem{CHSH}
John~F. Clauser, Michael~A. Horne, Abner Shimony, and Richard~A. Holt.
\newblock ``Proposed experiment to test local hidden-variable theories''.
\newblock \href{https://dx.doi.org/10.1103/PhysRevLett.23.880}{Phys. Rev. Lett.
  {\bf 23}, 880--884}~(1969).

\bibitem{AsymCHSH-Woodhead}
Erik Woodhead, Antonio Ac{\'{i}}n, and Stefano Pironio.
\newblock ``Device-independent quantum key distribution with asymmetric {CHSH}
  inequalities''.
\newblock \href{https://dx.doi.org/10.22331/q-2021-04-26-443}{{Quantum} {\bf
  5}, 443}~(2021).

\bibitem{masini-entropy-bounds}
Michele Masini, Stefano Pironio, and Erik Woodhead.
\newblock ``Simple and practical {DIQKD} security analysis via {BB}84-type
  uncertainty relations and pauli correlation constraints''.
\newblock \href{https://dx.doi.org/10.22331/q-2022-10-20-843}{Quantum {\bf 6},
  843}~(2022).

\bibitem{Sekatski-entropy-bounds}
Pavel Sekatski, Jean-Daniel Bancal, Xavier Valcarce, Ernest Y.-Z. Tan, Renato
  Renner, and Nicolas Sangouard.
\newblock ``Device-independent quantum key distribution from generalized {CHSH}
  inequalities''.
\newblock \href{https://dx.doi.org/10.22331/q-2021-04-26-444}{{Quantum} {\bf
  5}, 444}~(2021).

\bibitem{Colbeck2021}
Rutvij Bhavsar, Sammy Ragy, and Roger Colbeck.
\newblock ``Improved device-independent randomness expansion rates using two
  sided randomness''~(2023).
\newblock  \href{http://arxiv.org/abs/2103.07504}{arXiv:2103.07504}.

\bibitem{Brown-numerical-vNentropy-bounds}
Peter Brown, Hamza Fawzi, and Omar Fawzi.
\newblock ``Computing conditional entropies for quantum correlations''.
\newblock \href{https://dx.doi.org/10.1038/s41467-020-20018-1}{Nature
  Communications {\bf 12}, 575}~(2021).

\bibitem{Tan-numerical-vNentropy-bounds}
Ernest Y.-Z. Tan, Ren{\'e} Schwonnek, Koon~Tong Goh, Ignatius~William
  Primaatmaja, and Charles C.-W. Lim.
\newblock ``Computing secure key rates for quantum cryptography with untrusted
  devices''.
\newblock \href{https://dx.doi.org/10.1038/s41534-021-00494-z}{npj Quantum
  Information {\bf 7}, 158}~(2021).

\bibitem{Renner-SDP}
Ernest Y.-Z. Tan, Pavel Sekatski, Jean-Daniel Bancal, Ren{\'{e} } Schwonnek,
  Renato Renner, Nicolas Sangouard, and Charles C.-W. Lim.
\newblock ``Improved {DIQKD} protocols with finite-size analysis''.
\newblock \href{https://dx.doi.org/10.22331/q-2022-12-22-880}{Quantum {\bf 6},
  880}~(2022).

\bibitem{Mermin}
N.~David Mermin.
\newblock ``Extreme quantum entanglement in a superposition of macroscopically
  distinct states''.
\newblock \href{https://dx.doi.org/10.1103/PhysRevLett.65.1838}{Phys. Rev.
  Lett. {\bf 65}, 1838--1840}~(1990).

\bibitem{Ardehali}
M.~Ardehali.
\newblock ``Bell inequalities with a magnitude of violation that grows
  exponentially with the number of particles''.
\newblock \href{https://dx.doi.org/10.1103/PhysRevA.46.5375}{Phys. Rev. A {\bf
  46}, 5375--5378}~(1992).

\bibitem{BK93}
A.~V. Belinskiĭ and D.~N. Klyshko.
\newblock ``Interference of light and bell's theorem''.
\newblock \href{https://dx.doi.org/10.1070/PU1993v036n08ABEH002299}{Phys. Rev.
  A {\bf 36}, 653--693}~(1993).

\bibitem{AsymCHSH}
Antonio Ac\'{\i}n, Serge Massar, and Stefano Pironio.
\newblock ``Randomness versus nonlocality and entanglement''.
\newblock \href{https://dx.doi.org/10.1103/PhysRevLett.108.100402}{Phys. Rev.
  Lett. {\bf 108}, 100402}~(2012).

\bibitem{JeremyMABK}
J\'er\'emy Ribeiro, Gl\'aucia Murta, and Stephanie Wehner.
\newblock ``Fully device-independent conference key agreement''.
\newblock \href{https://dx.doi.org/10.1103/PhysRevA.97.022307}{Phys. Rev. A
  {\bf 97}, 022307}~(2018).

\bibitem{Grasselli-PRXQuantum}
Federico Grasselli, Gl\'aucia Murta, Hermann Kampermann, and Dagmar Bru\ss{}.
\newblock ``Entropy bounds for multiparty device-independent cryptography''.
\newblock \href{https://dx.doi.org/10.1103/PRXQuantum.2.010308}{PRX Quantum
  {\bf 2}, 010308}~(2021).

\bibitem{Masanes06}
Llu\'{\i}s Masanes.
\newblock ``Asymptotic violation of bell inequalities and distillability''.
\newblock \href{https://dx.doi.org/10.1103/PhysRevLett.97.050503}{Phys. Rev.
  Lett. {\bf 97}, 050503}~(2006).

\bibitem{uncertrel2010}
Mario Berta, Matthias Christandl, Roger Colbeck, Joseph~M. Renes, and Renato
  Renner.
\newblock ``The uncertainty principle in the presence of quantum memory''.
\newblock \href{https://dx.doi.org/10.1038/nphys1734}{Nature Physics {\bf 6},
  659--662}~(2010).

\bibitem{HolzComment}
Timo Holz, Daniel Miller, Hermann Kampermann, and Dagmar Bru\ss.
\newblock ``Comment on ``fully device-independent conference key agreement''''.
\newblock \href{https://dx.doi.org/10.1103/PhysRevA.100.026301}{Phys. Rev. A
  {\bf 100}, 026301}~(2019).

\bibitem{CKAbook}
Federico Grasselli.
\newblock ``Quantum cryptography''.
\newblock \href{https://dx.doi.org/10.1007/978-3-030-64360-7}{Springer
  International Publishing}. ~(2021).

\bibitem{Murta_2019}
G~Murta, S~B van Dam, J~Ribeiro, R~Hanson, and S~Wehner.
\newblock ``Towards a realization of device-independent quantum key
  distribution''.
\newblock \href{https://dx.doi.org/10.1088/2058-9565/ab2819}{Quantum Science
  and Technology {\bf 4}, 035011}~(2019).

\bibitem{Memory_attack}
Jonathan Barrett, Roger Colbeck, and Adrian Kent.
\newblock ``Memory attacks on device-independent quantum cryptography''.
\newblock \href{https://dx.doi.org/10.1103/PhysRevLett.110.010503}{Phys. Rev.
  Lett. {\bf 110}, 010503}~(2013).

\bibitem{EAT}
F.~{Dupuis} and O.~{Fawzi}.
\newblock ``Entropy accumulation with improved second-order term''.
\newblock \href{https://dx.doi.org/10.1109/TIT.2019.2929564}{IEEE Transactions
  on Information Theory {\bf 65}, 7596--7612}~(2019).

\bibitem{experimentalCKA-graphstates}
Alexander Pickston, Joseph Ho, Andrés Ulibarrena, Federico Grasselli,
  Massimiliano Proietti, Christopher~L. Morrison, Peter Barrow, Francesco
  Graffitti, and Alessandro Fedrizzi.
\newblock ``Experimental network advantage for quantum conference key
  agreement''~(2022).
\newblock  \href{http://arxiv.org/abs/2207.01643}{arXiv:2207.01643}.

\bibitem{Epping}
Michael Epping, Hermann Kampermann, Chiara Macchiavello, and Dagmar Bru{\ss}.
\newblock ``Multi-partite entanglement can speed up quantum key distribution in
  networks''.
\newblock \href{https://dx.doi.org/10.1088/1367-2630/aa8487}{New Journal of
  Physics {\bf 19}, 093012}~(2017).

\bibitem{Carrara2020}
Giacomo Carrara, Hermann Kampermann, Dagmar Bru{\ss}, and Gl{\'{a} }ucia Murta.
\newblock ``Genuine multipartite entanglement is not a precondition for secure
  conference key agreement''.
\newblock \href{https://dx.doi.org/10.1103/physrevresearch.3.013264}{Physical
  Review Research{\bf 3}}~(2021).

\bibitem{nielsen_chuang_2010}
Michael~A. Nielsen and Isaac~L. Chuang.
\newblock ``Quantum computation and quantum information: 10th anniversary
  edition''.
\newblock \href{https://dx.doi.org/10.1017/CBO9780511976667}{Cambridge
  University Press}. ~(2010).

\bibitem{Tendick2021}
Lucas Tendick, Hermann Kampermann, and Dagmar Bru{\ss}.
\newblock ``Quantifying necessary quantum resources for nonlocality''.
\newblock \href{https://dx.doi.org/10.1103/physrevresearch.4.l012002}{Physical
  Review Research{\bf 4}}~(2022).

\bibitem{Mathematica}
Wolfram~Research{,} Inc.
\newblock ``Mathematica, {V}ersion 10.3''~(2016).

\end{thebibliography}

\onecolumn\newpage

\appendix
\numberwithin{equation}{section}
\section{Summary of optimal strategies and entropy bounds}\label{app:bounds-and-optimal-measurements}

In this appendix we report the quantum strategies that lead to maximal violation of the Bell inequalities considered in the manuscript, as well as a summary of the one-outcome and two-outcome entropy bounds used to benchmark DICKA and DIRE protocols.

\paragraph{Holz inequality \cite{Holz2019DICKA}} For three parties, the inequality reads:
 \begin{equation}
     \beta_{\mathrm{H}} := \left<A_1B_{+}C_{+}\right>-\left<A_0B_{-}\right> -\left<A_0C_{-}\right> -\left<B_{-}C_{-}\right>\leq 1, 
\end{equation}
and has quantum bound $\beta^Q_{\mathrm{H}}=3/2$. The quantum bound is attained when the parties share a GHZ state and choose the following optimal measurements:
\begin{align}
    A_0=\sigma_z \quad&;\quad A_1=\sigma_x \nonumber\\
    B_{+}=C_{+}= \frac{\sqrt{3}}{2}\sigma_x \quad&;\quad B_{-}=C_{-}=-\frac{1}{2}\sigma_z \label{opt-measure-timo}.
\end{align}
The tight lower bound on the conditional entropy of Alice's outcome $A_0$, certified by a violation of the Holz inequality, reads (Theorem~\ref{thm:Hbound}): 
\begin{equation}
  H(A_0|E) \geq  1-h\left[\frac{1}{4}\left( \beta_{\mathrm{H}} + 1+ \sqrt{\beta_{\mathrm{H}}^2 + 2\beta_{\mathrm{H}} -3}\right)\right] \label{Hbound-app}. 
\end{equation}
The tight lower bound on the conditional entropy of Alice's and Bob's outcomes $A_0$ and $B_0$ is conjectured to be (Conjecture~\ref{conj:H2bound}):
\begin{align}
H(A_0 B_0 |E) \geq
&\left\lbrace\begin{array}{ll}
    \eta(\beta_{\mathrm{H}})     &\beta_{\mathrm{H}}\in[1,\sqrt{2}]  \\[2ex]
     \displaystyle\frac{\theta(\beta^*_{\mathrm{H}},x(\beta^*_{\mathrm{H}}))-1}{\beta^*_{\mathrm{H}}-\sqrt{2}}(\beta_{\mathrm{H}}-\sqrt{2}) +1      &\beta_{\mathrm{H}}\in(\sqrt{2},\beta^*_{\mathrm{H}}] \\[3ex]
    \theta(\beta_{\mathrm{H}},x(\beta_{\mathrm{H}})) &\beta_{\mathrm{H}}\in(\beta^*_{\mathrm{H}},3/2]
    \end{array}\right. 
    \label{ABentropybound-holz-app}
\end{align}
where the functions $\eta$ and $\theta$ are defined as:
\begin{align}
    \eta(\beta_{\mathrm{H}})= 2-H\left(\left\lbrace\frac{1}{4}\left(1+\sqrt{\beta^2_{\mathrm{H}}-1}\right),\frac{1}{4}\left(1+\sqrt{\beta^2_{\mathrm{H}}-1}\right),\frac{1}{4}\left(1-\sqrt{\beta^2_{\mathrm{H}}-1}\right),\frac{1}{4}\left(1-\sqrt{\beta^2_{\mathrm{H}}-1}\right)\right\rbrace\right) \label{eta}
\end{align}
and
\begin{align}
    \theta(\beta_{\mathrm{H}},x)&=H\left(\left\lbrace\frac{\beta_{\mathrm{H}}(2-\beta_{\mathrm{H}})-x^2}{8(\beta_{\mathrm{H}}-1)},\frac{\beta_{\mathrm{H}}(2-\beta_{\mathrm{H}})-x^2}{8(\beta_{\mathrm{H}}-1)},\frac{(\beta_{\mathrm{H}}-1)(\beta_{\mathrm{H}}+3)+x^2-1}{8(\beta_{\mathrm{H}}-1)},\frac{(\beta_{\mathrm{H}}-1)(\beta_{\mathrm{H}}+3)+x^2-1}{8(\beta_{\mathrm{H}}-1)}\right\rbrace\right) \nonumber\\
    &-h\left(\frac{2(1-x)-(\beta_{\mathrm{H}}-x)^2}{4x(\beta_{\mathrm{H}}-1)}\right). \label{theta}
\end{align}
The function $x(\beta_{\mathrm{H}})$ returns the real solution of the following transcendental equation in $x$:
\begin{align}
    &\left(\beta_{\mathrm{H}}^2-x^2-2\right) \log \left(-\beta_{\mathrm{H}}^2-2 \beta_{\mathrm{H}} x-x^2+2 x+2\right)+\left(x^2+2\right) \log \left(\beta_{\mathrm{H}}^2-2 \beta_{\mathrm{H}} x+x^2+2 x-2\right) \nonumber\\
    &+2x^3 \log \left(\beta_{\mathrm{H}}^2+2 \beta_{\mathrm{H}}+x^2-4\right)-\beta_{\mathrm{H}}^2 \log \left(\beta_{\mathrm{H}}^2-2 \beta_{\mathrm{H}} x+x^2+2x-2\right)-2 x^3 \log \left(-\beta_{\mathrm{H}}^2+2\beta_{\mathrm{H}}-x^2\right)=0 \label{x-sol},
\end{align}
which is obtained by setting $\partial\theta(\beta_{\mathrm{H}},x)/\partial x=0$. Finally, the violation $\beta^*_{\mathrm{H}}$ is approximately given by $\beta^*_{\mathrm{H}}\approx 1.49$ and is implicitly defined by the following equation:
\begin{align}
    \frac{d \theta(\beta^*_{\mathrm{H}},x(\beta^*_{\mathrm{H}}))}{d \beta^*_{\mathrm{H}}} (\beta^*_{\mathrm{H}}-\sqrt{2})=\theta(\beta^*_{\mathrm{H}},x(\beta^*_{\mathrm{H}}))-1 \label{betaprime}.
\end{align}

\paragraph{Parity-CHSH inequality \cite{JeremyParityCHSH}} The inequality reads, after renormalization, as follows:
\begin{equation}
    \beta_{\mathrm{pC}} = \left<A_1B_{-}C\right>+\left<A_0B_{+}\right> \leq 1 \label{parity-chsh-ineq-app}
\end{equation}
where $B_{\pm};=(B_0 \pm B_1)/2$ and $\beta^Q_{\mathrm{pC}}=\sqrt{2}$ is the quantum bound, which is attained when the parties share a GHZ state and choose the following optimal measurements:
\begin{align}
    A_0=\sigma_z \quad&;\quad A_1=\sigma_x \quad;\quad C=\sigma_x\nonumber\\
    B_{+}= \frac{1}{\sqrt{2}}\sigma_z \quad&;\quad B_{-}=\frac{1}{\sqrt{2}}\sigma_x. \label{opt-measure-parity}
\end{align}
A lower bound on the entropy of Alice's outcome $A_0$ certified by the Parity-CHSH inequality is given by:
\begin{equation}
    H(A_0|E)\geq 1- h\de{\frac{1}{2}+\frac{1}{2}\sqrt{\de{\beta_{\mathrm{pC}}}^2-1}} \label{entropybound-parity-app}.
\end{equation}
The above bound is derived in \cite{JeremyParityCHSH}, however it was not proved to be tight. We show its tightness in Appendix~\ref{app:tightness-parityCHSH-bound}.

A numerical lower bound on the two-outcome entropy $H(A_0 B_0 |E)$ is obtained in this work. For details, see Appendix~\ref{app:two-outcome-numerical}.

\paragraph{MABK inequality \cite{Mermin,Ardehali,BK93}} In the case of three parties the inequality reads:
\begin{align}
    \BM=\braket{A_0 B_0 C_1} + \braket{A_0 B_1 C_0} + \braket{A_1 B_0 C_0} - \braket{A_1 B_1 C_1} \leq 2 \label{mabk-ineq-app}.
\end{align}
The quantum bound $\BM^Q=4$ is achieved by the following optimal measurements on the GHZ state:
\begin{align}
    A_0=B_0=\sigma_y \quad&;\quad A_1=B_1=\sigma_x \nonumber\\
    C_0= -\sigma_y \quad&;\quad C_1=-\sigma_x \label{opt-measure-mabk}.
\end{align}
A lower bound on the entropy of Alice's outcome is given by \cite{JeremyMABK,Grasselli-PRXQuantum}:
\begin{equation}
    H(A_0|E) \geq 1-h\left(\frac{1}{2} + \frac{1}{2}\sqrt{\frac{\BM^2}{8}-1}\right). \label{entropybound-mabk-app}
\end{equation}
We provide an alternative proof of this bound in Appendix~\ref{sec:MABK-bound}.

For the two-party entropy $H(A_0 B_0 |E)$, a lower bound is given by \cite{Grasselli-PRXQuantum}:
\begin{align}
    H(A_0 B_0|E) \geq 2-H\left(\left\lbrace 1-3 f(\beta_{\mathrm{M}}),f(\beta_{\mathrm{M}}),f(\beta_{\mathrm{M}}),f(\beta_{\mathrm{M}})\right\rbrace\right) \label{ABentropybound-mabk-app},
\end{align}
where $H(\{p_1,p_2,\dots\})$ is the Shannon entropy of the probability distribution $\{p_1,p_2,\dots\}$ and the function $f$ is defined as:
\begin{equation}
    f(\beta_{\mathrm{M}}) = \frac{1}{4}-\frac{\sqrt{3}}{24}\sqrt{\beta_{\mathrm{M}}^2-4} \label{functionf}.
\end{equation} 
We remark that, according to the numerical calculations in \cite{Grasselli-PRXQuantum}, the bounds in \eqref{entropybound-mabk-app} and \eqref{ABentropybound-mabk-app} are not tight but are close to the corresponding tight lower bound.

\paragraph{Asymmetric CHSH inequalities \cite{AsymCHSH,AsymCHSH-Woodhead}} The family of inequalities is parametrized by $\alpha\in\mathbbm{R}$ and reads:
\begin{align}
    \beta_{\alpha\mathrm{C}}=\alpha\braket{A_0 B_0} + \alpha\braket{A_0 B_1} + \braket{A_1 B_0} - \braket{A_1 B_1} \leq \left\lbrace \begin{array}{ll}
    2 \abs{\alpha}     &\mbox{if } \abs{\alpha}>1 \\
    2   & \mbox{if } \abs{\alpha}\leq 1
    \end{array}\right. \label{alphachsh-ineq-app}
\end{align}
The CHSH inequality is maximally violated and reaches its quantum bound $\beta^Q_{\alpha\mathrm{C}}=2 \sqrt{1+\alpha^2}$ when Alice and Bob share $\ket{\Phi^+}$ and perform the optimal measurements \cite{AsymCHSH-Woodhead}:
\begin{align}
    A_0=\sigma_z \quad&;\quad A_1=\sigma_x \nonumber\\
    B_0=\frac{\alpha}{\sqrt{1+\alpha^2}}\sigma_z + \frac{1}{\sqrt{1+\alpha^2}}\sigma_x \quad&;\quad B_1=\frac{\alpha}{\sqrt{1+\alpha^2}}\sigma_z - \frac{1}{\sqrt{1+\alpha^2}}\sigma_x \label{opt-measure-alphachsh}.
\end{align}
A tight lower bound on the entropy of Alice's outcome was derived in \cite{AsymCHSH-Woodhead}, and reads:
\begin{align}
    H(A_0|E) \geq \left\lbrace \begin{array}{ll}
    g'(\beta^*_{\alpha\mathrm{C}})(\beta_{\alpha\mathrm{C}}-2) &\mbox{if }\abs{\alpha}< 1 \mbox{ and } 2 \leq \beta_{\alpha\mathrm{C}}<\beta^*_{\alpha\mathrm{C}} \\[1ex]
     g(\beta_{\alpha\mathrm{C}})    &\mbox{if }\abs{\alpha} \geq 1 \mbox{ or } \beta_{\alpha\mathrm{C}} \geq \beta^*_{\alpha\mathrm{C}}
    \end{array}\right. \label{entropybound-alphachsh-app},
\end{align}
where the function $g(x)$ is defined as: $g(x):= 1 - h(1/2 + (1/2)\sqrt{x^2/4-\alpha^2})$, $g'(x)$ is its first derivative and $\beta^*_{\alpha\mathrm{C}}$ is the solution of the following equation: $g'(x)(x-2)= g(x)$.

For the two-party entropy $H(A_0B_0|E)$, a numerical bound for the case of $\alpha=1$ (which reduces to the standard CHSH inequality) is obtained in this work (see Appendix~\ref{app:two-outcome-numerical}) and agrees with the numerical bound independently derived in~\cite{Colbeck2021}. Moreover, for the CHSH inequality a tight analytical lower bound on $H(A B |XY E)$ (where $X$ and $Y$ are Alice's and Bob's inputs) as a function of the CHSH value $\beta_C$ is conjectured in \cite{Colbeck2021} and reported here: 
\begin{align}
    H(AB|XYE) \geq \left\lbrace \begin{array}{ll}
    g_1'(\beta^*_{\mathrm{C}})(\beta_{\mathrm{C}}-2) &\mbox{if } 2 \leq \beta_{\mathrm{C}}\leq \beta^*_{\mathrm{C}} \\[1ex]
     g_1(\beta_{\mathrm{C}})    & \mbox{if } \beta^*_{\mathrm{C}} < \beta_{\mathrm{C}} \leq 2 \sqrt{2}
    \end{array}\right. \label{entropybound-colbeck-app},
\end{align}
where $g_1(x)$ is defined as $g_1(x)=1+h(1/2+x/8)-2\,h(1/2+\sqrt{2}x/8)$, $g_1 '(x)$ is its first derivative and $\beta^*_C$ is the solution of $g_1'(x)(x-2)=g_1(x)$ and is approximately given by $\beta^*_C \approx 2.75$.

\section{Proof of one-outcome entropy bound for the Holz inequality} \label{app:Holz-one-outcome-proof}

In this appendix we prove Theorem~\ref{thm:Hbound}, i.e., we prove the following lower bound on the von Neumann entropy of Alice's outcome $A_0$, conditioned on the eavesdropper's total side information $E_{\rm tot}$, when three parties test the Holz inequality:
\begin{equation}
    H(A_0|E_{\mathrm{tot}}) \geq  1-h\left[\frac{1}{4}\left( \beta_{\mathrm{H}} + 1+ \sqrt{\beta_{\mathrm{H}}^2 + 2\beta_{\mathrm{H}} -3}\right)\right] \label{Hbound-proof},
\end{equation}
where $h(x)=-x \log_2 x + (1-x) \log_2 (1-x)$ is the binary entropy. Additionally, we prove that the bound above is tight.

Directly performing an analytical minimization of the conditional entropy over every quantum state (of any dimension) and measurement would be prohibitive. Therefore, the first step to prove \eqref{Hbound-proof} is to simplify the problem at hand.

\subsection{Simplification of the problem} \label{subsec:simplification}

Here we simplify the generic quantum state shared by the parties and the form of the Holz inequality, without losing generality.

Holz's Bell inequality \cite{Holz2019DICKA}, for $N=3$ parties, is given by:
\begin{equation}
     \beta_{\mathrm{H}} := \left<A_1B_{+}C_{+}\right>-\left<A_0B_{-}\right> -\left<A_0C_{-}\right> -\left<B_{-}C_{-}\right>\leq 1, \label{timo-ineq-app}
\end{equation}
where $A_i$, $B_i$ and $C_i$ (for $i=0,1$) are Alice's, Bob's and Charlie's binary observables, respectively, and where we define the unnormalized observables: $B_{\pm}=(B_0 \pm B_1)/2$ and $C_{\pm}=(C_0 \pm C_1)/2$.

To start with, we observe that each party holds two observables with binary outcomes. By following the proof of Theorem~1 in \cite{Grasselli-PRXQuantum}, it is not restrictive to assume that (in every protocol round) Alice, Bob and Charlie share a mixture of three-qubit states and perform rank-one binary projective measurements on their respective qubits. In particular, due to the DI setting, we allow Eve to be in control of the state preparation and to determine the projective measurements performed by the parties on each state of the mixture. We formalize this by saying that, in each round, Eve distributes the following three-qubit mixture:
\begin{equation}
    \rho_{ABC\Xi E'}=\sum_{\alpha} p_\alpha \rho_\alpha \otimes \ketbra{\alpha}{\alpha}_{\xi_A}\otimes \ketbra{\alpha}{\alpha}_{\xi_B}\otimes \ketbra{\alpha}{\alpha}_{\xi_C} \otimes \ketbra{\alpha}{\alpha}_{E'}  \label{3qubit-mixture}
\end{equation}
together with a set of ancillae $\Xi=\{\xi_A,\xi_B,\xi_C\}$  that instruct the parties' devices on the projective measurements to implement on each state $\rho_\alpha$ of the mixture. Recall that Eve knows which state in the mixture gets distributed; this is represented by the classical register $\ket{\alpha}_{E'}$.

The conditional entropy of Alice's outcome $A_0$ given Eve's total information $E_{\mathrm{tot}}=E E'$ can then be expressed as follows\footnote{Note that the same result can also be derived from a generic pure state distributed by Eve \cite{AsymCHSH-Woodhead}. In this case, the classical mixture in \eqref{3qubit-mixture} is the result of the block-diagonal measurement of each party (due to Jordan's lemma \cite{Masanes06}). Thanks to the concavity of the conditional von Neumann entropy one recovers \eqref{entropy-goal1}.}:
\begin{align}
    H(A_0|E_{\mathrm{tot}}) &= \sum_\alpha p_\alpha H(A_0|E E'=\alpha) \nonumber\\
    &= \sum_\alpha p_\alpha H(A_0|E)_{\rho_\alpha}, \label{entropy-goal1}
\end{align}
and we aim at deriving a lower bound on $H(A_0|E_{\mathrm{tot}})$ as a function of the Bell value $\BH$ yielded by \eqref{3qubit-mixture}. The latter can be expressed in terms of the Bell values $\beta^\alpha_{\mathrm{H}}$ yielded by each state $\rho_\alpha$ of the mixture:
\begin{equation}
    \BH = \sum_\alpha p_\alpha \beta^\alpha_{\mathrm{H}} \label{Bell-value-decomp}.
\end{equation}

Equations \eqref{entropy-goal1} and \eqref{Bell-value-decomp} allow us to focus on a specific state $\rho_\alpha$ and derive a convex lower bound $F$ on its conditional entropy $H(A_0|E)_{\rho_\alpha}$ in terms of the Bell value $\beta^\alpha_{\mathrm{H}}$:
\begin{equation}
    H(A_0|E)_{\rho_\alpha} \geq F(\beta^\alpha_{\mathrm{H}}) \label{entropy-goal2}.
\end{equation}
Indeed, by combining the above expression with \eqref{entropy-goal1} and \eqref{Bell-value-decomp} and by exploiting the convexity of $F$, we obtain the desired lower bound:
\begin{equation}
    H(A_0|E_{\mathrm{tot}}) \geq F(\BH) \label{entropy-goal3}.
\end{equation}

For the above argument, we now focus on a specific three-qubit state $\rho_{\alpha}$ and derive the convex lower bound \eqref{entropy-goal2}. For ease of notation, in the following we omit the symbol $\alpha$.

\subsubsection{Reduction of the inequality} \label{sec:ineq-reduction}

\begin{Lmm} \label{lmm:simpl-bellvalue}
The Bell value of the Holz inequality \eqref{timo-ineq} can be reduced without loss of generality to the following form, for some angles $a_1$, $b_-$, and $c_-$,
\begin{equation}
     \beta_{\mathrm{H}} = \left(\cos a_1 \braket{ZXX} + \sin a_1 \braket{XXX}\right)\cos b_{-} \cos c_{-} + \sin b_{-} \braket{ZZ\id} + \sin c_{-} \braket{Z\id Z} -\sin b_{-}\sin c_{-}\braket{\id ZZ}, \label{reduced-bellvalue}
\end{equation}
where $\id$ is the identity operator.
\end{Lmm}

\begin{proof}
We identify the plane induced by the two qubit observables of each party to be the $(x,z)$ plane of the Bloch sphere. Then the parties' observables can be parametrized as follows:
\begin{align}
    A_i &= Z \cos a_i + X \sin a_i \label{Ai}\\
    B_i &= Z \cos b_i + X \sin b_i \label{Bi}\\
    C_i &= Z \cos c_i + X \sin c_i \label{Ci}
\end{align}
where $a_i,b_i,c_i \in [0,2\pi]$ and where $X,Y$ and $Z$ are the Pauli operators. By defining the parameters $b_{\pm}=(b_0 \pm b_1)/2$ and $c_{\pm}=(c_0 \pm c_1)/2$, we can recast the observables $B_{\pm}$ and $C_{\pm}$ as follows:
\begin{align}
    B_{+} &= \cos b_{-} (Z\cos b_{+} + X\sin b_{+}) \label{B+}\\
    B_{-} &= -\sin b_{-} (Z\sin b_{+} - X\cos b_{+}) \label{B-}\\
    C_{+} &= \cos c_{-} (Z\cos c_{+} + X\sin c_{+}) \label{C+}\\
    C_{-} &= -\sin c_{-} (Z\sin c_{+} - X\cos c_{+}) \label{C-}.
\end{align}

By employing the rotational degree of freedom of the local reference frame of each party (rotations around the $y$ axis), we can partially fix the observables (without loss of generality) by rotating Alice's, Bob's and Charlie's reference frames such that:
\begin{align}
    a_0 &=0  \label{a0}\\
    b_{+} = c_{+} &= \frac{\pi}{2}. \label{b+c+}
\end{align}
In particular, we have fixed Alice's key generation measurement $A_0$ to be $Z$. With the above non-restrictive conditions, we reduce the Bell value of the Holz inequality to the form given in \eqref{reduced-bellvalue}.
\end{proof}

\subsubsection{Reduction of the quantum state} \label{sec:state-reduction}

After having reduced the generic quantum state shared by Alice, Bob and Charlie in each round to a mixture of three-qubit states \eqref{3qubit-mixture}, here we prove that each state of the mixture, without loss of generality, is diagonal in the GHZ basis except for some real off-diagonal coefficients. The GHZ basis is an orthonormal basis for the Hilbert space of three qubits and is given by $\{\ket{\psi_{i,j,k}}\}^1_{i,j,k=0}$, with:
\begin{equation}
    \ket{\psi_{i,j,k}} =\frac{1}{\sqrt{2}} \left(\ket{0,j,k} + (-1)^i \ket{1,\bar{j},\bar{k}}\right) \label{GHZbasis},
\end{equation}
where the bar over a bit indicates its negation.

\begin{Lmm} \label{lmm:simpl-state}
Without loss of generality, the three-qubit state shared by Alice, Bob and Charlie can be parametrized as follows:
\begin{align}
    \rho= &\sum_{i,j,k=0}^1 \lambda_{ijk} \ketbra{\psi_{i,j,k}}{\psi_{i,j,k}} + \sum_{j,k=0}^1 r_{jk} \left(\ketbra{\psi_{0,j,k}}{\psi_{1,\bar{j},\bar{k}}} + \mathrm{h.c.}\right) \label{rho-GHZdiag},
\end{align}
where $\lambda_{ijk}$ are the diagonal terms and $r_{jk}$ are real off-diagonal coefficients.
\end{Lmm}

In view of later calculations, we provide the eigenvalues and eigenvectors of the state in \eqref{rho-GHZdiag}. The eigenvalues $\{\rho_{ijk}\}$ of $\rho$ are given by:
\begin{align}
    \rho_{ijk} = \frac{1}{2} \left(\lambda_{0jk} + \lambda_{1\bar{j}\bar{k}} + (-1)^i \sqrt{(\lambda_{0jk} - \lambda_{1\bar{j}\bar{k}})^2+4 r_{jk}^2}\right) \label{eigenvalues},
\end{align}
and the corresponding eigenvectors $\{\ket{\rho_{ijk}}\}$ read as follows:
\begin{align}
    \ket{\rho_{0jk}} &= \cos(t_{jk}) \ket{\psi_{0,j,k}} + \sin(t_{jk}) \ket{\psi_{1,\bar{j},\bar{k}}}    \nonumber\\
    \ket{\rho_{1jk}} &= -\sin(t_{jk}) \ket{\psi_{0,j,k}} + \cos(t_{jk}) \ket{\psi_{1,\bar{j},\bar{k}}}  , \label{eigenvectors}
\end{align}
where $t_{jk}$ is defined as:
\begin{equation}
    t_{jk} = \arctan \frac{2 r_{jk}}{\lambda_{0jk} - \lambda_{1\bar{j}\bar{k}} + \sqrt{(\lambda_{0jk} - \lambda_{1\bar{j}\bar{k}})^2+4 r_{jk}^2}} . \label{tjk}
\end{equation}
We remark that the two sets of parameters $\{\lambda_{ijk},r_{jk}\}$ and $\{\rho_{ijk},t_{jk}\}$ can be used interchangeably to completely describe the state $\rho$. The inverse relations of \eqref{eigenvalues} and \eqref{tjk} read as follows:
\begin{align}
    r_{jk} &= \sin(2t_{jk}) (\rho_{0jk}-\rho_{1jk}) \label{rjk}\\
    \lambda_{0jk} &= \cos^2(t_{jk}) \rho_{0jk} + \sin^2(t_{jk}) \rho_{1jk} \label{lambda0jk} \\
    \lambda_{1\bar{j}\bar{k}} &= \cos^2(t_{jk}) \rho_{1jk} + \sin^2(t_{jk}) \rho_{0jk}. \label{lambda1jbarkbar}
\end{align}

\begin{proof}
We start by noticing that the elements of the GHZ basis are eigenstates of the operators whose expectation value appear in the Holz inequality \eqref{reduced-bellvalue}, except for $ZXX$. More specifically, the action of every operator in \eqref{reduced-bellvalue} on the GHZ basis reads:
\begin{align}
    XXX \ket{\psi_{i,j,k}} &= (-1)^i \ket{\psi_{i,j,k}} \\
    ZZ\id \ket{\psi_{i,j,k}} &= (-1)^j \ket{\psi_{i,j,k}} \\
    Z\id Z \ket{\psi_{i,j,k}} &= (-1)^k \ket{\psi_{i,j,k}} \\
    \id ZZ \ket{\psi_{i,j,k}} &= (-1)^{j+k} \ket{\psi_{i,j,k}} \\
    ZXX \ket{\psi_{i,j,k}} &= \ket{\psi_{\bar{i},\bar{j},\bar{k}}} .\\
\end{align}
Let us now consider the most generic three-qubit state and express it in the GHZ basis:
\begin{align}
    \rho= \sum_{\substack{i,j,k=0 \\ i',j',k'=0}}^1 \rho_{(ijk),(i'j'k')} \ketbra{\psi_{i,j,k}}{\psi_{i',j',k'}} \label{3qubit-state}.
\end{align}
For the observation above, the only terms in \eqref{3qubit-state} that matter in the calculation of the Bell value $\BH$ are the diagonal elements $\rho_{(ijk),(ijk)}$ and the coherences $\rho_{(ijk),(\bar{i}\bar{j}\bar{k})}$. Any other term would provide no contribution to the Bell value $\BH$.

Let us denote by $\rho'$ the state with the same matrix elements $\rho_{(ijk),(ijk)}$ and $\rho_{(ijk),(\bar{i}\bar{j}\bar{k})}$ of $\rho$ in the GHZ basis and null elements otherwise. Recall that in the DI scenario Eve is in total control of the quantum channel and can distribute any arbitrary three-qubit state to Alice, Bob and Charlie.\medskip\\

\textbf{Reduction to block-diagonal state}\\
Here we show that we can assume, without loss of generality, that Eve distributes the state $\rho'$ in place of the generic state $\rho$. This is so because, by construction, the Bell value observed by the parties would not change if they are given $\rho'$ instead of $\rho$. Moreover, Eve's uncertainty about Alice's outcome when she measures $Z$ would not increase. More formally we show that:
\begin{equation}
    H(Z|E)_{\rho} \geq H(Z|E)_{\rho'} \label{increased-knowledge}.
\end{equation}

In order to show \eqref{increased-knowledge}, we first derive the quantum map $\mathcal{D}$ that brings any generic state $\rho$ to $\rho'$, i.e., that sets to zero every coherence of $\rho$ in the GHZ basis except for $\rho_{(ijk),(\bar{i}\bar{j}\bar{k})}$, while leaving the diagonal elements untouched. The map $\mathcal{D}$ can be better understood as a composition of two consecutive maps. The first map acts as follows on the generic state $\rho$:
\begin{align}
    \rho \mapsto \frac{1}{2}\rho + \frac{1}{2} \id ZZ \rho \id ZZ \label{map1},
\end{align}
so that any off-diagonal term in \eqref{3qubit-state} with $j+k \neq j'+ k'$ is set to zero, while the other terms are not affected. The second map is given by:
\begin{align}
    \rho \mapsto \frac{1}{2}\rho + \frac{1}{2}YYX \rho YYX \label{map2},
\end{align}
so that every coherence with $i+j \neq i'+ j'$ is set to zero\footnote{Note that $YYX \ket{\psi_{i,j,k}}=-(-1)^{i+j}\ket{\psi_{i,j,k}}$.}. The combined effect of  \eqref{map1} and \eqref{map2} is exactly the desired map $\mathcal{D}$. We can express the action of $\mathcal{D}$ in a more compact form as follows:
\begin{align}
    \mathcal{D}(\rho) &= \frac{1}{4}\left(\rho + \id ZZ \rho\id ZZ + YYX \rho YYX + YXY \rho YXY \right) \label{mapD} \\
    &= \sum_{i,j,k=0}^1 \rho_{(ijk),(ijk)} \ketbra{\psi_{i,j,k}}{\psi_{i,j,k}} + \rho_{(ijk),(\bar{i}\bar{j}\bar{k})} \ketbra{\psi_{i,j,k}}{\psi_{\bar{i},\bar{j},\bar{k}}} \label{Drho},
\end{align}
and we have that $\rho'=\mathcal{D}(\rho)$. We can now prove \eqref{increased-knowledge}.

To start with, we interpret the state $\mathcal{D}(\rho)$ in \eqref{mapD} as Eve preparing one of the four states ($\rho,\,\id ZZ \rho\id ZZ,\,YYX \rho YYX$ and $YXY \rho YXY$) in the mixture according to the value $t$ of a random variable $T$ known to her. We generically indicate each of the four states as $\rho^t$, for different values of $t$. Since we provide Eve with maximum power, we assume that she holds the purification $\ket{\phi_{ABCE}^t}$ of each $\rho^t$. Therefore, the global quantum state Eve produces reads:
\begin{equation}
    \rho_{ABCET} = \frac{1}{4} \sum_t \ketbra{\phi_{ABCE}^t}{\phi_{ABCE}^t}\otimes \ketbra{t}{t}_T, \label{rhoABCET}
\end{equation}
where the classical register $T$ storing the value $t$ is held by Eve. Finally, Eve also holds the purifying system $T'$ of \eqref{rhoABCET}, such that the global state reads:
\begin{align}
    \rho_{ABCETT'} = \frac{1}{2} \sum_t \ket{\phi_{ABCE}^t}\otimes \ket{t}_T \otimes \ket{t}_{T'}, \label{rhoABCETT'}
\end{align}
where the total information available to Eve is $E_{\mathrm{tot}}=ETT'$. Then, by the strong subadditivity property, we can upper bound the conditional entropy computed on $\rho'=\mathcal{D}(\rho)$ as follows:
\begin{equation}
    H(Z|E_{\mathrm{tot}})_{\mathcal{D}(\rho)} \leq H(Z|ET)_{\mathcal{D}(\rho)} \label{subadd},
\end{equation}
where the entropy on the rhs is computed on the state:
\begin{align}
    \rho_{ZET} &= \frac{1}{4}\sum_t \Tr_{BC}\left[(\mathcal{E}_Z \otimes \id_{BCE}) \ketbra{\phi_{ABCE}^t}{\phi_{ABCE}^t}\right] \otimes \ketbra{t}{t}_T \nonumber\\
    &=: \frac{1}{4} \sum_t \rho_{ZE}^t \otimes \ketbra{t}{t}_T, \label{rhoZET}
\end{align}
where $\mathcal{E}_Z$ is the quantum map describing Alice's $Z$ measurement and where we implicitly defined the state $\rho_{ZE}^t$.

Since Alice's $Z$ measurement is a projection on the computational basis states of subsystem $A$, we can recast $\rho_{ZE}^t$ as follows:
\begin{align}
    \rho_{ZE}^t = \sum_{z=0,1} \ketbra{z}{z}_Z \otimes \Tr_{BC}[\braket{z|\phi_{ABCE}^t} \braket{\phi_{ABCE}^t|z}]. \label{rhoZEt}
\end{align}
Let us now fix for concreteness the value of $t$ such that $\ket{\phi_{ABCE}^t}$ is the purification of the state $YYX\rho YYX$. Nevertheless, our conclusions hold for any other state of the mixture \eqref{mapD}. Then we have that:
\begin{equation}
    \ket{\phi_{ABCE}^t} = \sum_\lambda \sqrt{\lambda}\, YYX \ket{\lambda}_{ABC} \otimes \ket{e_\lambda}_E \label{purification}
\end{equation}
when the spectral decomposition of $\rho$ reads $\rho=\sum_\lambda \lambda \ketbra{\lambda}{\lambda}$ and with $\{\ket{e_\lambda}\}$ an orthonormal basis for $E$. By substituting \eqref{purification} into \eqref{rhoZEt} we obtain:
\begin{equation}
    \rho_{ZE}^t = \sum_{z=0,1} \ketbra{z}{z}_Z \otimes \sum_{\lambda,\sigma} \sqrt{\lambda \sigma} \Tr_{BC}[\bra{z}YYX\ketbra{\lambda}{\sigma} YYX\ket{z}] \otimes \ketbra{e_\lambda}{e_\sigma}_E.
\end{equation}
By using the cyclic property of the trace on $BC$ and the fact that $Y\ket{z}=\im(-1)^z \ket{\bar{z}}$ (we indicate the imaginary unit with $\im$), we can simplify the previous expression as follows:
\begin{align}
    \rho_{ZE}^t &= \sum_{z=0,1} \ketbra{z}{z}_Z \otimes \sum_{\lambda,\sigma} \sqrt{\lambda \sigma} \Tr_{BC}[\braket{\bar{z}|\lambda}\braket{\sigma|\bar{z}}] \otimes \ketbra{e_\lambda}{e_\sigma}_E \nonumber\\
    &= \sum_{z=0,1} \ketbra{\bar{z}}{\bar{z}}_Z \otimes \sum_{\lambda,\sigma} \sqrt{\lambda \sigma} \Tr_{BC}[\braket{z|\lambda} \braket{\sigma|z}] \otimes \ketbra{e_\lambda}{e_\sigma}_E \label{rhoZEt2}
\end{align}

Being $\rho_{ZET}$ classical on subsystem $T$ \eqref{rhoZET}, we can compute its conditional entropy as follows:
\begin{equation}
     H(Z|ET)_{\mathcal{D}(\rho)} = \frac{1}{4} \sum_t H(Z|E)_{\rho^t}, \label{HZET}
\end{equation}
i.e., as the average over $t$ of the entropies of \eqref{rhoZEt2}. However, from \eqref{rhoZEt2} we deduce that the states $\rho_{ZE}^t$ are all the same up to a relabeling of the classical register $Z$. Therefore, they lead to the same conditional entropy $H(Z|E)_{\rho^t}=H(Z|E)_{\rho}$ which is just the conditional entropy of the original state $\rho$. By employing this observation in \eqref{HZET} we can write:
\begin{equation}
    H(Z|ET)_{\mathcal{D}(\rho)} = H(Z|E)_{\rho}. \label{HZET2}
\end{equation}

Finally, by combining the last expression with \eqref{subadd}, we obtain \eqref{increased-knowledge}. We thus proved that it is not restrictive to assume that the parties are given a block-diagonal state of the form \eqref{Drho}. In order to continue with the proof, we relabel the non-null matrix elements of the distributed state in terms of real numbers:
\begin{align}
    \rho'= \sum_{i,j,k=0}^1 \lambda_{ijk} \ketbra{\psi_{i,j,k}}{\psi_{i,j,k}} + \sum_{j,k=0}^1 (r_{jk}+\im s_{jk}) \ketbra{\psi_{0,j,k}}{\psi_{1,\bar{j},\bar{k}}} + \mathrm{h.c.}\label{rhoprime}
\end{align}
where $\lambda_{ijk},r_{jk}$ and $s_{jk}$ are real numbers and $\mathrm{h.c.}$ indicates the Hermitian conjugate of the preceding addend.
\medskip\\

\textbf{Reduction to purely-real coherences}\\
Here we show that, without loss of generality, we can assume that $s_{jk}=0$ for every $j$ and $k$ in \eqref{rhoprime}. That is, the state shared by the parties only displays real off-diagonal elements.

We start by computing the Bell value \eqref{reduced-bellvalue} on the state $\rho'$. One obtains:
\begin{align}
    (\BH)_{\rho'} = &\sum_{i,j,k=0}^1 \lambda_{ijk} \left[(-1)^i\sin a_1 \cos b_{-}\cos c_{-} + (-1)^j\sin b_{-}+ (-1)^k\sin c_{-} - (-1)^{j+k}\sin b_{-}\sin c_{-}\right] \nonumber\\
    &+ 2 r_{jk} \cos a_1 \cos b_{-}\cos c_{-} \label{bellvalue-rhoprime}.
\end{align}
We observe that $(\BH)_{\rho'}$ is independent of the imaginary component of the coherences of $\rho'$, therefore it would read the same when computed for the complex conjugate of $\rho'$ with respect to the GHZ basis, namely
\begin{equation}
    (\rho')^* = \sum_{i,j,k=0}^1 \lambda_{ijk} \ketbra{\psi_{i,j,k}}{\psi_{i,j,k}} + \sum_{j,k=0}^1 (r_{jk}-\im s_{jk}) \ketbra{\psi_{0,j,k}}{\psi_{1,\bar{j},\bar{k}}} + \mathrm{h.c.}\label{rhoprime-star}
\end{equation}
Moreover, note that the states $\rho'_{ZE}$ and $(\rho'_{ZE})^*$ --obtained from \eqref{rhoprime} and \eqref{rhoprime-star} after purification, Alice's $Z$ measurement and partial trace over $BC$-- are still the complex conjugate of each other with respect to the orthonormal basis of $E$ used for the purification. Given that $\rho'_{ZE}$ is Hermitian, this implies that $(\rho'_{ZE})^*$ is also the transposed of $\rho' _{ZE}$ with respect to the same basis. Since a matrix and its transpose have the same eigenvalues, their von Neumann entropies must coincide:
\begin{equation}
    H(Z|E)_{\rho'} = H(Z|E)_{(\rho')^*} \label{same-entropy}.
\end{equation}

We conclude that $\rho'$ and $(\rho')^*$ lead to the same Bell value $(\BH)_{\rho'}$ and provide Eve with the same amount of information about Alice's $Z$ outcome. This means that Eve has no preference in preparing $\rho'$ rather than $(\rho')^*$. As a matter of fact, we can assume without loss of generality that Eve prepares a balanced mixture of the two states:
\begin{equation}
    \bar{\rho} := \frac{\rho' + (\rho')^*}{2} \label{rhobar}.
\end{equation}
Indeed, the Bell value $(\BH)_{\bar{\rho}}$ would be unchanged ($(\BH)_{\bar{\rho}}=(\BH)_{\rho'}$) and Eve's uncertainty would not increase:
\begin{equation}
    H(Z|E)_{\rho'} \geq H(Z|E_{\mathrm{tot}})_{\bar{\rho}}  \label{not-more-uncertainty}.
\end{equation}
In order to verify the above inequality, we interpret $\bar{\rho}$ as Eve preparing the purifications $\ket{\varphi_{ABCE}}$ and $\ket{\varphi^*_{ABCE}}$ of $\rho'$ and $(\rho')^*$, respectively, according to the value of a classical random variable $T$ known to her:
\begin{equation}
    \frac{1}{2}\ketbra{\varphi_{ABCE}}{\varphi_{ABCE}}\otimes\ketbra{0}{0}_T + \frac{1}{2}\ketbra{\varphi^*_{ABCE}}{\varphi^*_{ABCE}}\otimes\ketbra{1}{1}_T.
\end{equation}
Moreover, we provide Eve with the purifying system $T'$ of the above quantum state. Thus, similarly to \eqref{rhoABCETT'}, the global quantum state prepared by Eve reads:
\begin{align}
    \frac{1}{\sqrt{2}}\ket{\varphi_{ABCE}}\otimes \ket{0}_T \otimes \ket{0}_{T'} + \frac{1}{\sqrt{2}}\ket{\varphi^*_{ABCE}}\otimes \ket{1}_T \otimes \ket{1}_{T'}
\end{align}
and she holds systems $E_{\mathrm{tot}}=ETT'$.
By the strong subadditivity property and the fact that the states are classical on $T$, we can upper bound the rhs of \eqref{not-more-uncertainty} by:
\begin{align}
    H(Z|E_{\mathrm{tot}})_{\bar{\rho}} \leq H(Z|ET)_{\bar{\rho}} &= \frac{1}{2}H(Z|E)_{\rho'} + \frac{1}{2}H(Z|E)_{(\rho')^*} \nonumber\\
    &= H(Z|E)_{\rho'},
\end{align}
where we used \eqref{same-entropy} in the last equality. This proves \eqref{not-more-uncertainty}. Hence, without loss of generality the three-qubit state shared by Alice, Bob and Charlie is given by \eqref{rhobar}, which is exactly the state given in \eqref{rho-GHZdiag}.
\end{proof}

\subsection{Derivation of the bound}

Having simplified the inequality (Lemma~\ref{lmm:simpl-bellvalue}) and the form of a generic three-qubit state (Lemma~\ref{lmm:simpl-state}) shared by Alice, Bob and Charlie, we are now ready to prove Theorem~\ref{thm:Hbound}.

The bound derivation is based on a recent technique presented in \cite{AsymCHSH-Woodhead}. The main idea is to lower bound the conditional entropy in terms of a certain expectation value appearing in the Holz inequality, via the uncertainty relation for von Neumann entropies \cite{uncertrel2010}. The proof is then completed by relating the chosen expectation value to the whole Bell value of the Holz inequality.

\begin{proof}[Proof of Theorem~1]
In Subsec.~\ref{subsec:simplification} we show that it is not restrictive to assume that the state shared by Alice, Bob and Charlie is a mixture of three-qubit states. We now focus on a single element of the mixture, $\rho_{ABC}$, and on its extension $\rho_{ABCE}$ that accounts for Eve's quantum side information.

Since we fixed $A_0=Z$ as Alice's key generation measurement, we are interested in finding a lower bound on the conditional von Neumann entropy of Alice's key generation outcome given Eve's quantum side information, $H(Z|E)$.  The uncertainty relation in the presence of quantum memories \cite{uncertrel2010} states that:
\begin{equation}
    H(Z|E) \geq 1 - H(X|BC) \label{uncert-rel},
\end{equation}
where $H(X|BC)$ is the entropy of Alice's outcome if she measures $X$ on her qubit, given the quantum side information of Bob and Charlie. By the fact that quantum operations can only increase the conditional entropy when applied to the conditioning system (see e.g. Theorem 11.15 in \cite{nielsen_chuang_2010}), we have that:
\begin{align}
    H(X|BC) \leq H(X|X_B X_C) \leq H(X|X_{BC})  ,\label{dataprocess}
\end{align}
where $X_B$ ($X_C$) represents Bob's (Charlie's) outcome upon measuring in the $X$ basis and $X_{BC}$ is a classical random variable defined as the multiplication of $X_B$ and $X_C$, $X_{BC}=X_B  X_C$. Then, thanks to Fano's inequality, the Shannon entropy on the rhs of \eqref{dataprocess} can be bounded by the binary entropy of the probability that $X$ differs from $X_{BC}$, namely:
\begin{equation}
    H(X|X_{BC}) \leq  h(Q_X), \label{binary-entropy}
\end{equation}
with $Q_X=\Pr[XX_B X_C =-1]=(1-\braket{XXX})/2$ (where $XXX$ is intended as the product of the $X$ outcomes of Alice, Bob and Charlie). By combining \eqref{dataprocess} and \eqref{binary-entropy}, we obtain the following upper bound on $H(X|BC)$:
\begin{align}
    H(X|BC) &\leq h\left(\frac{1-\braket{XXX}}{2}\right) \nonumber\\
    &=h\left(\frac{1-\abs{\braket{XXX}}}{2}\right)  \nonumber\\
    &=h\left(\frac{1+\abs{\braket{XXX}}}{2}\right) \label{H(X|BC)-upp},
\end{align}
where we used the fact that $h(1/2-p/2)$ is symmetric in $p$ in the first equality and that $h(p)=h(1-p)$ in the second equality.

By combining \eqref{uncert-rel} and \eqref{H(X|BC)-upp} we derive the following lower bound on $H(Z|E)$:
\begin{equation}
    H(Z|E) \geq 1 -  h\left(\frac{1+\abs{\braket{XXX}}}{2}\right) \label{uncert-rel2}.
\end{equation}

The rest of the proof focuses on proving the following inequality between the expectation value $\braket{XXX}$ and the Bell value $\beta_{\mathrm{H}}$ of the Holz inequality:
\begin{equation}
    \abs{\braket{XXX}} \geq \frac{\beta_{\mathrm{H}}}{2} - \frac{1}{2} + \frac{1}{2}\sqrt{\beta_{\mathrm{H}}^2 + 2\beta_{\mathrm{H}} -3} \label{to-prove}.
\end{equation}
Indeed, by employing \eqref{to-prove} in \eqref{uncert-rel2} we obtain the desired lower bound on the conditional entropy \eqref{Hbound-proof}.

We implicitly assume throughout the proof that $\beta_{\mathrm{H}} > 1$, otherwise without Bell violation the conditional entropy is trivially bounded by zero. Moreover we assume that every inequality is to be proven for every value of its parameters, unless otherwise stated.

To start with, we recast the inequality to be proven \eqref{to-prove} as follows:
\begin{equation}
    2\abs{\braket{XXX}}+ 1 -\beta_{\mathrm{H}} \geq \sqrt{\beta_{\mathrm{H}}^2 + 2\beta_{\mathrm{H}} -3} ,
\end{equation}
which is true if and only if the following system of inequalities is true:
\begin{subnumcases}{}
  (2\abs{\braket{XXX}}+ 1 -\beta_{\mathrm{H}})^2 \geq \beta_{\mathrm{H}}^2 + 2\beta_{\mathrm{H}} -3 \label{I1} \\
  2\abs{\braket{XXX}}+ 1 -\beta_{\mathrm{H}} \geq 0 \label{I2} 
\end{subnumcases}
First we focus on proving \eqref{I2}. A sufficient condition for proving \eqref{I2} is given by:
\begin{equation}
    \abs{\braket{XXX}}+ \braket{XXX}\sin a_1 \cos b_{-} \cos c_{-} +  1 -\beta_{\mathrm{H}} \geq 0,
\end{equation}
which reads as follows after employing \eqref{reduced-bellvalue}:
\begin{align}
    1+\abs{\braket{XXX}} \geq \cos a_1 \cos b_{-} \cos c_{-} \braket{ZXX} + \sin b_{-} \braket{ZZ\id} + \sin c_{-} \braket{Z\id Z} -\sin b_{-}\sin c_{-}\braket{\id ZZ}. \label{eq1}
\end{align}
The inequality in \eqref{eq1} is implied by another inequality, namely:
\begin{align}
    1 \geq \abs{\cos b_{-} \cos c_{-} \braket{ZXX}} + \sin b_{-} \braket{ZZ\id} + \sin c_{-} \braket{Z\id Z} -\sin b_{-}\sin c_{-}\braket{\id ZZ}. \label{eq1.1}
\end{align}
Thus, it is sufficient that we prove the following inequality:
\begin{align}
    C:=\cos b_{-} \cos c_{-} \braket{ZXX} + \sin b_{-} \braket{ZZ\id} + \sin c_{-} \braket{Z\id Z} -\sin b_{-}\sin c_{-}\braket{\id ZZ} \leq 1 \label{C}
\end{align}
for \textit{every} value of $b_{-}$ and $c_{-}$ in order to show that \eqref{eq1.1}, and hence \eqref{eq1}, holds. Indeed, the first term in \eqref{C} can always be made equal to $\abs{\cos b_{-} \cos c_{-} \braket{ZXX}}$ by replacing $b_{-}$ with $\pi-b_{-}$ if the term is negative. Note that this replacement does not affect the other terms.

To show that \eqref{C} holds, we start by exploiting the inverse relations \eqref{rjk}, \eqref{lambda0jk} and \eqref{lambda1jbarkbar}, we compute the expectation values in \eqref{C} in terms of the parameters $\{\rho_{ijk},t_{jk}\}$ describing the shared state $\rho$. We obtain:
\begin{align}
    \braket{ZXX} &= \sum_{j,k=0}^1 (\rho_{0jk} - \rho_{1jk}) \sin(2t_{jk})  \label{ZXX}\\
    \braket{ZZ \id} &= \sum_{j,k=0}^1 (-1)^j (\rho_{0jk} - \rho_{1jk})\cos(2t_{jk}) \label{ZZI}\\
    \braket{Z\id Z} &= \sum_{j,k=0}^1 (-1)^k (\rho_{0jk} - \rho_{1jk})\cos(2t_{jk}) \label{ZIZ}\\
    \braket{\id ZZ} &= \sum_{j,k=0}^1 (-1)^{j+k} (\rho_{0jk} + \rho_{1jk}). \label{IZZ}
\end{align}
By inserting the expectation values in the lhs of \eqref{C} we get:
\begin{align}
    C = &\sum_{j,k=0}^1 (\rho_{0jk} - \rho_{1jk}) \left[\sin(2t_{jk}) \cos b_{-} \cos c_{-} +(-1)^j\cos(2t_{jk})\sin b_{-}+(-1)^k\cos(2t_{jk})\sin c_{-} \right] \nonumber\\
    &- \sum_{j,k=0}^1 (-1)^{j+k} (\rho_{0jk} + \rho_{1jk})\sin b_{-}\sin c_{-}.
\end{align}
We upper bound $C$ by maximizing it over $t_{jk}$. In doing so we use the fact that $\rho_{0jk}\geq \rho_{1jk}$ and that $A\cos \theta + B \sin \theta \leq \sqrt{A^2 + B^2}$. We thus obtain:
\begin{align}
    C &\leq  \sum_{j,k=0}^1 (\rho_{0jk} - \rho_{1jk}) \sqrt{\cos^2 b_{-}\cos^2 c_{-} + \sin^2 b_{-} + \sin^2 c_{-} + 2(-1)^{j+k} \sin b_{-}\sin c_{-}} \nonumber\\
    &\quad-\sum_{j,k=0}^1 (-1)^{j+k} (\rho_{0jk} + \rho_{1jk})\sin b_{-}\sin c_{-} \nonumber\\
    &=\sum_{j,k=0}^1 (\rho_{0jk} - \rho_{1jk}) \sqrt{\cos^2 c_{-} + \sin^2 b_{-}(1-\cos^2 c_{-}) + \sin^2 c_{-} + 2(-1)^{j+k} \sin b_{-}\sin c_{-}} \nonumber\\
    &\quad-\sum_{j,k=0}^1 (-1)^{j+k} (\rho_{0jk} + \rho_{1jk})\sin b_{-}\sin c_{-} \nonumber\\
    &=\sum_{j,k=0}^1 (\rho_{0jk} - \rho_{1jk}) \sqrt{1 + \sin^2 b_{-}\sin^2 c_{-} + 2(-1)^{j+k} \sin b_{-}\sin c_{-}} -\sum_{j,k=0}^1 (-1)^{j+k} (\rho_{0jk} + \rho_{1jk})\sin b_{-}\sin c_{-} \nonumber\\
    &=\sum_{j,k=0}^1 (\rho_{0jk} - \rho_{1jk}) \left[1 + (-1)^{j+k} \sin b_{-}\sin c_{-}\right] -\sum_{j,k=0}^1 (-1)^{j+k} (\rho_{0jk} + \rho_{1jk})\sin b_{-}\sin c_{-} \nonumber\\
    &= \sum_{j,k=0}^1 (\rho_{0jk} - \rho_{1jk}) + \sum_{j,k=0}^1 (-1)^{j+k}\sin b_{-}\sin c_{-}(\rho_{0jk} - \rho_{1jk} -\rho_{0jk} - \rho_{1jk}) \nonumber\\
    &\leq\sum_{j,k=0}^1 (\rho_{0jk} - \rho_{1jk}) + 2 \abs{\sum_{j,k=0}^1 (-1)^{j+k} \rho_{1jk}} \nonumber\\
    &\leq \sum_{j,k=0}^1 (\rho_{0jk} - \rho_{1jk}) + 2 \sum_{j,k=0}^1 \rho_{1jk}\nonumber\\
    &= \sum_{i,j,k=0}^1 \rho_{ijk} =1,
\end{align}
where we maximized over $b_{-}$ and $c_{-}$ in the second inequality, used the fact that $\rho_{ijk} \geq 0$ in the third inequality and that the eigenvalues $\rho_{ijk}$ of $\rho$ sum to one in the last equality. We proved \eqref{C} and thus proved \eqref{I2}.

We now focus on proving \eqref{I1}. By computing the square in the lhs of \eqref{I1}, we can recast the inequality as follows:
\begin{align}
    \left(1+ \abs{\braket{XXX}}\right) \beta_{\mathrm{H}} \leq 1 + \abs{\braket{XXX}} + \braket{XXX}^2.
\end{align}
We now insert \eqref{reduced-bellvalue} into the above expression and obtain:
\begin{align}
    &\left(1+ \abs{\braket{XXX}}\right)\nonumber\\
    &\left(\cos a_1 \cos b_{-} \cos c_{-} \braket{ZXX} + \sin a_1 \cos b_{-} \cos c_{-}\braket{XXX} + \sin b_{-} \braket{ZZ\id} + \sin c_{-} \braket{Z\id Z} -\sin b_{-}\sin c_{-}\braket{\id ZZ}\right) \nonumber\\
    &\leq 1 + \abs{\braket{XXX}} + \braket{XXX}^2 \label{eq4}
\end{align} 
Now consider that the above inequality must be proven true for every value of $a_1$.  In particular it must hold true for $a_1$ and $2\pi-a_1$, which is equivalent to having an arbitrary sign for the second term in the second bracket. Then, we can equivalently express the fact that \eqref{eq4} must hold for every $a_1$ as the requirement that the following inequality holds for every $a_1$:
\begin{align}
    &\left(1+ \abs{\braket{XXX}}\right)\nonumber\\
    &\left(\cos a_1 \cos b_{-} \cos c_{-} \braket{ZXX} + \sin a_1 \cos b_{-} \cos c_{-}\abs{\braket{XXX}} + \sin b_{-} \braket{ZZ\id} + \sin c_{-} \braket{Z\id Z} -\sin b_{-}\sin c_{-}\braket{\id ZZ}\right) \nonumber\\
    &\leq 1 + \abs{\braket{XXX}} + \braket{XXX}^2 \label{eq5},
\end{align} 
where we replaced $\braket{XXX}$ with $\abs{\braket{XXX}}$ in the second term of the second bracket. By rearranging the terms in \eqref{eq5} we obtain:
\begin{align}
    &\abs{\braket{XXX}}^2 (1- \sin a_1 \cos b_{-} \cos c_{-}) + \abs{\braket{XXX}}  \nonumber\\
    &(1- \sin a_1 \cos b_{-} \cos c_{-}-\cos a_1 \cos b_{-} \cos c_{-} \braket{ZXX} - \sin b_{-} \braket{ZZ\id} - \sin c_{-} \braket{Z\id Z} +\sin b_{-}\sin c_{-}\braket{\id ZZ}) \nonumber\\
    &+ 1 -\cos a_1 \cos b_{-} \cos c_{-} \braket{ZXX} - \sin b_{-} \braket{ZZ\id} - \sin c_{-} \braket{Z\id Z} +\sin b_{-}\sin c_{-}\braket{\id ZZ} \geq 0. \label{eq6}
\end{align}
A sufficient condition for \eqref{eq6} to be true is when the second degree equation, obtained by replacing $\abs{\braket{XXX}}$ with a generic variable $x$ and taking the equals sign in \eqref{eq6}, has no solution or only one solution in $\mathbbm{R}$. Indeed, in that case the parabola defined by the lhs of \eqref{eq6} never intersects the $x$ axis and always sits above zero (except at most in one point), thus proving the inequality\footnote{We can conclude this since the parabola described by the lhs of \eqref{eq6} is concave upward, which we deduce from the positivity of the coefficient of $\braket{XXX}^2$. In the special case where $\sin a_1 \cos b_{-} \cos c_{-}=1$, the inequality \eqref{eq6} is trivially satisfied.}. Therefore, a sufficient condition for \eqref{eq6} to be true is having the discriminant of the second degree equation smaller or equal to zero:
\begin{align}
    &(1- \sin a_1 \cos b_{-} \cos c_{-}-\cos a_1 \cos b_{-} \cos c_{-} \braket{ZXX} - \sin b_{-} \braket{ZZ\id} - \sin c_{-} \braket{Z\id Z} +\sin b_{-}\sin c_{-}\braket{\id ZZ})^2 \nonumber\\
    &-4(1- \sin a_1 \cos b_{-} \cos c_{-})( 1 -\cos a_1 \cos b_{-} \cos c_{-} \braket{ZXX} - \sin b_{-} \braket{ZZ\id} - \sin c_{-} \braket{Z\id Z} +\sin b_{-}\sin c_{-}\braket{\id ZZ}) \leq 0 ,
\end{align}
which can be rewritten as
\begin{align}
    &(1- \sin a_1 \cos b_{-} \cos c_{-}+\cos a_1 \cos b_{-} \cos c_{-} \braket{ZXX} + \sin b_{-} \braket{ZZ\id} + \sin c_{-} \braket{Z\id Z} -\sin b_{-}\sin c_{-}\braket{\id ZZ})^2 \nonumber\\
    &\leq 4(1- \sin a_1 \cos b_{-} \cos c_{-}) ,
\end{align}
and hence as
\begin{align}
    &\abs{1- \sin a_1 \cos b_{-} \cos c_{-}+\cos a_1 \cos b_{-} \cos c_{-} \braket{ZXX} + \sin b_{-} \braket{ZZ\id} + \sin c_{-} \braket{Z\id Z} -\sin b_{-}\sin c_{-}\braket{\id ZZ}} \nonumber\\
    &\leq 2\sqrt{1- \sin a_1 \cos b_{-} \cos c_{-}} \label{eq7}.
\end{align}
We now remove the absolute value by splitting the previous inequality into an equivalent system of inequalities:
\begin{subnumcases}{}
  1- \sin a_1 \cos b_{-} \cos c_{-}+\cos a_1 \cos b_{-} \cos c_{-} \braket{ZXX} + \sin b_{-} \braket{ZZ\id} + \sin c_{-} \braket{Z\id Z} -\sin b_{-}\sin c_{-}\braket{\id ZZ} \nonumber \\
  \leq 2\sqrt{1- \sin a_1 \cos b_{-} \cos c_{-}} \\
  -(1- \sin a_1 \cos b_{-} \cos c_{-}+\cos a_1 \cos b_{-} \cos c_{-} \braket{ZXX} + \sin b_{-} \braket{ZZ\id} + \sin c_{-} \braket{Z\id Z} -\sin b_{-}\sin c_{-}\braket{\id ZZ}) \nonumber\\
  \leq 2\sqrt{1- \sin a_1 \cos b_{-} \cos c_{-}}
\end{subnumcases}
which can be rearranged as follows:
\begin{subnumcases}{}
  \cos a_1 \cos b_{-} \cos c_{-} \braket{ZXX} + \sin b_{-} \braket{ZZ\id} + \sin c_{-} \braket{Z\id Z} -\sin b_{-}\sin c_{-}\braket{\id ZZ} \nonumber \\
  \leq \sqrt{1- \sin a_1 \cos b_{-} \cos c_{-}}(2-\sqrt{1- \sin a_1 \cos b_{-} \cos c_{-}}) \label{I3}\\
  -\cos a_1 \cos b_{-} \cos c_{-}\braket{ZXX} - \sin b_{-} \braket{ZZ\id} - \sin c_{-} \braket{Z\id Z} +\sin b_{-}\sin c_{-}\braket{\id ZZ} \nonumber \\
  \leq \sqrt{1- \sin a_1 \cos b_{-} \cos c_{-}}(2+\sqrt{1- \sin a_1 \cos b_{-} \cos c_{-}}) \label{I4}.
\end{subnumcases}
We emphasize that once we prove \eqref{I3} and \eqref{I4} we are done, since this is a sufficient condition for the validity of \eqref{I1}.

We first observe that \eqref{I3}, together with a simple inequality to be proved, implies \eqref{I4}. In order to see this, let us label the lhs and rhs of \eqref{I3} as $l$ and $r$, respectively, so that \eqref{I3} can be written as $l \leq r$. Now notice that since \eqref{I4} must be proved for every angle $a_1$, $b_{-}$ and $c_{-}$, we can obtain an equivalent inequality by replacing $a_1 \rightarrow \pi-a_1$, $b_{-} \rightarrow -b_{-}$ and $c_{-} \rightarrow -c_{-}$ and requiring that the new inequality is satisfied for every $a_1$, $b_{-}$ and $c_{-}$. The resulting inequality reads:
\begin{align}
    \cos a_1 \cos b_{-} \cos c_{-}\braket{ZXX} + \sin b_{-} \braket{ZZ\id} + \sin c_{-} \braket{Z\id Z} +\sin b_{-}\sin c_{-}\braket{\id ZZ} \nonumber \\
  \leq \sqrt{1- \sin a_1 \cos b_{-} \cos c_{-}}(2+\sqrt{1- \sin a_1 \cos b_{-} \cos c_{-}}) \label{I4new},
\end{align}
and can be recast in terms of $l$ and $r$ as follows:
\begin{align}
    l +2\sin b_{-}\sin c_{-}\braket{\id ZZ} \leq r + 2(1- \sin a_1 \cos b_{-} \cos c_{-}) \label{I4new2}.
\end{align}
Now assuming that \eqref{I3} holds, \eqref{I4new} --and hence \eqref{I4}-- follows upon proving that the following inequality is true:
\begin{equation}
    \sin b_{-}\sin c_{-}\braket{\id ZZ} \leq 1- \sin a_1 \cos b_{-} \cos c_{-}. \label{eq8}
\end{equation}
The proof of \eqref{eq8} is easily obtained from the following sufficient condition for its validity:
\begin{equation}
    \abs{\sin b_{-}\sin c_{-}} + \abs{\cos b_{-} \cos c_{-}} \leq 1, \label{cos(a-b)}
\end{equation}
which is trivially true for $b_{-},c_{-}\in[0,\pi/2]$ since
\begin{equation}
    1 \geq \cos(b_{-}-c_{-}) = \sin b_{-}\sin c_{-} + \cos b_{-} \cos c_{-} = \abs{\sin b_{-}\sin c_{-}} + \abs{\cos b_{-} \cos c_{-}}.
\end{equation}
Note that for angles outside the interval $[0,\pi/2]$ similar arguments can be made.

We are thus left to prove that \eqref{I3} holds. In order to do so, we again express the expectation values in \eqref{I3} in terms of the parameters describing the shared state $\rho$:
\begin{align}
    &\sum_{j,k=0}^1 (\rho_{0jk} - \rho_{1jk}) \left[\cos a_1 \cos b_{-} \cos c_{-}\sin(2t_{jk})  +(-1)^j\sin b_{-}\cos(2t_{jk})+(-1)^k\sin c_{-}\cos(2t_{jk}) \right] \nonumber\\
    &- \sum_{j,k=0}^1 (-1)^{j+k} (\rho_{0jk} + \rho_{1jk})\sin b_{-}\sin c_{-} \leq \sqrt{1- \sin a_1 \cos b_{-} \cos c_{-}}(2-\sqrt{1- \sin a_1 \cos b_{-} \cos c_{-}}) \label{I3-1}
\end{align}
We find a sufficient condition for \eqref{I3-1} by maximizing the lhs over $t_{jk}$. In doing so, we use the fact that $\rho_{0jk}\geq \rho_{1jk}$ and that $A\cos \theta + B \sin \theta \leq \sqrt{A^2 + B^2}$. We obtain:
\begin{align}
    &\sum_{j,k=0}^1 (\rho_{0jk} - \rho_{1jk})\sqrt{\cos^2 a_1 \cos^2 b_{-} \cos^2 c_{-} + \sin^2 b_{-} + \sin^2 c_{-} + 2(-1)^{j+k}\sin b_{-}\sin c_{-}} \nonumber\\
    &- \sum_{j,k=0}^1 (-1)^{j+k} (\rho_{0jk} + \rho_{1jk})\sin b_{-}\sin c_{-} \leq \sqrt{1- \sin a_1 \cos b_{-} \cos c_{-}}(2-\sqrt{1- \sin a_1 \cos b_{-} \cos c_{-}}) \label{I3-2}.
\end{align}
In turn, a sufficient condition for \eqref{I3-2} is given by: 
\begin{align}
    &\sum_{j,k=0}^1 \tau_{jk}\left[\sqrt{\cos^2 a_1 \cos^2 b_{-} \cos^2 c_{-} + \sin^2 b_{-} + \sin^2 c_{-} + 2(-1)^{j+k}\sin b_{-}\sin c_{-}}-  (-1)^{j+k}\sin b_{-}\sin c_{-}\right] \nonumber\\
    &\leq \sqrt{1- \sin a_1 \cos b_{-} \cos c_{-}}(2-\sqrt{1- \sin a_1 \cos b_{-} \cos c_{-}}) \label{I3-3},
\end{align}
where we defined $\tau_{jk}:=\rho_{0jk} + \rho_{1jk}$. We recast \eqref{I3-3} in the following chain of equivalent inequalities:
\begin{align}
    &\quad\sum_{j,k=0}^1 \tau_{jk}\left[\sqrt{(1-\sin^2 b_{-}) \cos^2 c_{-} + \sin^2 b_{-} + \sin^2 c_{-} + 2(-1)^{j+k}\sin b_{-}\sin c_{-}-\sin^2 a_1 \cos^2 b_{-} \cos^2 c_{-}} \right. \nonumber\\
    &\quad\left.-(-1)^{j+k}\sin b_{-}\sin c_{-}\right] \leq \sqrt{1- \sin a_1 \cos b_{-} \cos c_{-}}(2-\sqrt{1- \sin a_1 \cos b_{-} \cos c_{-}}) \nonumber\\
    &\Leftrightarrow\,\,\sum_{j,k=0}^1 \tau_{jk}\left[\sqrt{1 + \sin^2 b_{-}\sin^2 c_{-} + 2(-1)^{j+k}\sin b_{-}\sin c_{-}-\sin^2 a_1 \cos^2 b_{-} \cos^2 c_{-}}-(-1)^{j+k}\sin b_{-}\sin c_{-}\right] \nonumber\\
    &\quad\leq \sqrt{1- \sin a_1 \cos b_{-} \cos c_{-}}(2-\sqrt{1- \sin a_1 \cos b_{-} \cos c_{-}}) \nonumber\\
    &\Leftrightarrow\,\,\sum_{j,k=0}^1 \tau_{jk}\left[\sqrt{(1 + (-1)^{j+k}\sin b_{-}\sin c_{-})^2-\sin^2 a_1 \cos^2 b_{-} \cos^2 c_{-}} -(-1)^{j+k}\sin b_{-}\sin c_{-}\right] \nonumber\\
    &\quad\leq \sqrt{1- \sin a_1 \cos b_{-} \cos c_{-}}(2-\sqrt{1- \sin a_1 \cos b_{-} \cos c_{-}}) \nonumber\\
    &\Leftrightarrow\,\,(\tau_{00}+\tau_{11})A+(\tau_{01}+\tau_{10})B\leq \sqrt{1- \sin a_1 \cos b_{-} \cos c_{-}}(2-\sqrt{1- \sin a_1 \cos b_{-} \cos c_{-}}) \label{I3-4}
\end{align}
where in the last inequality we defined:
\begin{align}
    A:= \sqrt{(1 + \sin b_{-}\sin c_{-})^2-\sin^2 a_1 \cos^2 b_{-} \cos^2 c_{-}}-\sin b_{-}\sin c_{-} \label{A}\\
    B:=  \sqrt{(1 - \sin b_{-}\sin c_{-})^2-\sin^2 a_1 \cos^2 b_{-} \cos^2 c_{-}}+\sin b_{-}\sin c_{-} \label{B}.
\end{align}
Now, a sufficient condition for \eqref{I3-4} is obtained by replacing $A$ and $B$ by $\max\{A,B\}$. However, since $A$ and $B$ can be mapped to each other under $b_{-} \leftrightarrow -b_{-}$ and since \eqref{I3-4} must hold for every $b_{-}$, we can always assume that $A \geq B$. Thus we replace $B$ with $A$ in \eqref{I3-4} and use the fact that $\sum_{j,k}\tau_{jk}=1$ to obtain the following sufficient condition for \eqref{I3}:
\begin{align}
    \sqrt{(1 + \sin b_{-}\sin c_{-})^2-\sin^2 a_1 \cos^2 b_{-} \cos^2 c_{-}}-\sin b_{-}\sin c_{-} \leq \sqrt{1- \sin a_1 \cos b_{-} \cos c_{-}}(2-\sqrt{1- \sin a_1 \cos b_{-} \cos c_{-}}) \label{I3-5},
\end{align}
which is equivalent to:
\begin{align}
    \sqrt{(1 + \sin b_{-}\sin c_{-})^2-(\sin a_1 \cos b_{-} \cos c_{-})^2} \leq \sin b_{-}\sin c_{-} + 2\sqrt{1- \sin a_1 \cos b_{-} \cos c_{-}}-(1- \sin a_1 \cos b_{-} \cos c_{-}) \label{I3-6}.
\end{align}
By taking the square of both sides in the last inequality, we obtain the equivalent  system of inequalities:
\begin{subnumcases}{}
  (1 + \sin b_{-}\sin c_{-})^2-(\sin a_1 \cos b_{-} \cos c_{-})^2 \nonumber \\
  \leq \left[\sin b_{-}\sin c_{-} + 2\sqrt{1- \sin a_1 \cos b_{-} \cos c_{-}}-(1- \sin a_1 \cos b_{-} \cos c_{-})\right]^2 \label{I5}\\
  \sin b_{-}\sin c_{-} + 2\sqrt{1- \sin a_1 \cos b_{-} \cos c_{-}}-(1- \sin a_1 \cos b_{-} \cos c_{-}) \geq 0 \label{I6}.
\end{subnumcases}

We first focus on proving \eqref{I6}. If $\sin a_1 \cos b_{-} \cos c_{-}=1$, then the inequality \eqref{I6} is trivially true. If $\sin a_1 \cos b_{-} \cos c_{-}\neq 1$, we can recast \eqref{I6} as follows:
\begin{equation}
    2 \geq \sqrt{1- \sin a_1 \cos b_{-} \cos c_{-}} - \frac{\sin b_{-}\sin c_{-}}{\sqrt{1- \sin a_1 \cos b_{-} \cos c_{-}}} \label{I6-1},
\end{equation}
and deduce the following sufficient condition:
\begin{equation}
    \frac{\abs{\sin b_{-}\sin c_{-}}}{\sqrt{1- \sin a_1 \cos b_{-} \cos c_{-}}} \leq 1 \quad \Leftrightarrow \quad \abs{\sin b_{-}\sin c_{-}} \leq \sqrt{1- \sin a_1 \cos b_{-} \cos c_{-}}.
\end{equation}
The validity of the last inequality can be easily proved from \eqref{cos(a-b)} and from the fact that $1- \sin a_1 \cos b_{-} \cos c_{-}\leq \sqrt{1- \sin a_1 \cos b_{-} \cos c_{-}}$. This completes the proof of \eqref{I6}.

We now focus on proving \eqref{I5}. For ease of notation, we define the variables $s$ and $c$ as follows:
\begin{align}
    s:= \sin b_{-}\sin c_{-} \label{s}\\
    c:= \sin a_1 \cos b_{-} \cos c_{-} \label{c}.
\end{align}
Then \eqref{I5} can be recast as follows:
\begin{align}
    &(1+s)^2-c^2 \leq \left[s+ 2\sqrt{1-c} - (1-c) \right]^2 \nonumber\\
    &\Leftrightarrow\,\, 1 + s^2 +2s -c^2 \leq 4(1-c) + (s+c -1)^2 + 4\sqrt{1-c}(s+c-1) \nonumber\\
    &\Leftrightarrow\,\, 2s -c^2 \leq 4(1-c) -2s + c^2 -2c (1-s) + 4\sqrt{1-c}(s+c-1) \nonumber\\
    &\Leftrightarrow\,\, 2c^2 + 4-4c -4s  -2c (1-s) + 4\sqrt{1-c}(s+c-1) \geq 0 \nonumber\\
    &\Leftrightarrow\,\, 1 + (1-c)^2 -2s  -c +cs + 2\sqrt{1-c}(s+c-1) \geq 0 \nonumber\\
    &\Leftrightarrow\,\, (1-c)^2 + (1 -c) -s(1-c) + 2s\sqrt{1-c} - 2\sqrt{1-c}(1-c) -s \geq 0 \nonumber\\
    &\Leftrightarrow\,\, (1-c)^2 - 2\sqrt{1-c}(1-c) + (1 -c)(1-s) + 2s\sqrt{1-c}-s \geq 0 \label{eq9}.
\end{align}
We now view \eqref{eq9} as a fourth degree inequality in the variable $x:=\sqrt{1-c}$, i.e. we rewrite it as follows:
\begin{equation}
    f(x):= x^4 -2x^3 + (1-s) x^2+2sx -s \geq 0. \label{f(x)}
\end{equation}
Then a sufficient condition for the validity of \eqref{eq9} is that $f(x)\geq 0$ for every $x\in[0,\sqrt{2}]$, which is the domain induced by the definition of $x$. Nevertheless, if there are intervals in the domain where $f(x) <0$, inequality \eqref{eq9} can still hold true as far as such intervals are not compatible with the underlying definitions of $s$ and $c$ given in \eqref{s} and \eqref{c}. We will see that this is indeed the case.

In order to study the plot of $f(x)$, we first find its zeroes\footnote{We used Mathematica's ``Reduce'' function to easily find the zeroes of $f(x)$.} for different parametric regions of $s$:
\begin{itemize}
    \item If $-1 \leq s <0$, then $f(x)$ has only one zero in $x_0=1$. Since $f(1/2)=1/16 -s/4>0$ and $f(x)$ is a $C^\infty$ function, we conclude that $f(x) \geq 0$ for $x\in[0,1]$. Similarly, $f(5/4)=241/256-s/16>0$ which implies that $f(x) \geq 0$ for $x\in[1,\sqrt{2}]$. Thus we conclude that $f(x) \geq 0$ in all its domain.
    \item If $s=0$ then $f(x)=x^2(1-x^2)\geq 0$ for every $x$.
    \item If $s=1$ then it follows from \eqref{c} that $c=0$ and thus $x=1$. We have that $f(1)|_{s=1}=0$.
    \item If $0<s<1$, then $f(x)$ has three zeroes in $x_0=-\sqrt{s}$, $x'_0=\sqrt{s}$ and $x''_0=1$. Since $f(\sqrt{s}/2)=-3s(1-\sqrt{s})/4-3s^2/16<0$, we conclude that $f(x) \leq 0$ for $x\in[0,\sqrt{s}]$.
    
    Moreover, by studying the first derivative of $f(x)$ we find the following critical points (where $f'(x)=0$):
    \begin{align}
        x_1 &= \frac{1}{4}\left(1-\sqrt{1+8s}\right)  \label{x1}\\
        x'_1 &= \frac{1}{4}\left(1+\sqrt{1+8s}\right)  \label{x1p}\\
        x''_1 &= 1  \label{x1pp}.
    \end{align}
    We can easily deduce that $x_1<0$ and that $\sqrt{s}<x'_1<1$. By combining this with the fact that $f''(1)=2(1-s)>0$, we conclude that $f(x)\geq 0$ for $x\in[\sqrt{s},\sqrt{2}]$ and that it presents a local minimum in $x=1$.
\end{itemize}

From the above analysis, we deduce that $f(x)<0$ for $x \in [0,\sqrt{s})$ when $0<s<1$. However, as anticipated, the condition $0 \leq x < \sqrt{s}$ is not compatible with the definitions of $s$ and $c$. Indeed, by using the definitions \eqref{s} and \eqref{c} we show that $x \geq \sqrt{s}$ holds:
\begin{align}
    x \geq \sqrt{s} \,\,&\Leftrightarrow\,\, \sqrt{1-\sin a_1 \cos b_{-} \cos c_{-}} \geq \sqrt{\sin b_{-}\sin c_{-}} \nonumber\\
    &\Leftrightarrow\,\, 1-\sin a_1 \cos b_{-} \cos c_{-} \geq \sin b_{-}\sin c_{-}
\end{align}
which can be proved via the sufficient condition \eqref{cos(a-b)}. This implies that $f(\sqrt{1-\sin a_1 \cos b_{-} \cos c_{-}}) \geq 0$ for every $a_1$, $b_{-}$ and $c_{-}$, which proves \eqref{eq9}.

We thus proved \eqref{I5}, which completes the proof of \eqref{I3}, which in turn completes the proof of the validity of \eqref{I1}. This proves the lower bound \eqref{to-prove}, which employed in \eqref{uncert-rel2} provides us with the bound \eqref{entropy-goal2} on the conditional entropy of a fixed state $\rho_\alpha$ in the mixture:
\begin{equation}
    H(A_0|E)_{\rho_\alpha} \geq 1-h\left[\frac{1}{4}\left( \beta^{\alpha}_{\mathrm{H}} + 1+ \sqrt{(\beta^{\alpha}_{\mathrm{H}})^2 + 2\beta^{\alpha}_{\mathrm{H}} -3}\right)\right] \label{entropy-goal2-proved}
\end{equation}

Finally, by the convexity of \eqref{entropy-goal2-proved}, we extend the derived bound to the whole mixed state \eqref{3qubit-mixture} as shown in \eqref{entropy-goal3}, thus obtaining the entropy bound \eqref{Hbound-proof}. This concludes the proof of Theorem~\ref{thm:Hbound}.
\end{proof}

\subsection{Tightness of the bound}

In order to demonstrate that the entropy bound \eqref{Hbound-proof} is tight we need to show that, for every Bell value $\BH$, there exists a quantum state and a set of measurements performed by the parties such that the Bell value is exactly given by $\BH$ and such that the conditional entropy of Alice's outcome $A_0$ is equal to the rhs of \eqref{Hbound-proof}.

The states that satisfy the above conditions (for every $\BH$) belong to the following family of states diagonal in the GHZ basis:
\begin{equation}
    \tau(\nu) = \nu \ketbra{\psi_{0,0,0}}{\psi_{0,0,0}} + (1-\nu) \ketbra{\psi_{1,0,0}}{\psi_{1,0,0}} ,\label{taufam}
\end{equation}
where $\nu\in[1/2,1]$. In order to see this, we first assign to Eve maximum knowledge by letting her hold the purifying system $E$ of $\tau(\nu)$: 
\begin{equation}
    \ket{\Psi_{ABCE}} = \sqrt{\nu} \ket{\psi_{0,0,0}}\otimes\ket{e_0}_E + \sqrt{1-\nu} \ket{\psi_{1,0,0}}\otimes\ket{e_1}_E \label{purification-tau}.
\end{equation}
We now fix Alice's observable $A_0$ to be $Z$ and compute the classical-quantum state of Alice's $Z$ outcome and Eve's quantum system:
\begin{align}
    \tau_{ZE}(\tau) &= \sum_{z=0}^1 \ketbra{z}{z}_Z \otimes\Tr_{BC}[\braket{z|\Psi_{ABCE}}\braket{\Psi_{ABCE}|z}] \nonumber\\
    &= \sum_{z=0}^1 \frac{1}{2} \ketbra{z}{z}_Z \otimes \rho^z_E \label{tauZE1}.
\end{align}
In the above expression, $\rho^z_E$ is the conditional state of Eve given that Alice obtained outcome $Z=z$ and can be easily computed as:
\begin{align}
    \rho^z_E = \nu \ketbra{e_0}{e_0} + (-1)^z\sqrt{\nu(1-\nu)}(\ketbra{e_0}{e_1} + \mathrm{h.c.}) + (1-\nu) \ketbra{e_1}{e_1} \label{rhoEz},
\end{align}
with eigenvalues $\{0,1\}$.
Then, the conditional entropy of $\tau(\nu)$ can be computed in terms of the parameter $\nu$ as follows:
\begin{align}
    H(Z|E)_{\tau(\nu)} &= H(E|Z)_{\tau(\nu)} + H(Z)_{\tau(\nu)} - H(E)_{\tau(\nu)}\nonumber\\
    &=1-h(\nu) \label{Htau}
\end{align}
where $h(x)$ is the binary entropy.

The second step is to derive the maximal Bell value achievable by the state $\tau(\nu)$. We parametrize the parties' observables and partially fix the measurement angles\footnote{The measurement angle $a_0$ of Alice's observable $A_0$ is already fixed by the fact that we chose to study the conditional entropy of Alice's $Z$ outcome when the parties share the state $\tau(\nu)$.} $a_0,b_{+}$ and $c_{+}$ as in Subsec.~\ref{sec:ineq-reduction}. Then, by computing the Bell value \eqref{reduced-bellvalue} for the state $\tau(\nu)$, we get:
\begin{align}
    \beta^{\tau(\nu)}_{\mathrm{H}} =(2\nu -1)\sin a_1 \cos b_{-} \cos c_{-} + \sin b_{-} + \sin c_{-} - \sin b_{-} \sin c_{-}
\end{align}
The above expression can be maximized over the remaining measurement directions $a_1$, $b_{-}$ and $c_{-}$ thus yielding the following maximal Bell value:
\begin{align}
    \beta^{\tau(\nu)}_{\mathrm{H}} = 2\nu + \frac{1}{2\nu} -1 \label{Belltau},
\end{align}
with corresponding optimal angles:
\begin{align}
    a_1=\frac{\pi}{2} \quad,\quad b_{-}=\arctan\frac{1}{\sqrt{4\nu^2 -1}} 
    \quad,\quad c_{-}=\arcsin\frac{1}{2\nu} \label{optimal-angles}.
\end{align}

By reverting \eqref{Belltau} and by inserting the result in \eqref{Htau}, we express the conditional entropy of $\tau(\nu)$ in terms of its achievable Bell value $\beta^{\tau(\nu)}_{\mathrm{H}}$:
\begin{equation}
    H(Z|E)_{\tau(\nu)} = 1-h\left[\frac{1}{4}\left(\beta^{\tau(\nu)}_{\mathrm{H}}+ 1+ \sqrt{(\beta^{\tau(\nu)}_{\mathrm{H}})^2 + 2\beta^{\tau(\nu)}_{\mathrm{H}} -3}\right)\right] \label{Htau-optimal},
\end{equation}
which coincides with the lower bound \eqref{Hbound-proof} on the conditional entropy of Alice's $A_0$ outcome. This proves that the entropy bound in \eqref{Hbound-proof} is tight, since there exists an honest implementation that attains it.

Interestingly, we notice that the optimal measurement angles \eqref{optimal-angles} of Bob and Charlie for the state \eqref{taufam} are given by $b_{-}=c_{-}=\pi/2$ when $\nu\to 1/2$. In other words, the optimal observables of Bob and Charlie that maximize the Bell value tend to be compatible ($B_0=-B_1$ and $C_0=-C_1$) when $\tau(\nu)$ tends to a separable state. This fact has been recently observed in \cite{Tendick2021} for the CHSH inequality \cite{CHSH}, where less incompatible observables  yield higher Bell values while demanding less entanglement from the state.

\section{Proof of one-outcome entropy bound for MABK inequality} \label{app:MABK-one-outcome-proof}
\label{sec:MABK-bound}
The proof technique based on the uncertainty relation, used to derive the entropy bound of Theorem~\ref{thm:Hbound}, can be easily adapted to obtain further entropy bounds.

In this Appendix we rederive the conditional entropy bound on Alice's outcome $A_0$ when Alice, Bob and Charlie test the MABK inequality \cite{Mermin,Ardehali,BK93}. This bound was first obtained in \cite{JeremyMABK} via a correspondence between the MABK inequality and its bipartite counterpart, the CHSH inequality \cite{CHSH}. The bound is also obtained in \cite{Grasselli-PRXQuantum} by direct minimization of the conditional entropy for a fixed violation $\BM$. We report the bound for clarity.

\begin{thm}\label{thm:MABKbound}
Let Alice, Bob and Charlie test the MABK inequality \cite{Mermin,Ardehali,BK93} and obtain a Bell value of $\beta_{\mathrm{M}}$. Then, the von Neumann entropy of Alice's outcome $A_0$ conditioned on Eve's information $E$ satisfies
\begin{equation}
    H(A_0|E) \geq 1-h\left(\frac{1}{2} + \frac{1}{2} \sqrt{\frac{\BM^2}{8}-1}\right) \label{MABKbound},
\end{equation}
where $h(x)=-x \log_2 x + (1-x) \log_2 (1-x)$ is the binary entropy.
\end{thm}
We point out that the bound in \eqref{MABKbound} also holds for Alice's observable $A_1$ due to the symmetry of the MABK inequality.

\begin{proof}
From \cite{Grasselli-PRXQuantum}, we know that we can reduce the state shared by the parties to a three-qubit state and their measurements to rank-one projective measurements on their respective qubits.

We start by deriving an upper bound on the three-party MABK value. In order to do so, we consider its expression as obtained from the recursive definition \cite{Grasselli-PRXQuantum}:
\begin{align}
    \beta_M=\frac{1}{2}\braket{\left[A_0 (B_0 + B_1) + A_1 (B_0 - B_1)\right](C_0 + C_1) + \left[A_1(B_0 + B_1) - A_0 (B_0 - B_1)\right](C_0 - C_1)} \label{MABK-ineq},
\end{align}
and we exploit the degrees of freedom in the choice of the local reference frames to impose that every party's observable lies in the $(x,y)$ plane of the Bloch sphere:
\begin{align}
    A_i &= X \cos a_i + Y \sin a_i \\
    B_i &= X \cos b_i + Y \sin b_i \\
    C_i &= X \cos c_i + Y \sin c_i .
\end{align}
Moreover, we rotate Alice's reference frame along the $z$ axis such that:
\begin{align}
    A_0 &= X \label{A0}\\
    A_1 &= X\cos a  + Y\sin a  \label{A1}.
\end{align}
We define new operators $B_{\pm} := (B_0 \pm B_1)/2$ and $C_{\pm} := (C_0 \pm C_1)/2$ for Bob and Charlie, such that they can be recast as follows:
\begin{align}
    B_{+} &= \cos b_{-} (X\cos b_{+} + Y\sin b_{+}) \label{MABK-B+}\\
    B_{-} &= -\sin b_{-} (X\sin b_{+} - Y\cos b_{+}) \label{MABK-B-}\\
    C_{+} &= \cos c_{-} (X\cos c_{+} + Y\sin c_{+}) \label{MABK-C+}\\
    C_{-} &= -\sin c_{-} (X\sin c_{+} - Y\cos c_{+}) \label{MABK-C-},
\end{align}
where $b_{\pm} := (b_0 \pm b_1)/2$ and $c_{\pm} := (c_0 \pm c_1)/2$. We rotate Bob's and Charlie's reference frames such that $b_{+}=c_{+}=0$. By inserting everything in \eqref{MABK-ineq} we obtain the following simplified MABK value:
\begin{align}
    \beta_M &= 2\cos b_{-}\cos c_{-} \braket{XXX} - 2\sin b_{-}\sin c_{-} \braket{XYY} \nonumber\\
    &+ 2\braket{(X\cos a  + Y\sin a)(\sin b_{-}\cos c_{-} YX +\cos b_{-}\sin c_{-} XY)} \nonumber\\
    &=: 2\, \vec{V} \cdot \vec{W} \label{MABK-sim},
\end{align}
where in the last line we defined the vectors:
\begin{align}
    \vec{V} &= \left(\braket{XXX},\braket{XYY},\braket{XYX},\braket{XXY},\braket{YYX},\braket{YXY}\right) \label{V} \\
    \vec{W} &= \left(\cos b_{-}\cos c_{-},-\sin b_{-}\sin c_{-},\cos a \sin b_{-}\cos c_{-},\cos a \cos b_{-}\sin c_{-},\sin a \sin b_{-}\cos c_{-}, \right.\nonumber\\
    &\left.\sin a \cos b_{-}\sin c_{-}\right) \label{W}.
\end{align}
By the Cauchy-Schwarz inequality and by observing that $\norm{W}=1$, we obtain:
\begin{align}
    \beta_M \leq 2 \sqrt{\braket{YXY}^2 + \braket{YYX}^2 + (\braket{XXX}^2 + \braket{XXY}^2) + (\braket{XYY}^2 + \braket{XYX}^2)} \label{MABK-bound1}.
\end{align}
We now prove that for a generic three-qubit state the following inequalities hold:
\begin{align}
    \braket{XXX}^2 + \braket{XXY}^2 &\leq 1 \label{ineq1}\\
    \braket{XYY}^2 + \braket{XYX}^2 &\leq 1\label{ineq2}.
\end{align}
To show this, we write the generic three-qubit state in the GHZ basis:
\begin{equation}
    \rho = \sum_{i,j,k=0}^1 \lambda_{ijk}\ketbra{\psi_{i,j,k}}{\psi_{i,j,k}} + \sum_{j,k=0}^1 \left(c_{jk} \ketbra{\psi_{0,j,k}}{\psi_{1,j,k}} + \mathrm{h.c.}\right) + \dots \label{3qubit-rho},
\end{equation}
where $\lambda_{ijk}$ are the real diagonal elements, $c_{jk}$ are the complex coherences between the states $\ket{\psi_{0,j,k}}$ and $\ket{\psi_{1,j,k}}$ and $\mathrm{h.c.}$ stands for the Hermitian conjugate of the term preceding it. In the dots ``$\dots$'' we include every other coherence of the state $\rho$, since they do not play a role in the expectation values appearing in \eqref{ineq1} and \eqref{ineq2}. This can be readily seen by considering the action of the operators of \eqref{ineq1} and \eqref{ineq2} on the states of the GHZ basis:
\begin{align}
    XXX \ket{\psi_{i,j,k}} &= (-1)^i \ket{\psi_{i,j,k}} \label{XXX}\\
    XXY \ket{\psi_{i,j,k}} &= -\im (-1)^{i+k} \ket{\psi_{\bar{i},j,k}} \label{XXY}\\
    XYY \ket{\psi_{i,j,k}} &= (-1)^{i+j+k+1} \ket{\psi_{i,j,k}} \label{XYY}\\
    XYX \ket{\psi_{i,j,k}} &= -\im (-1)^{i+j} \ket{\psi_{\bar{i},j,k}} \label{XYX},
\end{align}
which yield the following expectation values on $\rho$:
\begin{align}
    \braket{XXX} &= \sum_{j,k=0}^1 \lambda_{0jk} - \lambda_{1jk} \label{expXXX}\\
    \braket{XXY} &= \sum_{j,k=0}^1 (-1)^k 2\mathrm{Im}(c_{jk}) \label{expXXY}\\
    \braket{XYY} &= -\sum_{j,k=0}^1 (-1)^{j+k}(\lambda_{0jk} - \lambda_{1jk}) \label{expXYY}\\
    \braket{XYX} &= \sum_{j,k=0}^1 (-1)^j 2\mathrm{Im}(c_{jk}) \label{expXYX}.
\end{align}

Before proving \eqref{ineq1} and \eqref{ineq2}, we need to derive a couple of conditions satisfied by the parameters describing the state $\rho$. Since $\rho \geq 0$, then also its restriction to every $2$-dimensional subspace spanned by $\{\ket{\psi_{0,j,k}},\ket{\psi_{1,j,k}}\}$ (for $j,k\in \{0,1\}$) must be positive-semidefinite. Indeed, let $P_{jk}:=\ketbra{\psi_{0,j,k}}{\psi_{0,j,k}} + \ketbra{\psi_{1,j,k}}{\psi_{1,j,k}}$ be the projector on such a subspace. Then,
\begin{align}
    &P_{jk} \rho P_{jk} \geq 0 \quad\iff \quad  \lambda_{0jk}\ketbra{\psi_{0,j,k}}{\psi_{0,j,k}} + \lambda_{1jk}\ketbra{\psi_{1,j,k}}{\psi_{1,j,k}} + c_{jk} \ketbra{\psi_{0,j,k}}{\psi_{1,j,k}} + \mathrm{h.c.} \geq 0 \label{restricted-state},
\end{align}
for all $j$ and $k$. A necessary condition for \eqref{restricted-state} is that its determinant is non-negative (the eigenvalues of positive-semidefinite operators are all non-negative):
\begin{equation}
    \lambda_{0jk}\lambda_{1jk} \geq \abs{c_{jk}}^2 \geq \mathrm{Im}^2(c_{jk}) \label{determinant-condition}.
\end{equation}

A second condition on the state's parameters can be found starting from the following chain of equivalences (note that $\lambda_{ijk}\geq 0$ follows from $\rho \geq 0$):
\begin{align}
    &2 \left(\sqrt{\lambda_{0jk}\lambda_{1j'k'}} -\sqrt{\lambda_{1jk}\lambda_{0j'k'}}\right)^2 \geq 0 \nonumber\\
    &\iff \,\, 2 \left(\lambda_{0jk}\lambda_{1j'k'} + \lambda_{1jk}\lambda_{0j'k'} -2\sqrt{\lambda_{0jk}\lambda_{1jk}\lambda_{0j'k'}\lambda_{1j'k'}} \right) \geq 0 \nonumber\\
    &\iff \lambda_{0jk}\lambda_{1j'k'} + \lambda_{1jk}\lambda_{0j'k'} \geq 4\sqrt{\lambda_{0jk}\lambda_{1jk}\lambda_{0j'k'}\lambda_{1j'k'}} -\lambda_{0jk}\lambda_{1j'k'} - \lambda_{1jk}\lambda_{0j'k'} \label{condition2}.
\end{align}
By using \eqref{determinant-condition} twice we can lower bound the square-root on the rhs as follows:
\begin{align}
    4\sqrt{\lambda_{0jk}\lambda_{1jk}\lambda_{0j'k'}\lambda_{1j'k'}} \geq 4 \abs{\mathrm{Im}(c_{jk})\mathrm{Im}(c_{j'k'})} \geq 4 (-1)^{k+k'}\mathrm{Im}(c_{jk})\mathrm{Im}(c_{j'k'}).
\end{align}
By employing the last expression in \eqref{condition2}, we obtain the second condition needed to prove inequalities \eqref{ineq1} and \eqref{ineq2}: 
\begin{align}
    \lambda_{0jk}\lambda_{1j'k'} + \lambda_{1jk}\lambda_{0j'k'} \geq 4 (-1)^{k+k'}\mathrm{Im}(c_{jk})\mathrm{Im}(c_{j'k'}) -\lambda_{0jk}\lambda_{1j'k'} - \lambda_{1jk}\lambda_{0j'k'} \label{condition2final}.
\end{align}

We now proceed on proving \eqref{ineq1}. The lhs of \eqref{ineq1} can be expressed using \eqref{expXXX} and \eqref{expXXY} as follows:
\begin{align}
    \braket{XXX}^2 + \braket{XXY}^2 &= \left[ \sum_{j,k=0}^1 (\lambda_{0jk} - \lambda_{1jk})\right]^2 + 4 \left[\sum_{j,k=0}^1 (-1)^k \mathrm{Im}(c_{jk})\right]^2 \nonumber\\
    &= \sum_{j,k=0}^1 \left[(\lambda_{0jk} - \lambda_{1jk})^2 + 4\mathrm{Im}^2(c_{jk})\right]  \nonumber\\
    &+ 2\sum_{(j,k)\neq (j',k')}\left[(\lambda_{0jk} - \lambda_{1jk})(\lambda_{0j'k'}  - \lambda_{1j'k'}) + 4 (-1)^{k+k'}\mathrm{Im}(c_{jk})\mathrm{Im}(c_{j'k'})\right] \nonumber\\
    &= \sum_{j,k=0}^1 \left[(\lambda_{0jk} - \lambda_{1jk})^2 + 4\mathrm{Im}^2(c_{jk})\right]  \nonumber\\
    &+ 2\sum_{(j,k)\neq (j',k')}\left[\lambda_{0jk}\lambda_{0j'k'} +\lambda_{1jk}\lambda_{1j'k'} + 4 (-1)^{k+k'}\mathrm{Im}(c_{jk})\mathrm{Im}(c_{j'k'}) -\lambda_{0jk}\lambda_{1j'k'} -\lambda_{1jk}\lambda_{0j'k'} \right].
\end{align}
We now upper bound the above expression by using \eqref{determinant-condition} in the first sum and \eqref{condition2final} in the second sum. We obtain:
\begin{align}
    \braket{XXX}^2 + \braket{XXY}^2 &\leq \sum_{j,k=0}^1 (\lambda_{0jk} + \lambda_{1jk})^2 +2\sum_{(j,k)\neq (j',k')} (\lambda_{0jk}+\lambda_{1jk})(\lambda_{0j'k'}+\lambda_{1j'k'}) \nonumber\\
    &= \left[ \sum_{j,k=0}^1 \lambda_{0jk} + \lambda_{1jk}\right]^2 =1,
\end{align}
where in the last equality we used the fact that $\Tr \rho=1$. This proves \eqref{ineq1}.

Similarly, the lhs of \eqref{ineq2} can be expressed through \eqref{expXYY} and \eqref{expXYX} and upper bounded using \eqref{determinant-condition} and \eqref{condition2final} as follows:
\begin{align}
    \braket{XYY}^2 + \braket{XYX}^2 &= \sum_{j,k=0}^1 \left[(\lambda_{0jk} - \lambda_{1jk})^2 + 4\mathrm{Im}^2(c_{jk})\right]  \nonumber\\
    &+ 2\sum_{(j,k)\neq (j',k')}(-1)^{j+j'+ k+k'}\left[(\lambda_{0jk} - \lambda_{1jk})(\lambda_{0j'k'}  - \lambda_{1j'k'}) + 4 (-1)^{k+k'}\mathrm{Im}(c_{jk})\mathrm{Im}(c_{j'k'})\right] \nonumber\\
    &\leq \sum_{j,k=0}^1 (\lambda_{0jk} + \lambda_{1jk})^2 +2\sum_{(j,k)\neq (j',k')} (-1)^{j+j'+ k+k'}(\lambda_{0jk}+\lambda_{1jk})(\lambda_{0j'k'}+\lambda_{1j'k'}) \nonumber\\
    &=\left[ \sum_{j,k=0}^1 (-1)^{j+k}(\lambda_{0jk} + \lambda_{1jk})\right]^2 \leq 1,
\end{align}
which proves \eqref{ineq2}.

Finally, by employing \eqref{ineq1} and \eqref{ineq2} in \eqref{MABK-bound1} we derive the following upper bound on the MABK value:
\begin{align}
    \beta_M \leq 2 \sqrt{2+ 2\max\{\abs{\braket{YXY}},\abs{\braket{YYX}}\}^2} \label{MABK-bound2}.
\end{align}

Now, we can turn to the conditional entropy of Alice's outcome when she measures $A_0=X$ and lower bound it with the uncertainty relation \cite{uncertrel2010}:
\begin{align}
    H(X|E) \geq 1 - H(Y|BC) \geq 1-h(Q_{Y,O_B O_C}) \label{MABK-entropybound1},
\end{align}
where in the second inequality we followed the same steps that lead to \eqref{dataprocess} and \eqref{binary-entropy}, and where $Q_{Y,O_B O_C}$ is defined as the probability that the $Y$ outcome of Alice differs from the product of the outcomes of Bob ($O_B$) and Charlie ($O_C$), that is:
\begin{equation}
 Q_{Y,O_B  O_C} := \Pr[YO_B O_C=-1] = \frac{1-\braket{Y O_B O_C}}{2}.   
\end{equation}
Moreover, by using the properties of the binary entropy $h(x)$ as in \eqref{H(X|BC)-upp}, we obtain:
\begin{align}
    H(X|E) \geq 1-h\left(\frac{1+\abs{\braket{Y O_B O_C}}}{2}\right) \label{MABK-entropybound2}.
\end{align}

We emphasize that we can employ the uncertainty relation and derive a bound like the one in \eqref{MABK-entropybound2} independently of the measurement settings $A_i$, $B_i$ and $C_i$ of the parties in the DI scenario. Of course, in order to make the derived inequality useful in our case, we set one of Alice's measurements to $X$ --which is also one of Alice's settings in the DI scenario, see \eqref{A0}-- so that we obtain an inequality for the conditional entropy $H(X|E)$. For this argument, we can arbitrarily choose Bob's and Charlie's measurements in \eqref{MABK-entropybound2} to be either $O_B=X$ and $O_C=Y$ or $O_B=Y$ and $O_C=X$. Both cases lead to valid lower bounds on the conditional entropy of Alice's $X$ outcome. We can then lower bound the conditional entropy $H(X|E)$ by:
\begin{align}
    H(X|E) \geq 1-h\left(\frac{1+\max\{\abs{\braket{YXY}},\abs{\braket{YYX}}\}}{2}\right) \label{MABK-entropybound3}.
\end{align}

By reverting the upper bound on the MABK value \eqref{MABK-bound2}, we obtain:
\begin{equation}
    \max\{\abs{\braket{YXY}},\abs{\braket{YYX}}\} \geq \sqrt{\frac{\BM^2}{8}-1} \label{MABK-bound-reverted}.
\end{equation}
Finally, by employing \eqref{MABK-bound-reverted} in \eqref{MABK-entropybound3} we recover the entropy bound \eqref{MABKbound}.

We point out that, although the proof derives a lower bound on the conditional entropy of Alice's $X$-basis outcome, the derived bound is general and holds for any measurement Alice implements. This is because at the beginning of the proof, we purposely set Alice's local reference frame such that her qubit observable $A_0$ coincides with the Pauli operator $X$. This has no effect on the description of the state \eqref{3qubit-rho} shared by Alice, Bob and Charlie since it is a completely generic three-qubit state for any choice of local reference frames.
\end{proof}

\section{Numerical computation of two-outcome entropy bounds} \label{app:two-outcome-numerical}

In this Appendix we describe the steps that allow us to numerically compute lower bounds on the two-outcome entropy $H(A_0 B_0|E)$ for the Holz inequality, the Parity-CHSH inequality and the CHSH inequality. The bounds are plotted in Fig.~\ref{fig:ABentropy-comparison} and are used in Sec.~\ref{sec:DIRE} to compare the performance of DIRE protocols based on different Bell inequalities.

\paragraph{Holz inequality} We compute a numerical lower bound on $H(A_0 B_0|E)$ as a function of the violation $\beta_{\mathrm{H}}$ of the Holz inequality for three parties. The bound is obtained by direct optimization of the entropy once the violation is fixed. Based on the numerical bound, we also conjecture the correspondent analytical expression (Conjecture~\ref{conj:H2bound}).

In order to make the optimization numerically feasible, we arbitrarily fix the local reference frames of Alice, Bob and Charlie and parametrize the state they share as shown in Subsec.~\ref{subsec:simplification}. Then, the measurement angle $b_1$ is fixed by $b_0$ through \eqref{b+c+} as follows:
$b_1=\pi-b_0$, which implies that $b_{-}=b_0-\pi/2$. By substituting in the Bell value \eqref{reduced-bellvalue} of the Holz inequality, we obtain:
\begin{equation}
     v_{\mathrm{H}} = \left(\cos a_1 \braket{ZXX} + \sin a_1 \braket{XXX}\right)\sin b_0 \cos c_{-} - \cos b_0 \braket{ZZ\id} + \sin c_{-} \braket{Z\id Z} +\cos b_0 \sin c_{-}\braket{\id ZZ} \label{reduced-bellvalue2},
\end{equation}
which is now written in terms of the free measurement angles $a_1$ and $c_{-}$ and the angle $b_0$ that instead appears in the conditional entropy expression. Note that, thanks to the parametrization of the state in Subsec.~\ref{sec:state-reduction}, the expectation values in \eqref{reduced-bellvalue2} are easily written in terms of the state parameters $\{\rho_{ijk},t_{jk}\}$ as follows:
\begin{align}
    \braket{XXX} &= \sum_{j,k=0}^1 (\rho_{0jk} - \rho_{1jk}) \cos(2 t_{jk}) \label{<XXX>}\\
    \braket{ZXX} &= \sum_{j,k=0}^1 (\rho_{0jk} - \rho_{1jk}) \sin(2t_{jk})  \label{<ZXX>}\\
    \braket{ZZ \id} &= \sum_{j,k=0}^1 (-1)^j (\rho_{0jk} - \rho_{1jk})\cos(2t_{jk}) \label{<ZZI>}\\
    \braket{Z\id Z} &= \sum_{j,k=0}^1 (-1)^k (\rho_{0jk} - \rho_{1jk})\cos(2t_{jk}) \label{<ZIZ>}\\
    \braket{\id ZZ} &= \sum_{j,k=0}^1 (-1)^{j+k} (\rho_{0jk} + \rho_{1jk}). \label{<IZZ>}
\end{align}
By also expressing the entropy $H(A_0 B_0|E)$ in terms of the state parameters $\{\rho_{ijk},t_{jk}\}$ and the angle $b_0$ (recall that Alice's reference frame is chosen such that $a_0=0$), we numerically solve the following optimization problem:
\begin{align}
    &\min_{\{\rho_{ijk},t_{jk},b_0,a_1,c_{-}\}} H(A_0 B_0|E) (\rho_{ijk},t_{jk},b_0) \nonumber\\
    &\quad\quad\mbox{sub. to}\quad v_{\mathrm{H}}(\rho_{ijk},t_{jk},b_0,a_1,c_{-})=\beta_{\mathrm{H}} \,;\,{\textstyle\sum_{ijk}}\,\rho_{ijk}=1 \,;\,\rho_{ijk}\geq 0, \label{optimization-Holz}
\end{align}
when varying $\beta_{\mathrm{H}}$ in the interval $(1,3/2]$. In reality, we solve the equivalent --in the sense that leads to the same entropy for every $\beta_{\mathrm{H}}$-- but simpler optimization problem: 
\begin{align}
    &\min_{\{\rho_{ijk},t_{jk},b_0\}} H(A_0 B_0|E) (\rho_{ijk},t_{jk},b_0) \nonumber\\
    &\quad\mbox{sub. to}\quad\bar{v}_{\mathrm{H}}(\rho_{ijk},t_{jk},b_0)\geq\beta_{\mathrm{H}} \,;\,{\textstyle\sum_{ijk}}\,\rho_{ijk}=1 \,;\,\rho_{ijk}\geq 0, \label{optimization-Holz2}
\end{align}
where $\bar{v}_{\mathrm{H}}$ is the maximum of \eqref{reduced-bellvalue2} over the free angles $a_1$ and $c_{-}$:
\begin{align}
    \bar{v}_{\mathrm{H}}=\sqrt{\sin^2 b_0 (\braket{ZXX}^2 + \braket{XXX}^2)+\left(\braket{Z\id Z} + \cos b_0 \braket{\id ZZ}\right)^2} -\cos b_0 \braket{ZZ \id} \label{reduced-bellvalue3}.
\end{align}

The numerical plot points obtained by solving \eqref{optimization-Holz2} with the built-in functions of Wolfram Mathematica \cite{Mathematica} are reported in Fig.~\ref{fig:ABentropy-Holz}, together with our conjectured bound on $H(A_0 B_0|E)$. Our conjecture on the analytical expression of the entropy bound is given in Conjecture~\ref{conj:H2bound}.

\paragraph{Parity-CHSH inequality} The bound on $H(A_0 B_0|E)$ when three parties test the Parity-CHSH inequality is also obtained by direct numerical optimization, similarly to the bound for the Holz inequality.

As a matter of fact, note that the Parity-CHSH inequality \eqref{parity-chsh-ineq} is a particular case (upon relabeling the observables) of the Holz inequality \eqref{timo-ineq} for $c_0=c_1$, i.e., when Charlie's two measurements coincide. Thus, the optimization problem yielding the entropy bound for the Parity-CHSH inequality is equal to \eqref{optimization-Holz}, where we set $c_{-}=0$.

\paragraph{CHSH inequality} The bound on $H(A_0 B_0|E)$ when two parties test the CHSH inequality is again obtained by direct numerical optimization. In order to simplify the optimization, we apply the results of \cite{Grasselli-PRXQuantum} to the CHSH scenario and parametrize the state shared by Alice and Bob as a Bell-diagonal state:
\begin{equation}
    \rho=\sum_{i,j=0}^1 \lambda_{ij}\ketbra{\psi_{i,j}}{\psi_{i,j}} \label{CHSHstate},
\end{equation}
where $\ket{\psi_{i,j}}=(\ket{0j}+(-1)^i\ket{1\bar{j}})/\sqrt{2}$ are the states of the Bell basis. We also assume without loss of generality that the parties' observables are rank-one projective measurements in the $(x,y)$-plane, defined by the eigenstates:
\begin{align}
    \ket{a}&=\frac{1}{\sqrt{2}}\left(\ket{0}+(-1)^{a} e^{\im \varphi_{A_k}}\ket{1}\right) \label{keta}\\
    \ket{b}&=\frac{1}{\sqrt{2}}\left(\ket{0}+(-1)^{b} e^{\im \varphi_{B_k}}\ket{1}\right) \label{ketb}
\end{align}
where $a,b\in\{0,1\}$ are the outcomes of Alice's and Bob's observables $A_k$ and $B_k$ and $\varphi_{A_k},\varphi_{B_k}$ are the corresponding measurement directions, respectively. Then, one can compute the joint probability of obtaining outcomes $a$ and $b$ when Alice and Bob measured $A_k$ and $B_l$. We obtain:
\begin{align}
    p(a,b|k,l)=\frac{1}{4}\left[1+(-1)^{a+b}\cos(\varphi_{A_k}+\varphi_{B_l})(\lambda_{00}-\lambda_{10})+(-1)^{a+b}\cos(\varphi_{A_k}-\varphi_{B_l})(\lambda_{01}-\lambda_{11})\right], \label{p(a,b)}
\end{align}
and observe that $p(0,0|k,l)=p(1,1|k,l)$ and $p(0,1|k,l)=p(1,0|k,l)=1/2-p(0,0|k,l)$. We can thus express all the probabilities appearing in the conditional entropy $H(A_0 B_0|E)$ in terms of $p:=p(0,0|0,0)$.

We can now derive a simple expression for the conditional entropy of interest:
\begin{align}
    H(A_0 B_0|E)&= H(A_0 B_0) + H(E|A_0 B_0) - H(E) \nonumber\\
    &= -2p\log_2 p-2(1/2-p)\log_2(1/2-p) -H(\{\lambda_{ij}\}) \nonumber\\
    &= 1+h(2p) -H(\{\lambda_{ij}\}), \label{CHSHentropy}
\end{align}
where in the second equality we used the fact that the state shared by Alice, Bob and Eve is pure (hence $H(E)=H(\rho)$) and the conditional state $\rho^{a,b}_E$ of Eve, given that Alice and Bob obtained outcomes $a$ and $b$, is still pure (thus $H(E|A_0 B_0)=0$).

The CHSH Bell value, for the parametrization described above, reduces to:
\begin{align}
    v_{\mathrm{C}}&=\braket{A_0 B_0} + \braket{A_0 B_1} + \braket{A_1 B_0}-\braket{A_1 B_1} \nonumber\\
    &= (\lambda_{00}-\lambda_{10})(\cos(\varphi_{A_0}+\varphi_{B_0})+\cos(\varphi_{A_0}+\varphi_{B_1})+\cos(\varphi_{A_1}+\varphi_{B_0})-\cos(\varphi_{A_1}+\varphi_{B_1})) \nonumber\\
    &+(\lambda_{01}-\lambda_{11})(\cos(\varphi_{A_0}-\varphi_{B_0})+\cos(\varphi_{A_0}-\varphi_{B_1})+\cos(\varphi_{A_1}-\varphi_{B_0})-\cos(\varphi_{A_1}-\varphi_{B_1})) \label{CHSHvalue}.
\end{align}

We then numerically solved the following optimization problem with the built-in functions of Wolfram Mathematica \cite{Mathematica}:
\begin{align}
    &\min_{\{\lambda_{ij},\varphi_{A_0},\varphi_{B_0},\varphi_{A_1},\varphi_{B_1}\}} 1+h(2p) -H(\{\lambda_{ij}\}) \nonumber\\
    &\quad\quad\mbox{sub. to}\quad v_{\mathrm{C}}(\lambda_{ij},\varphi_{A_0},\varphi_{B_0},\varphi_{A_1},\varphi_{B_1})=\beta_{\mathrm{C}} \,;\,{\textstyle\sum_{ij}}\,\lambda_{ij}=1 \,;\,\lambda_{ij}\geq 0 \label{optimization-CHSH},
\end{align}
where $p$ is given by:
\begin{align}
    p=\frac{1}{4}\left[1+\cos(\varphi_{A_0}+\varphi_{B_0})(\lambda_{00}-\lambda_{10})+\cos(\varphi_{A_0}-\varphi_{B_0})(\lambda_{01}-\lambda_{11})\right]. \label{p}
\end{align}

The numerical solution of \eqref{optimization-CHSH} is the entropy bound reported in Fig.~\ref{fig:ABentropy-comparison}. We remark that the same bound has been independently computed in \cite{Colbeck2021} by combining an analytical simplification similar to the one reported here with numerical techniques.

\section{Tightness of one-outcome entropy bound for Parity-CHSH inequality} \label{app:tightness-parityCHSH-bound}

In this Appendix we demonstrate that the lower bound on the entropy of Alice's outcome $A_0$ when three parties test the Parity-CHSH inequality, reported in \eqref{entropybound-parity-app}, is tight.

\begin{Lmm} \label{lmm:tightness-parityCHSH-bound}
Let Alice, Bob and Charlie test the Parity-CHSH inequality \cite{JeremyParityCHSH} and obtain a Bell value of $\beta_{\mathrm{pC}}$. Then, the following lower bound on the von Neumann entropy of Alice's outcome $A_0$, conditioned on Eve's information $E$,
\begin{equation}
    H(A_0|E)\geq 1- h\de{\frac{1}{2}+\frac{1}{2}\sqrt{\de{\beta_{\mathrm{pC}}}^2-1}}, \label{entropybound-parity-proof}
\end{equation}
is tight. Namely, that there exists a quantum state and a set of measurements yielding a Bell value of $\beta_{\mathrm{pC}}$ with conditional entropy of Alice's outcome $A_0$ given by the rhs of \eqref{entropybound-parity-proof}.
\end{Lmm}

\begin{proof}
Consider the same family of states used to prove the tightness of the bound in \eqref{Hbound}, that is:
\begin{equation}
    \tau(\nu) = \nu \ketbra{\psi_{0,0,0}}{\psi_{0,0,0}} + (1-\nu) \ketbra{\psi_{1,0,0}}{\psi_{1,0,0}} ,\label{taufam-proof}
\end{equation}
where $\nu\in[1/2,1]$. Then, the conditional entropy of Alice's outcome $A_0=Z$ can be computed in terms of the parameter $\nu$ and reads:
\begin{align}
    H(A_0|E)_{\tau(\nu)} =1-h(\nu) . \label{entropy-taufam-parityCHSH}
\end{align}
Now, we compute the maximal violation of the Parity-CHSH inequality \eqref{parity-chsh-ineq} achieved by the state $\tau(\nu)$. To do this, we first parametrize the parties' observables as in Subsec.~\ref{sec:ineq-reduction} and orient the reference frames such that $C=X$, $A_0=Z$ and $b_{+}=0$ (thus $B_{+}=\cos b_0 Z$ and $B_{-}=\sin b_0 X$). Then, the Bell value of the Parity-CHSH inequality reads:
\begin{align}
    \beta_{\rm pC}=\sin b_0 (\cos a_1 \braket{ZXX} + \sin a_1 \braket{XXX}) + \cos b_0 \braket{ZZ \id}.
\end{align}
By computing the expectation values on the state $\tau(\nu)$, we obtain:
\begin{align}
    \beta^{\tau(\nu)}_{\mathrm{pC}} =\sin b_0 \sin a_1 (2\nu -1) + \cos b_0.
\end{align}
The above expression can be maximized over the remaining measurement directions $a_1$ and $b_0$ yielding the following maximal Bell value:
\begin{align}
    \beta^{\tau(\nu)}_{\mathrm{pC}} = \sqrt{(2\nu -1)^2 + 1}.
\end{align}
By reverting the last expression we obtain:
\begin{align}
    \nu = \frac{1}{2} + \frac{1}{2}\sqrt{(\beta^{\tau(\nu)}_{\mathrm{pC}})^2 -1}
\end{align}
which substituted in \eqref{entropy-taufam-parityCHSH} returns exactly the lower bound in \eqref{entropybound-parity-proof}. Hence we proved that the bound is tight.
\end{proof}

\end{document}